\documentclass[12pt]{article}
\usepackage{amssymb}
\usepackage[top=1.20in,bottom=1.20in,left=1.00in,right=1.00in] {geometry}

\usepackage{fancyhdr}
\usepackage{amsmath}
\usepackage{bm}
\usepackage{esint}
\usepackage{amsfonts}
\usepackage{amsthm}
\usepackage{amssymb}
\usepackage{amsbsy}
\usepackage{verbatim}
\usepackage{graphicx}
\usepackage[caption=false,font=footnotesize]{subfig}
\usepackage{url}
\usepackage{mathrsfs}
\usepackage{color}
\usepackage{amstext}
\usepackage{stmaryrd}
\usepackage{enumerate}
\usepackage{algpseudocode}
\usepackage{algorithm}
\usepackage{pifont}
\usepackage{prettyref}

\font\msbm=msbm10

\numberwithin{equation}{section}

\theoremstyle{plain}
\newtheorem{theorem}{Theorem}[section]
\newtheorem{lemma}[theorem]{Lemma}

\newtheorem{condition}[theorem]{Condition}
\newtheorem{remark}[theorem]{Remark}
\def\mathbb#1{\hbox{\msbm{#1}}}

\renewcommand{\epsilon}{\varepsilon}
\newcommand{\ba}{\boldsymbol{a}}
\newcommand{\bb}{\boldsymbol{b}}
\newcommand{\bc}{\boldsymbol{c}}

\newcommand{\be}{\boldsymbol{e}}
\newcommand{\bff}{\boldsymbol{f}}
\newcommand{\bg}{\boldsymbol{g}}
\newcommand{\bh}{\boldsymbol{h}}

\newcommand{\bn}{\boldsymbol{n}}

\newcommand{\bq}{\boldsymbol{q}}

\newcommand{\bu}{\boldsymbol{u}}
\newcommand{\bv}{\boldsymbol{v}}
\newcommand{\bw}{\boldsymbol{w}}
\newcommand{\bx}{\boldsymbol{x}}
\newcommand{\by}{\boldsymbol{y}}
\newcommand{\bz}{\boldsymbol{z}}

\newcommand{\BA}{\boldsymbol{A}}
\newcommand{\BB}{\boldsymbol{B}}
\newcommand{\BC}{\boldsymbol{C}}
\newcommand{\BD}{\boldsymbol{D}}

\newcommand{\BF}{\boldsymbol{F}}
\newcommand{\BI}{\boldsymbol{I}}

\newcommand{\BH}{\boldsymbol{H}}

\newcommand{\BS}{\boldsymbol{S}}

\newcommand{\BU}{\boldsymbol{U}}
\newcommand{\BV}{\boldsymbol{V}}

\newcommand{\BX}{\boldsymbol{X}}

\newcommand{\BZ}{\boldsymbol{Z}}

\newcommand{\tF}{\widetilde{F}}

\newcommand{\hbh}{\hat{\boldsymbol{h}}}

\newcommand{\hbx}{\hat{\boldsymbol{x}}}

\newcommand{\Keps}{ \mathcal{N}_{\epsilon}}
\newcommand{\Kmu}{ \mathcal{N}_{\mu}}
\newcommand{\Kd}{ \mathcal{N}_{d_0}}
\newcommand{\KF}{ \mathcal{N}_{\widetilde{F}}}
\newcommand{\Kint}{\MN_{d_0} \cap \MN_{\mu} \cap \MN_{\epsilon}}

\newcommand{\bzero}{\boldsymbol{0}}

\newcommand{\A}{\mathcal{A}}

\newcommand{\PP}{\mathcal{P}}

\newcommand{\pa}{\partial}

\newcommand{\CC}{\mathbb{C}}

\newcommand{\CZ}{\mathcal{Z}}

\newcommand{\Dh}{\Delta\boldsymbol{h}}
\newcommand{\Dx}{\Delta\boldsymbol{x}}

\newcommand{\MS}{\mathcal{S}}

\newcommand{\MN}{\mathcal{N}}

\newcommand{\I}{\boldsymbol{I}}
\newcommand{\RR}{\mathbb{R}}
\newcommand{\lag}{\left\langle}
\newcommand{\rag}{\right\rangle}

\newcommand{\lp}{\left(} 
\newcommand{\rp}{\right)} 
\newcommand{\ls}{\left[} 
\newcommand{\rs}{\right]} 
\newcommand{\lc}{\left\{} 
\newcommand{\rc}{\right\}} 

\newcommand{\Tr}{\text{Tr}}
\newcommand{\eps}{\epsilon}
\newcommand{\TB}{T^{\bot}}

\newcommand{\mi}{\mathrm{i}}

\DeclareMathOperator{\Real}{Re}

\DeclareMathOperator{\E}{E}

\DeclareMathOperator{\diag}{diag}

\DeclareMathOperator{\subjectto}{\text{subject to}}

\DeclareMathOperator{\rank}{\text{rank}}

\renewcommand{\qed}{\rule{2.5mm}{2.5mm}}

\newcommand{\lsy}[1]{\textcolor{blue}{#1}}

\newcommand{\vct}[1]{\bm{#1}}
\newcommand{\mtx}[1]{\bm{#1}}
\definecolor{xl}{RGB}{200,50,50}
\newcommand{\XL}[1]{\textcolor{xl}{#1}}

\begin{document}
\title{\bf Rapid, Robust, and Reliable Blind Deconvolution via Nonconvex Optimization}


\author{Xiaodong Li\thanks{Department of Statistics, University of California Davis, Davis CA 95616}, Shuyang Ling\thanks{Department of Mathematics, University of California Davis, Davis CA 95616}, Thomas Strohmer$^{\dagger}$, and Ke Wei$^{\dagger}$}

\maketitle

\begin{abstract} 
We study the question of reconstructing two signals $\bff$ and $\bg$ from their convolution $\by = \bff\ast \bg$.
This  problem, known as {\em blind deconvolution}, pervades many areas of science and technology, including astronomy, medical imaging, optics, and wireless communications. 
A key challenge of this intricate non-convex optimization problem is that it  might exhibit many local minima.
We present an efficient numerical algorithm that is guaranteed to recover the exact solution, when the number of measurements is 
(up to log-factors) slightly larger than the information-theoretical minimum, and under reasonable conditions on $\bff$ and $\bg$.
The proposed regularized gradient descent algorithm converges at a geometric rate and is provably robust in the presence of noise.
To the best of our knowledge, our algorithm is the first blind deconvolution algorithm that is numerically
efficient, robust against noise, and comes with rigorous recovery guarantees under certain subspace conditions.
Moreover, numerical experiments do not only provide empirical verification of our theory, but they also demonstrate that our
method yields excellent performance even in situations beyond our theoretical framework.
\end{abstract}

\section{Introduction}
\label{s:intro}

Suppose we are given the convolution of two signals, $\by = \bff\ast \bg$.
When, under which conditions, and how can we reconstruct $\bff$ and $\bg$ from the knowledge of $\by$ 
if both $\bff$ {\em and} $\bg$
are unknown? This challenging problem, commonly referred to as {\em blind deconvolution} problem, arises
in many areas of science and 
technology, including astronomy, medical imaging, optics, and communications
engineering, see
e.g.~\cite{jefferies1993restoration,Tong95,chan1998total,wang1998blind,campisi2007blind,LWD11,CVR14,WBSJ14}.
Indeed, the quest for finding a fast and reliable algorithm for blind deconvolution has confounded
researchers for many decades. 

It is clear that without any additional assumptions, the blind deconvolution problem is ill-posed.
One common and useful assumption is to stipulate that $\bff$ and $\bg$ belong to known
subspaces~\cite{RR12,LLB15BIP,LLJB15,LS15}. This assumption is reasonable in various applications, provides
flexibility and at the same time lends itself to mathematical rigor. We also adopt this subspace assumption 
in our algorithmic framework (see Section~\ref{s:prelim} for details). But even with this assumption, blind deconvolution 
is a very difficult non-convex optimization problem that suffers from an overabundance of local minima, making its numerical
solution rather challenging.

 In this paper, we present a  numerically efficient  blind deconvolution algorithm that  converges geometrically
to the optimal solution.
Our regularized gradient descent algorithm comes with rigorous mathematical convergence guarantees.
The number of measurements required for the algorithm to succeed is only slightly larger than the information theoretic minimum.
Moreover, our algorithm is also robust against noise. To the best of our knowledge, the proposed algorithm  is the first blind deconvolution algorithm that is numerically
efficient, robust against noise, and comes with rigorous recovery guarantees under certain subspace conditions.


\subsection{State of the art}

Since blind deconvolution problems are ubiquitous in science and engineering, it is not surprising that there
is extensive literature on this topic. It is beyond the scope of this paper to review the existing literature;
instead we briefly discuss those results that are closest to our approach. We gladly acknowledge that those papers
that are closest to ours, namely~\cite{RR12,CLS14,KMO09IT,SL15}, are also the ones that greatly influenced our research in this project.

In the inspiring article~\cite{RR12}, Ahmed, Recht, and Romberg develop a convex optimization framework for blind 
deconvolution. The formal setup of our blind deconvolution problem follows essentially their setting. Using the meanwhile
well-known lifting trick,~\cite{RR12} transforms the blind deconvolution problem into the problem
of recovering a rank-one matrix from an underdetermined system of linear equations. By replacing the rank
condition by a nuclear norm condition, the computationally infeasible rank minimization problem turns into a
convenient convex problem.
The authors provide explicit conditions under which the resulting semidefinite program is guaranteed to have the
same solution as the original problem. In fact, the number of required measurements is not too far from the theoretical
minimum. The only drawback of this otherwise very appealing convex optimization approach is that the computational complexity of solving
the semidefinite program is  rather high for large-scale data and/or for applications where computation time
is of the essence. Overcoming this drawback was one of the main motivations for our paper. While~\cite{RR12} does
suggest a fast matrix-factorization based algorithm to solve the semidefinite program, the convergence of this
algorithm to the optimal solution is not established in that paper. The theoretical number of measurements required in~\cite{RR12}
for the semidefinite program to succeed is essentially comparable to that for our non-convex algorithm
to succeed. The advantage of the proposed non-convex algorithm is of course that it is dramatically faster. Furthermore, numerical 
simulations indicate that the empirically observed number of measurements for our non-convex approach is actually
 smaller than for the convex approach.

The philosophy underlying the method presented in our paper is strongly motivated by the non-convex optimization 
algorithm for phase retrieval proposed in~\cite{CLS14}, see also~\cite{CC15,CLM15}. In the pioneering paper~\cite{CLS14} 
the authors use a two-step approach: 
(i) Construct in a numerically efficient manner a good initial guess; (ii) Based on this initial guess, show
that simple gradient descent will converge to the true solution. Our paper follows a similar two-step scheme.
At first glance one would assume that many of the proof techniques from~\cite{CLS14} should carry over to the blind deconvolution problem.
Alas, we quickly found out that despite some general similarities between the two problems, phase retrieval and blind
deconvolution are indeed surprisingly different. At the end, we mainly adopted some of the general ``proof
principles'' from~\cite{CLS14} (for instance we also have a notion of local regularity condition - although it
deviates significantly from the one in~\cite{CLS14}), but the actual proofs are quite different.  For instance, in~\cite{CLS14} and~\cite{CC15}
convergence of the gradient descent algorithm is shown by directly proving that the distance between the true solution and the iterates decreases.
The key conditions (a local regularity condition and a local smoothness condition) are tailored to this aim. For the blind deconvolution problem
we needed to go a different route. We first show that the objective function decreases during the iterations and then use a certain local
restricted isometry property to transfer this decrease to the iterates to establish convergence  to the solution. 

We also gladly acknowledge being influenced by the papers~\cite{KMO09IT,KMO09noise} by Montanari and coauthors on matrix completion via 
non-convex methods. While the setup and analyzed problems are quite different from ours, their  approach informed our strategy in various ways. In~\cite{KMO09IT,KMO09noise}, 
the authors propose an algorithm which comprises a two-step procedure. First, an initial guess is computed
via a spectral method, and then a nonconvex problem is formulated and solved via an iterative method. The authors prove convergence to a low-rank solution,
but do not establish a rate of convergence.  As mentioned before, we also employ a two-step strategy. Moreover, our approach to prove stability 
of the proposed algorithm draws from ideas in~\cite{KMO09noise}. 

We also benefitted tremendously from~\cite{SL15}. In that paper, Sun and Luo devise a non-convex algorithm for
low-rank matrix completion and provide theoretical guarantees for convergence to the correct solution. We got the idea of
adding a penalty term to the objective function from~\cite{SL15} (as well as from~\cite{KMO09IT}). Indeed, the
particular structure of our penalty term closely resembles that in~\cite{SL15}. In addition, the introduction of the
various neighborhood regions around the true solution, that are used to eventually characterize a ``basin of
attraction'', stems
from~\cite{SL15}. These correspondences may not be too surprising, given
the connections between low-rank matrix completion and blind deconvolution.
Yet, like also discussed in the previous paragraph, despite some obvious similarities between the two problems, it
turned out that many steps and tools in our proofs differ significantly from those in~\cite{SL15}. Also this
should not come as a surprise, since the measurement matrices and setup differ significantly\footnote{Anyone who has experience in the proofs behind compressive sensing and matrix
completion is well aware of the substantial challenges one can already face by ``simply'' changing the sensing matrix from,
say, a Gaussian random matrix to one with less randomness.}.
Moreover, unlike~\cite{SL15}, our paper also provides robustness guarantees for the case of noisy data. Indeed, it
seems plausible that some of our techniques to establish robustness
against noise are applicable to the analysis of various recent matrix completion algorithms, such as e.g.~\cite{SL15}.

We briefly discuss other interesting papers that are to some extent related to our work. 
~\cite{LLJB15} proposes a projected gradient descent algorithm 
based on matrix factorizations and provide a convergence analysis to recover sparse signals from subsampled convolution.
However, this projection step can be hard to implement, which does impact the efficiency and practical use of this method. 
As suggested in~\cite{LLJB15}, one can avoid this expensive projection step by resorting 
to a heuristic approximate projection, but then the global convergence is not fully guaranteed. On the other hand, the papers~\cite{LLB15BIP,LLB15ID} consider identifiability issue of blind deconvolution problem with both $\bff$ and $\bg$ in random linear subspaces and achieve nearly optimal result of sampling complexity in terms of information theoretic limits. Very recently,~\cite{KK16} improved the result from~\cite{LLB15BIP,LLB15ID} by using techniques from algebraic geometry. 

The past few years have also witnessed an increasing number of excellent works other than blind deconvolution but related to nonconvex optimization~\cite{SQW16,WCCL15,Sun16,CLM15,LSJR16}. 
The paper~\cite{TBSR15} analyzes
the problem of recovering a low-rank positive semidefinite matrix from linear measurements via a gradient descent algorithm. 
The authors assume that the measurement matrix fulfills the standard and convenient restricted isometry property, a condition 
that is not suitable for the blind deconvolution problem (besides the fact that the positive semidefinite assumption is not satisfied in our setting). 
 In~\cite{CM15}, Chen and Wainwright study various the solution of low-rank estimation 
problems by projected gradient descent. 
The very recent paper~\cite{ZL16} investigates matrix completion for rectangular matrices. By ``lifting'', they convert the unknown matrix into a positive semidefinite one and apply matrix factorization combined with gradient descent to reconstruct the unknown entries of the matrix. ~\cite{CJ16} considers an interesting blind calibration problem with a special type of measurement matrix via nonconvex optimization. 
Besides some general similarities, there is little overlap of the aforementioned papers with our framework.
Finally, a convex optimization approach to blind deconvolution and self-calibration that extends the work of~\cite{RR12} can be found in~\cite{LS15}, 
while~\cite{LS15Blind} also covers the joint blind deconvolution-blind demixing problem.


\subsection{Organization of our paper}
This paper is organized as follows. We introduce some notation 
used throughout the paper in the remainder of this section.
The model setup and problem formulation are presented in Section~\ref{s:prelim}. Section~\ref{s:Algo} describes the proposed algorithm and our main theoretical
result establishing the convergence of our algorithm.  Numerical simulations can be found in Section~\ref{s:numerics}. 
Section~\ref{s:proof} is devoted to the proof of the main theorem.  Since the proof is quite involved, we have split this section 
into several subsections. Some auxiliary results are collected in the Appendix.

\subsection{Notation}
We introduce notation which will be used throughout the paper. Matrices and vectors are denoted in boldface such as
$\BZ$ and $\bz$. The individual entries of a matrix or a vector are denoted in normal font such as $Z_{ij}$ or
$z_i.$
For any matrix $\BZ$, $\|\BZ\|_*$ denotes its nuclear norm, i.e., the sum of its singular values; $\|\BZ\|$
denotes its operator norm, i.e., the largest singular value, and $\|\BZ\|_F$ denotes its the Frobenius norm, i.e.,
$\|\BZ\|_F =\sqrt{\sum_{ij} |Z_{ij}|^2 }$. For any vector $\bz$, $\|\bz\|$ denotes its Euclidean norm. For both
matrices and vectors, $\BZ^\top$ and $\bz^\top$ stand for the transpose of $\BZ$ and $\bz$ respectively while $\BZ^*$ and
$\bz^*$ denote their complex conjugate transpose.  We equip the matrix space $\CC^{K\times N}$ with the inner
product defined as $\lag \BU, \BV\rag : =\Tr(\BU^*\BV).$ A special case is the inner product of two vectors, i.e.,
$\lag \bu, \bv\rag = \Tr(\bu^*\bv) = \bu^*\bv.$  For a given vector $\bv$, $\diag(\bv)$ represents the diagonal
matrix whose diagonal entries are given by the vector $\bv$. For any $z\in \RR$, denote $z_+$ as $z_+ = \frac{z +
|z|}{2}.$
$C$ is an absolute constant and $C_{\gamma}$ is a constant which depends linearly on $\gamma$, but on no other parameters.

\section{Problem setup}
\label{s:prelim}

We consider the  blind deconvolution model
\begin{equation}\label{eq:model-0}
\by = \vct{f} \ast \bg + \bn,
\end{equation}
where $\by$ is given, but $\bff$ and $\bg$ are unknown. Here ``$\ast$" denotes circular convolution\footnote{As usual, ordinary convolution can be well approximated by circulant convolution,  as long as the function $\bff$ decays sufficiently fast~\cite{Str00}.}.  We will usually consider $\bff$ as the ``blurring function'' and $\bg$ as the signal of interest.
It is clear that without any further assumption it is impossible to recover $\bff$ and $\bg$ from $\by$.
We want to impose conditions on $\bff$ and $\bg$ that are realistic, flexible, and not tied to one particular
application (such as, say, image deblurring). At the same time, these conditions should be concrete enough to 
lend themselves to meaningful mathematical analysis. 

A natural setup that fits these demands is to assume that
$\bff$ and $\bg$ belong to known linear subspaces. Concerning the blurring function, it is reasonable in many applications to assume
that $\bff$ is either compactly supported or that $\bff$ decays sufficiently fast so that it can be well approximated by a compactly supported function. Therefore, we 
assume that $\vct{f}\in \CC^L$ satisfies  
\begin{equation}\label{def:ff}
\vct{f}: = 
\begin{bmatrix}
\bh \\ {\bf 0}_{L-K}
\end{bmatrix}
\end{equation}
where $\bh\in\CC^K$, i.e., only the first $K$ entries of $\bff$ are nonzero and $f_l = 0$ for all $l = K+1, K+2, \ldots, L$. 
Concerning the signal of interest, we assume that  $\bg$ belongs to
a linear subspace spanned by the columns of a known matrix $\BC$, i.e., $\bg = \BC\bar{\bx}$ for  some matrix $\BC$ of size $L \times N$. Here we use $\bar{\bx}$ instead of $\bx$ for the simplicity of notation later.
For theoretical purposes we assume that $\BC$ is a Gaussian random matrix. Numerical simulations suggest that this assumption is clearly not necessary. For example, we observed excellent performance of the proposed algorithm also in cases when $\BC$ represents a wavelet subspace
(appropriate for images) or when $\BC$ is a Hadamard-type matrix (appropriate for communications). We hope to address these other, more realistic,
choices for $\BC$ in our future research.
Finally, we assume that $\bn\sim \mathcal{N}(\bzero, \frac{\sigma^2d_0^2}{2}\I_L) + \mi\mathcal{N}(\bzero, \frac{\sigma^2d_0^2}{2}\I_L) $ as the additive white complex Gaussian noise where $d_0 = \|\bh_0\| \|\bx_0\|$ and $\bh_0$ and $\bx_0$ are the true blurring function and the true signal of interest, respectively. In that way $\sigma^{-2}$ actually serves as a measure of SNR (signal to noise ratio).
 
For our theoretical analysis as well as for numerical purposes, it is much more convenient to express~\eqref{eq:model-0} in the Fourier domain, see
also~\cite{RR12}.
To that end, let $\BF$ be the $L\times L$ unitary Discrete Fourier Transform (DFT) matrix and let the $L \times K$ matrix $\BB$ be given by
 the first $K$ columns of $\BF$ (then $\BB^{*} \BB = \BI_{K}$ ). By applying the scaled DFT matrix $\sqrt{L}\BF$ to  both sides of~\eqref{eq:model-0} we get
\begin{equation}
\label{conv2diag}
\sqrt{L}\BF {\by}= \diag(\sqrt{L} \BF \vct{f} ) (\sqrt{L}\BF \bg) + \sqrt{L}\BF \bn,
\end{equation}
which follows from the property of  circular convolution and Discrete Fourier Transform. 
Here, $ \diag(\sqrt{L} \BF \vct{f} ) (\sqrt{L}\BF \bg)= (\sqrt{L} \BF \vct{f} )\odot(\sqrt{L}\BF \bg)$ where $\odot$ denotes pointwise product. 
By definition of $\bff$ in~\eqref{def:ff}, we have 
\begin{equation*}
\BF \vct{f} = \BB\bh.
\end{equation*}
Therefore,~\eqref{conv2diag} becomes,
\begin{equation}
\label{conv2diag-2}
\frac{1}{\sqrt{L}}\hat{\by} = \diag(\BB\bh) \overline{\BA\bx} + \frac{1}{\sqrt{L}}\BF\bn
\end{equation}
where $\overline{\BA} = \BF\BC$ (we use $\overline{\BA}$ instead of $\BA$ simply because it gives rise to a more convenient notation later, see e.g.~\eqref{def:F}).

Note that if $\BC$ is a Gaussian matrix, then so is $\BF\BC$. Furthermore, $\be = \frac{1}{\sqrt{L}} \BF\bn\sim \mathcal{N}(\bzero, \frac{\sigma^2d_0^2}{2L}\I_L) + \mi\mathcal{N}(\bzero, \frac{\sigma^2d_0^2}{2L}\I_L) $ is again a complex Gaussian random vector.
Hence by replacing $\frac{1}{\sqrt{L}}\hat{\by}$ in~\eqref{conv2diag-2} by $\by$,  we arrive with a slight abuse of notation at
\begin{equation}\label{eq:model}
\by = \diag(\BB\bh) \overline{\BA\bx} + \be.
\end{equation}

For the remainder of the paper, instead of considering the original blind deconvolution problem~\eqref{eq:model-0}, we 
focus on its mathematically equivalent version~\eqref{eq:model},
where $\bh\in\CC^{K\times 1}$, $\bx \in \CC^{N\times 1}$, $\by\in\CC^{L\times 1}$, $\BB\in\CC^{L\times K}$ and $\BA\in\CC^{L\times N}$. 
As mentioned before, $\bh_0$ and $\bx_0$ are the ground truth. Our
goal is to recover $\bh_0$ and $\bx_0$ when $\BB$, $\BA$ and $\by$ are given. 
 It is clear that  if  $(\bh_0, \bx_0)$ is a solution to the blind deconvolution problem.
then so is $(\alpha \bh_0, \alpha^{-1} \bx_0)$ for any $\alpha \neq 0$. Thus, all we can hope for in the absence of any further information, is to recover
a solution from the equivalence class $(\alpha \bh_0, \alpha^{-1}\bx_0), \alpha \neq 0$. Hence, we can as well assume that  $\|\bh_0\| = \|\bx_0\| := \sqrt{d}_0$.


As already mentioned, we choose $\BB$ to be the ``low-frequency" discrete Fourier matrix, i.e., the first $K$ columns of an $L\times L$ unitary DFT (Discrete Fourier Transform) matrix. 
Moreover, we choose $\BA$ to be an $L\times N$ complex Gaussian random matrix, i.e., 
\begin{equation*}
A_{ij} \sim \mathcal{N}\left(0, \frac{1}{2}\right) + \mi \mathcal{N}\left(0, \frac{1}{2}\right),
\end{equation*}
where ``$\mi$" is the imaginary unit.

\bigskip

We define the matrix-valued linear operator $\A(\cdot)$ via
\begin{equation}\label{def:A}
\A: \A(\BZ) = \{\bb_l^*\BZ \ba_l\}_{l=1}^L,
\end{equation}
where $\bb_l$ denotes the $l$-th column of $\BB^*$ and $\ba_l$ is the $l$-th column of $\BA^*.$ Immediately, we have  $\sum_{l=1}^L \bb_l\bb_l^* = \BB^*\BB = \I_K$, $\|\bb_l\| = \frac{K}{L}$ and $\E(\ba_l\ba_l^*) = \I_N$ for all $1\leq l\leq L$. 
This is essentially a smart and popular trick called ``lifting"~\cite{CSV11,RR12,LS15}, which is able to convert a class of nonlinear models into linear models at the costs of increasing the dimension
of the solution space.

It seems natural and tempting to recover $(\bh_0, \bx_0)$ obeying~\eqref{eq:model} by solving the following optimization problem
\begin{equation}\label{eq:minF}
\min_{(\bh, \bx)}  \,\, F(\bh, \bx),
\end{equation}
where 
\begin{equation}\label{def:F}
F(\bh, \bx) := \|\diag(\BB\bh)\overline{\BA\bx} - \by\|^2 = \|\A(\bh\bx^* - \bh_0\bx_0^*) - \be\|^2.
\end{equation}
We also let 
\begin{equation}\label{def:F0}
F_0(\bh, \bx) : = \|\A(\bh\bx^* - \bh_0\bx_0^*)\|_F^2
\end{equation}
which is a special case of $F(\bh, \bx)$ when $\be = \bzero.$  Furthermore, we define $\delta(\bh, \bx)$, an important quantity throughout our discussion, via
 \begin{equation}\label{def:delta}
 \delta(\bh, \bx) : = \frac{\|\bh\bx^* - \bh_0\bx_0^*\|_F}{d_0}.
 \end{equation}
When there is no danger of ambiguity, we will often denote $\delta(\bh, \bx)$  simply by $\delta$. But let us remember that $\delta(\bh, \bx)$ is always a function of $(\bh, \bx)$ and measures the relative approximation error of $(\bh, \bx)$.

Obviously, minimizing~\eqref{def:F} becomes a nonlinear least square problem, i.e.,
one wants to find a pair of vectors $(\bh,\bx)$ or a rank-1 matrix $\bh\bx^*$ which fits the measurement equation in~\eqref{def:F} best. 
Solving~\eqref{eq:minF} is a challenging optimization problem since it is highly nonconvex and most of the available algorithms, such as alternating minimization and gradient descent, may suffer from getting easily trapped in some local minima. Another possibility is to consider a convex relaxation of~\eqref{eq:minF} at the cost of having to solve an expensive semidefinite program. In the next section we will describe how to avoid this dilemma and design an efficient gradient descent algorithm that, under reasonable conditions, will always converge to the true solution.

\section{Algorithm and main result}
\label{s:Algo}

In this section we introduce our algorithm as well as our main theorems which establish  convergence of the proposed
algorithm to the true solution.
As mentioned above, in a nutshell our algorithm consists of two parts: First we use a carefully chosen initial guess, 
and second we use a variation of gradient descent, starting at the initial guess
to converge to the true solution. One of the most critical aspects is of course that we must avoid getting stuck 
in local minimum or saddle point. Hence, we need to ensure that our iterates are inside some properly chosen
{\em basin of attraction} of the true solution. The appropriate characterization of such a basin of attraction
requires some diligence, a task that will occupy us in the next subsection. We will then proceed to introducing
our algorithm and analyzing its convergence.

\subsection{Building a basin of attraction}

The road toward designing a proper basin of attraction is basically paved by three observations, described below. 
These observations prompt us to introduce three neighborhoods (inspired by~\cite{KMO09IT,SL15}), whose intersection
will form the desired basin of attraction of the solution.

\medskip
\noindent
{\bf Observation 1 - Nonuniqueness of the solution:}
As pointed out earlier,  if $(\bh,\bx)$ is a solution to~\eqref{eq:model}, then so is $(\alpha \bh, \alpha^{-1}\bx)$ for any
$\alpha \neq 0$. Thus, without any prior information about $\|\bh\|$ and/or $\|\bx\|$, it is clear that we can only
recover the true solution up to such an unknown constant $\alpha$. Fortunately, this suffices for most
applications. From the viewpoint of numerical stability however, we do want to avoid, while $\|\bh\|\|\bx\|$
remains bounded, that $\|\bh\|\to 0$ and $\|\bx\| \to \infty$ (or vice versa). To that end we introduce the
following neighborhood:
\begin{equation}
\Kd :=  \{ (\bh, \bx) : \|\bh\| \leq 2\sqrt{d_0}, \|\bx\| \leq 2\sqrt{d_0} \} \label{def:Kd}.
\end{equation}
(Recall that $d_0 = \|\bh_0\|  \|\bx_0\|$.)


\medskip
\noindent
{\bf Observation 2 - Incoherence:}
Our numerical simulations indicate that the number of measurements required for solving the blind deconvolution problem
with the proposed algorithm does depend (among others) on how much 
$\bh_0$ is correlated with the rows of the matrix $\BB$ --- the smaller 
the correlation the better.
A similar effect has been observed in blind deconvolution via convex programming~\cite{RR12,LS15Blind}.
We quantify this property by defining the {\em incoherence} between the rows of $\BB$ and $\bh_0$ via
\begin{equation}\label{def:muh}
\mu^2_h = \frac{L \|\BB\bh_0\|_{\infty}^2}{\|\bh_0\|^2}.
\end{equation}
It is easy to see that $1\leq \mu_h^2 \leq K$ and both lower and upper bounds are tight; i.e.,  $\mu^2_h = K$  
if $\bh_0$ is parallel to one of $\{\bb_l\}_{l=1}^L$ and $\mu^2_h = 1$ if $\bh_0$ is a 1-sparse vector of length $K$.
Note that in our setup, we assume that $\BA$ is a random matrix and $\bx_0$ is fixed, thus with high probability, $\bx_0$ is already
sufficiently incoherent with the rows of $\BA$ and thus we only need to worry about the incoherence between $\BB$ and
$\bh_0$. 

It should not come as a complete surprise that the incoherence between $\bh_0$ and the rows of $\BB$ is important.
The reader may recall that in matrix completion~\cite{CR08,Recht11MC} the left and right singular vectors of the
solution cannot be ``too aligned" with those of the measurement matrices. A similar philosophy seems to apply here.
Being able to control the incoherence of the solution is instrumental in deriving rigorous convergence guarantees
of our algorithm. For that reason, we introduce  the neighborhood
\begin{equation}
\Kmu :=  \{ \bh : \sqrt{L} \|\BB\bh\|_{\infty} \leq 4\sqrt{d_0}\mu \}, \label{def:Kmu}
\end{equation}
where $\mu_h \leq \mu$.

\medskip
\noindent
{\bf Observation 3 - Initial guess:} It is clear that due to the non-convexity of the objective
function, we need a carefully chosen initial guess. 
We quantify the distance to the true solution via the following neighborhood
\begin{equation}
\Keps  :=  \{ (\bh, \bx) : \|\bh\bx^* - \bh_0\bx_0^*\|_F \leq \eps d_0 \}. \label{def:Keps}
\end{equation}
where $\eps$ is a predetermined parameter in $(0, \frac{1}{15}]$.

\bigskip
It is evident that the true solution $(\bh_0,\bx_0) \in \Kd \cap \Kmu$.
Note that $(\bh, \bx)\in \Kd \bigcap \Keps$ implies $\|\bh\bx^*\|\geq (1 - \eps)d_0$ and $\frac{1}{\|\bh\|} \leq \frac{\|\bx\|}{(1 - \eps)d_0} \leq \frac{2}{(1 - \eps)\sqrt{d_0}}$. Therefore, for any element $(\bh, \bx)\in \Kint,$ its incoherence can be well controlled by 
\begin{equation*}
\frac{\sqrt{L}\|\BB\bh\|_{\infty}}{\|\bh\|} \leq \frac{4\sqrt{d_0}\mu}{ \|\bh\| } \leq \frac{ 8\mu }{1 - \eps}.
\end{equation*}

\subsection{Objective function and key ideas of the algorithm}

Our approach consists of two parts: We first construct an initial guess that is inside the ``basin of attraction'' 
$\Kint$. We then apply
a carefully regularized gradient descent algorithm that will ensure that all the iterates remain inside $\Kint$.

Due to the difficulties of directly projecting onto $\Kd\cap \Kmu$ (the neigbourhood $\Keps$ is easier to manage) 
we add instead a regularizer $G(\bh,\bx)$
to the objective function $F(\bh, \bx)$ to enforce that the iterates remain inside $\Kd\cap\Kmu$.
While the idea of adding a penalty function to control incoherence is proposed in different forms
to solve matrix completion problems, see e.g.,~\cite{Sun15,LRST10}, our version is mainly 
inspired by~\cite{SL15,KMO09IT}. 

Hence, we aim to minimize the following regularized objective function to solve the blind deconvolution problem:
\begin{equation}\label{def:FG}
\tF(\bh, \bx) = F(\bh, \bx) + G(\bh, \bx),
\end{equation}
where $F(\bh, \bx)$ is defined in~\eqref{def:F} and $G(\bh, \bx)$, the penalty function, is of the form 
\begin{equation}\label{def:G}
G(\bh, \bx) = \rho \left[ G_0\left(\frac{\|\bh\|^2}{2d}\right) + G_0\left(\frac{\|\bx\|^2}{2d}\right) + \sum_{l=1}^L G_0\left(\frac{L|\bb_l^*\bh|^2}{8d\mu^2}\right) \right],
\end{equation}
where $G_0(z) = \max\{z - 1, 0 \}^2$ and $\rho \geq d^2  + 2\|\be\|^2$. Here we assume $\frac{9}{10}d_0 \leq d \leq \frac{11}{10}d_0$ and $\mu \geq \mu_h$.

The idea behind this, at first glance complicated, penalty function  is quite simple. The first two terms in~\eqref{def:G} enforce the
projection of $(\bh, \bx)$ onto $\Kd$ while the last term is related to $\Kmu$; it will be shown later that any $(\bh, \bx)\in \frac{1}{\sqrt{3}}\Kd \bigcap \frac{1}{\sqrt{3}}\Kmu$ gives $G(\bh, \bx) = 0$ if $\frac{9}{10}d_0 \leq d\leq \frac{11}{10}d_0$. 
Since $G_0(z)$ is a truncated quadratic function, it is obvious that $G_0'(z) = 2\sqrt{G_0(z)}$ and $G(\bh, \bx)$ is a continuously differentiable function. Those two properties play a crucial role in proving geometric convergence of our algorithm presented later. 

Another important issue concerns the selection of parameters. We have three unknown parameters in $G(\bh, \bx)$, i.e., $\rho$, $d$, and $\mu$. Here,  $d$ can be obtained via Algorithm~1 and $\frac{9}{10}d_0 \leq d\leq \frac{11}{10}d_0$ is guaranteed by Theorem~\ref{thm:init}; $\rho \geq d^2 + 2\|\be\|^2 \approx d^2 +  2\sigma^2 d_0^2$ because $\|\be\|^2 \sim \frac{\sigma^2 d_0^2}{2L}\chi^2_{2L}$ and $\|\be\|^2$ concentrates around $\sigma^2 d_0^2$. In practice, $\sigma^2$ which is the inverse of the SNR, can often be well estimated.
Regarding $\mu$, we require $\mu_h \leq \mu$ in order to make sure that $\bh_0\in \Kmu$. It will depend on the specific application how well one can estimate $\mu_h^2$. For instance, in wireless communications, a very common channel model for $\bh_0$ is to assume a {\em Rayleigh fading model}~\cite{TV05}, i.e., $\bh_0 \sim \mathcal{N}(\bzero, \frac{\sigma^2_h}{2}\I_K) + \mi \mathcal{N}(\bzero, \frac{\sigma^2_h}{2}\I_K)$. In that case it is easy to see that $\mu_h^2 = \mathcal{O}(\log L).$


\subsection{Wirtinger derivative of the objective function and algorithm}

Since we are dealing with complex variables, instead of using regular derivatives it is more convenient to utilize Wirtinger derivatives~\footnote{For any complex function $f(\bz)$ where $\bz = \bu + \mi \bv\in\CC^n$ and $\bu, \bv\in\RR^n$, the Wirtinger derivatives are defined as $\frac{\pa f}{\pa \bz} = \frac{1}{2}\left( \frac{\pa f}{\pa \bu} - \mi \frac{\pa f}{\pa \bv}\right)$ and $\frac{\pa f}{\pa \bar{\bz}} = \frac{1}{2}\left( \frac{\pa f}{\pa \bu} + \mi \frac{\pa f}{\pa \bv}\right).$ Two important examples used here are $\frac{\pa \|\bz\|^2}{\pa \bar{\bz}} = \bz$ and $\frac{\pa (\bz^*\bw) }{\pa \bar{\bz}} = \bw$.
}, which has become increasingly popular since~\cite{CLS14,CLM15}. 
Note that $\tF$ is a real-valued function and hence we only need to consider the derivative of $\tF$ with respect to $\bar{\bh}$ and $\bar{\bx}$ and the corresponding updates of $\bh$ and $\bx$ because a simple relation holds, i.e.,
\begin{equation*}
\frac{\pa \tF}{\pa \bar{\bh}} = \overline{ \frac{\pa \tF}{\pa \bh}} \in \CC^{K\times 1}, \quad \frac{\pa \tF}{\pa \bar{\bx}} = \overline{ \frac{\pa \tF}{\pa \bx}} \in \CC^{N\times 1}.
\end{equation*}
In particular, we denote $\nabla \tF_{\bh} := \frac{\pa \tF}{\pa \bar{\bh}}$ and $\nabla \tF_{\bx} := \frac{\pa \tF}{\pa \bar{\bx}}$.

We also introduce the adjoint operator of $\A: \CC^L \rightarrow \CC^{K\times N}$, given by
\begin{equation}\label{def:AD}
\A^*(\bz) = \sum_{l=1}^L z_l \bb_l\ba_l^*.
\end{equation}

Both $\nabla \tF_{\bh}$ and  $\nabla \tF_{\bx}$  can now be expressed as
\begin{equation*}
\nabla \tF_{\bh} = \nabla F_{\bh} + \nabla G_{\bh}, \quad \nabla \tF_{\bx} = \nabla F_{\bx} + \nabla G_{\bx},
\end{equation*}
where each component yields
\begin{eqnarray}
\nabla F_{\bh} & = & \A^*(\A(\bh\bx^*) - \by)\bx = \A^*(\A(\bh\bx^* - \bh_0\bx_0^*) -\be)\bx, \label{eq:WFh} \\
\nabla F_{\bx} & = & [\A^*(\A(\bh\bx^*) - \by)]^*\bh = [\A^*(\A(\bh\bx^* - \bh_0\bx_0^*) - \be)]^*\bh, \label{eq:WFx}\\
\nabla G_{\bh}
& = & \frac{\rho}{2d}\left[G'_0\left(\frac{\|\bh\|^2}{2d}\right) \bh 
+ \frac{L}{4\mu^2} \sum_{l=1}^L G'_0\left(\frac{L|\bb_l^*\bh|^2}{8d\mu^2}\right) \bb_l\bb_l^*\bh \right], \label{eq:WGh} \\
\nabla G_{\bx} & = & \frac{\rho}{2d} G'_0\left( \frac{\|\bx\|^2}{2d}\right) \bx. \label{eq:WGx}
\end{eqnarray}


\bigskip

Our algorithm consists of two steps: initialization and gradient descent with constant stepsize. The initialization is achieved via a spectral method followed by projection. The idea behind spectral method is that 
\begin{equation*}
\E \A^*(\by) = \E \A^*\A(\bh_0\bx_0^*) + \E \A^*(\be) = \bh_0\bx_0^*
\end{equation*}
and hence one can hope that the leading singular value and vectors of $\A^*(\by)$ can be a good approximation of $d_0$ and $(\bh_0, \bx_0)$ respectively. 
The projection step ensures $\bu_0\in \Kmu$, which the spectral method alone might not guarantee. 
We will address the implementation and computational complexity issue in Section~\ref{s:numerics}.
\begin{algorithm}[h!]
\caption{Initialization via spectral method and projection}
\label{Initial}
\begin{algorithmic}[1]
\State Compute $\A^*(\by).$
\State Find the leading singular value, left and right singular vectors of $\A^*(\by)$, denoted by $d$, $\hat{\bh}_0$ and $\hat{\bx}_0$ respectively. 
\State Solve the following optimization problem:
\begin{equation*}
\bu_0 := \text{argmin}_{\bz} \|\bz - \sqrt{d}\hbh_0\|^2, \quad \subjectto \sqrt{L}\|\BB\bz\|_{\infty} \leq 2\sqrt{d}\mu
\end{equation*}
and $\bv_0 = \sqrt{d}\hat{\bx}_0.$ 
\State Output: $(\bu_0, \bv_0).$
\end{algorithmic}
\end{algorithm}

\begin{algorithm}[h!]
\caption{Wirtinger gradient descent with constant stepsize $\eta$}
\label{AGD}
\begin{algorithmic}[1]
\State {\bf Initialization:} obtain $(\bu_0, \bv_0)$ via Algorithm~\ref{Initial}.
\For{ $t = 1, 2, \dots, $}
\State $\bu_t = \bu_{t-1} - \eta \nabla \tF_{\bh}(\bu_{t-1}, \bv_{t-1})$
\State $\bv_t = \bv_{t-1} - \eta \nabla \tF_{\bx}(\bu_{t-1}, \bv_{t-1})$
\EndFor

\end{algorithmic}
\end{algorithm}

\subsection{Main results}

Our main finding is that with a diligently chosen initial guess $(\bu_0, \bv_0)$, simply running gradient descent to minimize the regularized non-convex objective function
$\tF(\bh,\bx)$
will not only guarantee  linear convergence of the sequence $(\bu_t, \bv_t)$ to the global minimum $(\bh_0, \bx_0)$ in the noiseless case, but also provide robust recovery in the presence of noise. The results are summarized in the following two theorems. 

\begin{theorem}\label{thm:init}
The initialization obtained via Algorithm~\ref{Initial} satisfies
\begin{equation}\label{eq:init-val}
(\bu_0, \bv_0) \in \frac{1}{\sqrt{3}}\Kd\bigcap \frac{1}{\sqrt{3}} \MN_{\mu} \bigcap \MN_{\frac{2}{5}\eps },
\end{equation}
and  
\begin{equation}\label{eq:d}
\frac{9}{10}d_0 \leq d\leq \frac{11}{10}d_0
\end{equation}
holds with probability at least
$1 - L^{-\gamma}$ if the number of measurements satisfies
\begin{equation}
\label{Lbound}
L \geq C_{\gamma} (\mu_h^2 + \sigma^2)\max\{K, N\} \log^2 (L)/\eps^2.
\end{equation}
Here $\eps$ is any predetermined constant in $(0, \frac{1}{15}]$, and $C_{\gamma}$ is a constant only linearly depending on $\gamma$ with $\gamma \geq 1$. 
\end{theorem}

The proof of Theorem~\ref{thm:init} is given in Section~\ref{s:init}. 
While the initial guess is carefully chosen, it is in general not of sufficient accuracy to already be used as good approximation to the true solution. 
The following theorem establishes that as long as the initial guess lies inside the basin of attraction of the true solution, regularized gradient descent
will indeed converge to this solution (or to a solution nearby in case of noisy data).

\begin{theorem}\label{thm:main}
Consider the model in~\eqref{eq:model} with the ground truth $(\bh_0, \bx_0)$, $\mu^2_h = \frac{L\|\BB\bh_0\|_{\infty}^2}{\|\bh_0\|^2}\leq\mu^2$ 
and the noise $\be\sim \mathcal{N}(\bzero, \frac{\sigma^2d_0^2}{2L}\I_L) + \mi\mathcal{N}(\bzero, \frac{\sigma^2d_0^2}{2L}\I_L)$. 
Assume that the initialization  $(\bu_0, \bv_0)$ belongs to 
$\frac{1}{\sqrt{3}}\Kd\bigcap \frac{1}{\sqrt{3}} \Kmu \bigcap \MN_{\frac{2}{5}\eps }$ and that
\begin{equation*}
L \geq C_{\gamma} (\mu^2 + \sigma^2)\max\{K, N\} \log^2 (L)/\eps^2,
\end{equation*}
Algorithm~\ref{AGD} will create a sequence $(\bu_t, \bv_t)\in \Kint$ which converges geometrically to $(\bh_0, \bx_0)$ in the sense that with probability at least $1 - 4L^{-\gamma} - \frac{1}{\gamma}\exp(-(K + N))$, there holds
\begin{equation}\label{eq:main-res1}
\max\{ \sin \angle(\bu_t, \bh_0), \sin \angle(\bv_t, \bx_0)\} \leq \frac{1}{d_t}\left( \frac{2}{3}(1 - \eta\omega)^{t/2}\eps d_0 + 50\|\A^*(\be)\| \right)
\end{equation}
and 
\begin{equation}\label{eq:main-res2}
|d_t - d_0| \leq \frac{2}{3}(1 - \eta\omega)^{t/2}\eps d_0 + 50 \|\A^*(\be)\|,
\end{equation}
where $d_t := \|\bu_t\|\|\bv_t\|$, $\omega  > 0$, $\eta$ is the fixed stepsize and  $\angle(\bu_t, \bh_0)$ is the angle between $\bu_t$ and $\bh_0$.
Here 
\begin{equation}\label{eq:Ae-1}
\|\A^*(\be)\| \leq C_0\sigma d_0 \max\Big\{ \sqrt{\frac{(\gamma + 1)\max\{K,N\} \log L}{L}}, \frac{(\gamma + 1)\sqrt{KN}\log^2 L }{L}  \Big\}.
\end{equation}
holds with probability $1 - L^{-\gamma}.$
\end{theorem}

\medskip
\noindent
{\bf Remarks:}
\begin{enumerate}
\item
While the setup in~\eqref{eq:model} assumes that $\BB$ is a matrix consisting of the first $K$ columns of the DFT matrix, this is actually not necessary
for Theorem~\ref{thm:main}. As the proof will show, the only conditions on $\BB$ are that $\BB^{\ast} \BB = \BI_K$ and that the norm of the $l$-th row
of $\BB$ satisfies $\| \bb_l \|^2 \leq C \frac{K}{L}$ for some numerical constant $C$.
\item
The minimum number of measurements required for our method to succeed is roughly comparable to that of the convex approach proposed in~\cite{RR12}
(up to log-factors). Thus there is no price to be paid for trading a slow, convex-optimization based approach with a
fast non-convex based approach. Indeed, numerical experiments indicate that the non-convex approach even requires a smaller number of measurements compared
to the convex approach, see Section~\ref{s:numerics}.
\item
The  convergence rate of our algorithm is completely determined by $\eta\omega$. Here, the {\em regularity constant} $\omega = \mathcal{O}(d_0)$ is specified in~\eqref{eq:reg} and $\eta \leq \frac{1}{C_L}$ where $C_L = \mathcal{O}(d_0 ( N\log L + \frac{\rho L}{d_0^2 \mu^2}))$. The attentive reader may have noted that $C_L$ depends essentially linearly on $\frac{\rho L}{\mu^2}$, which actually reflects a tradeoff between sampling complexity (or statistical estimation quality) and computation time. Note that if $L$ gets larger, the number of constraints is also increasing and hence leads to a larger $C_L$. However, this issue can be solved by choosing parameters smartly. Theorem~\ref{thm:main} tells us that $\mu^2$ should be roughly between $\mu_h^2$ and $\frac{L}{(K + N)\log^2L}$. Therefore, by choosing $\mu^2 = \mathcal{O}(\frac{L}{(K + N)\log^2L})$ and $\rho \approx d^2 + 2\|\be\|^2$, $C_L$ is optimized and
\begin{equation*}
\eta\omega = \mathcal{O}(((1 + \sigma^2)(K + N) \log^2L)^{-1}),
\end{equation*}
which is shown in details in Section~\ref{s:smooth}.
\item
Relations~\eqref{eq:main-res1} and~\eqref{eq:main-res2} are basically equivalent to the following:
\begin{equation}\label{eq:main-result}
\| \bu_t\bv_t^* - \bh_0\bx_0^*\|_F \leq \frac{2}{3}(1 - \eta\omega)^{t/2}\eps d_0 + 50 \|\A^*(\be)\|,
\end{equation}
which says that $(\bu_t, \bv_t)$ converges to an element of the equivalence class associated with the true solution $(\bh_0,\bx_0)$  (up to a deviation 
governed by the amount of additive noise).
\item
The matrix $\A^*(\be) = \sum_{l=1}^L e_l \bb_l\ba_l^*$, as a sum of $L$ rank-1 random matrices, has nice concentration of measure properties under the assumption of Theorem~\ref{thm:main}. Asymptotically, $\|\A^*(\be)\|$ converges to $0$ with rate $\mathcal{O}(L^{-1/2})$, which will be justified in Lemma~\ref{lem:Ay-hx} of Section~\ref{s:init} (see also~\cite{CLM15}).
Note that 
\begin{equation}\label{eq:F-decom}
F(\bh, \bx) = \|\be\|^2 + \|\A(\bh\bx^* - \bh_0\bx_0^*)\|_F^2 - 2\Real (\lag \A^*(\be) , \bh\bx^* - \bh_0\bx_0^*\rag).
\end{equation}
If one lets $L\rightarrow \infty$, then $\|\be\|^2\sim \frac{\sigma^2 d_0^2}{2L} \chi^2_{2L}$ will converge almost surely to $\sigma^2d_0^2$ under the Law of Large Numbers and the cross term $\Real (\lag \bh\bx^* - \bh_0\bx_0^*, \A^*(\be) \rag)$ will converge to $0$. In other words, asymptotically, 
\begin{equation*}
\lim_{L\rightarrow \infty} F(\bh, \bx) = F_0(\bh, \bx) + \sigma^2 d_0^2
\end{equation*}
for all fixed $(\bh, \bx)$. This implies that if the number of measurements is large, then $F(\bh, \bx)$ behaves ``almost  like" $F_0(\bh, \bx) = \|\A(\bh\bx^* - \bh_0\bx_0^*)\|^2$, the noiseless version of $F(\bh, \bx)$. This provides the key insight into analyzing the robustness of our algorithm, which is reflected in the so-called ``\textit{Robustness Condition}" in~\eqref{eq:AEnorm}. Moreover, the asymptotic property of $\A^*(\be)$ is also seen in our main result~\eqref{eq:main-result}. Suppose $L$ is becoming larger and larger, the effect of noise diminishes and heuristically, we might just rewrite our result as $\|\bu_t\bv_t^* - \bh_0\bx_0^*\|_F \leq \frac{2}{3}(1 - \eta\omega)^{t/2}\eps d_0 + o_p(1)$, which is consistent with the result without noise. 
\end{enumerate}

\section{Numerical simulations}
\label{s:numerics}

We present empirical evaluation of our proposed gradient descent algorithm (Algorithm~\ref{AGD}) using simulated data as well as examples from blind deconvolution problems appearing in  communications and in image processing. 

\subsection{Number of measurements  vs size of signals}\label{sec:simulation,phase}
We first investigate how many  measurements  are necessary in order for an algorithm to reliably recover two signals from their convolution. We compare Algorithm~\ref{AGD}, a gradient descent algorithm for  the sum of the loss function $F(\bh,\bx)$ and the regularization term $G(\bh,\bx)$, with the gradient descent algorithm only applied to $F(\bh,\bx)$ and  the nuclear norm minimization proposed in \cite{RR12}. These three tested algorithms are abbreviated as {\em regGrad, Grad} and {\em NNM} respectively.  To make fair comparisons, both regGrad and Grad are initialized with the normalized leading singular vectors of $\A^*(\by)$, which are computed by running the power method for $50$ iterations. Though we do not further compute the projection of $\hat{\bh}_0$  for regGrad as stated in the third step of Algorithm~\ref{Initial}, we emphasize that the projection can be computed efficiently as it is a  linear programming on $K$-dimensional vectors. A careful reader may notice that in addition to the computational cost for the loss function $F(\bh,\bx)$, regGrad also requires to evaluate $G(\bh,\bx)$ and its gradient in each iteration. 
When $B$ consists of the first $K$ columns of  a unitary DFT matrix,  we can evaluate $\left\{\bb_l^*\bh\right\}_{l=1}^L$ and the gradient of $\sum_{l=1}^L G_0\left(\frac{L|\bb_l^*\bh|^2}{8d\mu^2}\right)$ using FFT.  Thus the additional per iteration computational cost for the regularization term $G(\bh,\bx)$ is only $O(L\log L)$ flops.
The stepsizes in both regGrad and Grad are selected adaptively in each iteration via backtracking.  As suggested by the theory, the choices for $\rho$ and $\mu$ in regGrad are  $\rho=d^2/100$ and $\mu=6\sqrt{L/(K+N)}/\log L$.

We conduct tests on random Gaussian signals $\bh\in\CC^{K\times 1}$ and $\bx\in\CC^{N\times 1}$ with $K=N=50$. The matrix $\BB\in\CC^{L\times K}$ is the first $K$ columns of a unitary $L\times L$ DFT matrix, while $\BA\in\CC^{L\times N}$ is either a Gaussian random matrix or a partial Hadamard matrix with randomly selected $N$ columns and then multiplied by a random sign matrix from the left. When $\BA$ is a Gaussian random matrix, $L$ takes $16$ equal spaced values from $K+N$ to $4(K+N)$. When $\BA$ is a partial Hadamard matrix, we only test $L=2^s$ with $6\leq s\leq 10$ being integers.  For each given triple $(K,N,L)$, fifty random tests are conducted. We consider an algorithm to have successfully recovered $(\bh_0,\bx_0)$ if it returns a matrix $\hat{\BX}$ which satisfies
\begin{align*}
\frac{\|\hat{\BX}-\bh_0\bx_0^*\|_F}{\| \bh_0\bx_0^*\|_F}<10^{-2}.
\end{align*}
We present the probability of successful recovery plotted against the number of measurements in Figure~\ref{fig_phase}.
\begin{figure*}[!t]
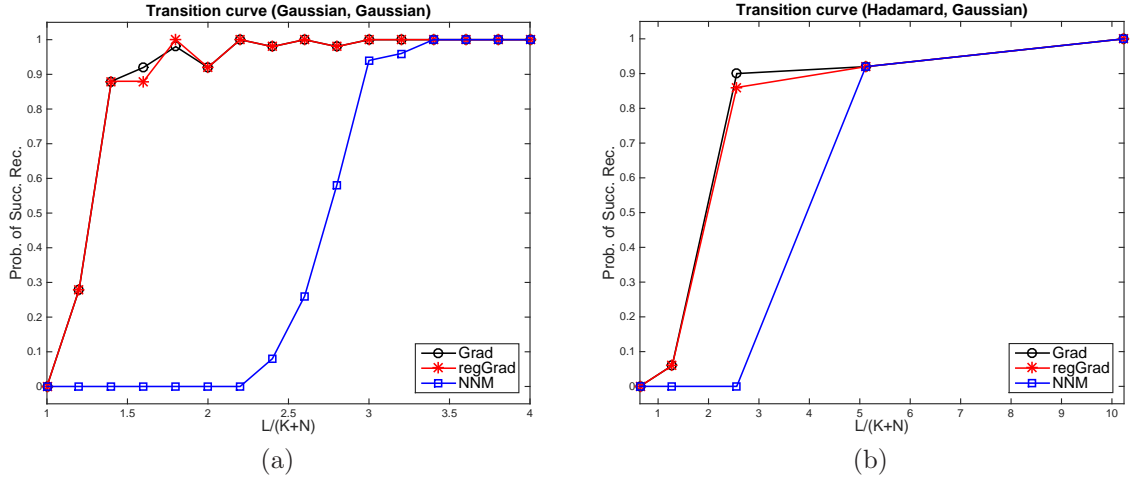

\centering
\subfloat[]{\includegraphics[width=70mm]{transition_plot_Gaussian_Gaussian.eps}
\label{fig_phase_Gaussian}}
\hfil
\subfloat[]{\includegraphics[width=70mm]{transition_plot_Hadamard_Gaussian.eps}
\label{fig_phase_Hadamard}}
\caption{Empirical phase transition curves when (a) $\BA$ is random Gaussian (b) $\BA$ is partial Hadamard. Horizontal axis $L/(K+N)$ and 
vertical axis probability of successful recovery out of $50$ random tests.}
\label{fig_phase}
\end{figure*}
It can be observed that regGrad and Grad have similar performance, and both of them require a significantly smaller number of measurements than NNM to achieve successful recovery of high probability. 
\subsection{Number of measurements vs incoherence}\label{sec:simulation,LvsIncoherence}
Theorem~\ref{thm:main} indicates that the number of  measurements $L$ required for Algorithm~\ref{AGD} to achieve successful recovery scales linearly with $\mu_h^2$. We conduct numerical experiments to investigate the dependence of $L$ on $\mu_h^2$ empirically. 
The tests are conducted with $\mu_h^2$ taking on $10$ values $\mu_h^2\in\{ 3,6,\cdots,30\}$. For each fixed $\mu_h^2$, we choose $\bh_0$ to be a vector whose first $\mu_h^2$ entries are $1$ and the others are $0$ so that its incoherence is equal to $\mu_h^2$ when $\BB$  is low frequency Fourier matrix. Then the tests are repeated for random Gaussian  matrices $\BA$ and random Gaussian vectors $\bx_0$. The empirical probability of successful recovery on the $(\mu_h^2,L)$ plane is presented in Figure~\ref{fig_L_mu}, which suggests that $L$ does scale linearly with $\mu_h^2$.
\begin{figure*}[!t]
\centering
\includegraphics[width=100mm]{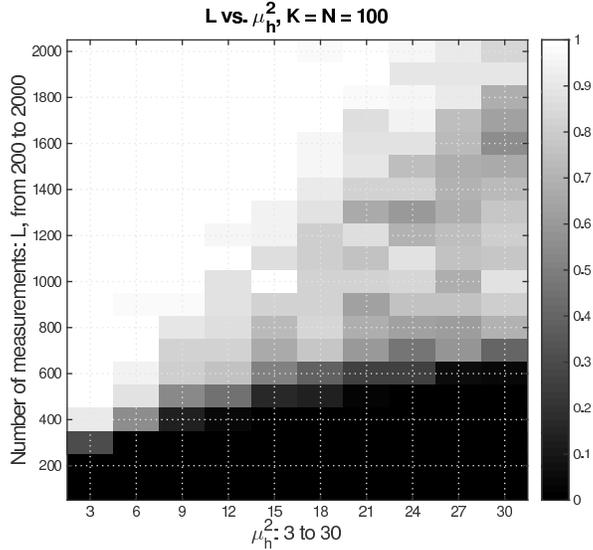}
\caption{Empirical probability of successful recovery.   Horizontal axis $\mu_h^2$ and vertical axis $L$. White: $100\%$ success and black: $0\%$ success.}
\label{fig_L_mu}
\end{figure*}
\subsection{A comparison when $\mu_h^2$ is large}
While regGrad and Grad have similar performances in the simulation when $\bh_0$ is a random Gaussian signal (see Figure~\ref{fig_phase}), we investigate their performances on a fixed $\bh_0$ with a large incoherence. 
 The tests are conducted for $K=N=200$, $\bx_0\in\CC^{N\times 1}$ being a random Gaussian signal,  $\BA\in\CC^{L\times N}$ being a random Gaussian matrix, and $\BB\in\CC^{L\times K}$ being a low frequency Fourier matrix. The signal $\bh_0$ with $\mu_h^2=100$ is formed in the same way as in Section~\ref{sec:simulation,LvsIncoherence}; that is, the first $100$ entries of $\bh_0$ are one and the other entries are zero. The number of measurements $L$ varies from $3(K+N)$ to $8(K+N)$. For each $L$, $100$ random tests are conducted. Figure~\ref{fig_phase_fixed_h}  shows the probability of successful recovery for regGrad and Grad. It can be observed that the successful recovery probability  of regGrad is at least $10\%$ larger than that of Grad when $L\geq 6(K+N)$.
 
 \begin{figure*}[!t]
\centering
\includegraphics[width=80mm]{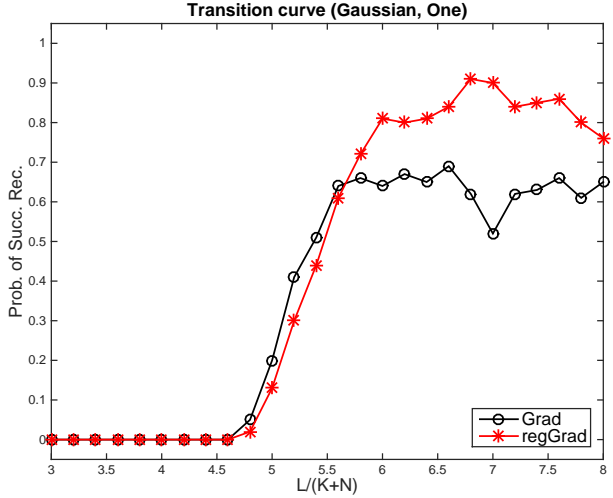}
\caption{Empirical phase transition curves when $\mu_h^2$ is large. Horizontal axis $L/(K+N)$ and 
vertical axis probability of successful recovery out of $50$ random tests.}
\label{fig_phase_fixed_h}
\end{figure*}

\subsection{Robustness to additive noise}

We explore the robustness of Algorithm~\ref{AGD} when the measurements are contaminated by additive noise. The tests are conducted with $K=N=100$, $L\in\{500, 1000\}$ when $\BA$ is a random Gaussian matrix, and $L\in\{512,1024\}$ when $\BA$
is a partial Hadamard matrix. Tests with additive noise have the measurement vector $\by$ corrupted by the vector 
\begin{align*}
\be=\sigma\cdot\|\by\|\cdot\frac{\bm{w}}{\|\bm{w}\|},
\end{align*}
where $\bm{w}\in\CC^{L\times 1}$ is standard Gaussian random vector, and $\sigma$ takes nine different values from $10^{-4}$ to $1$. For each $\sigma$,  fifty random tests are conducted. The average reconstruction error  in dB plotted against the signal to noise ratio (SNR) is presented in Fig.~\ref{fig_stability}. First  the
plots clearly show the desirable linear scaling between the noise levels and the relative reconstruction errors. Moreover, as desired, the relative reconstruction error decreases linearly on a $\log$-$\log$ scale as the number of measurements $L$ increases.
\begin{figure*}[!t]
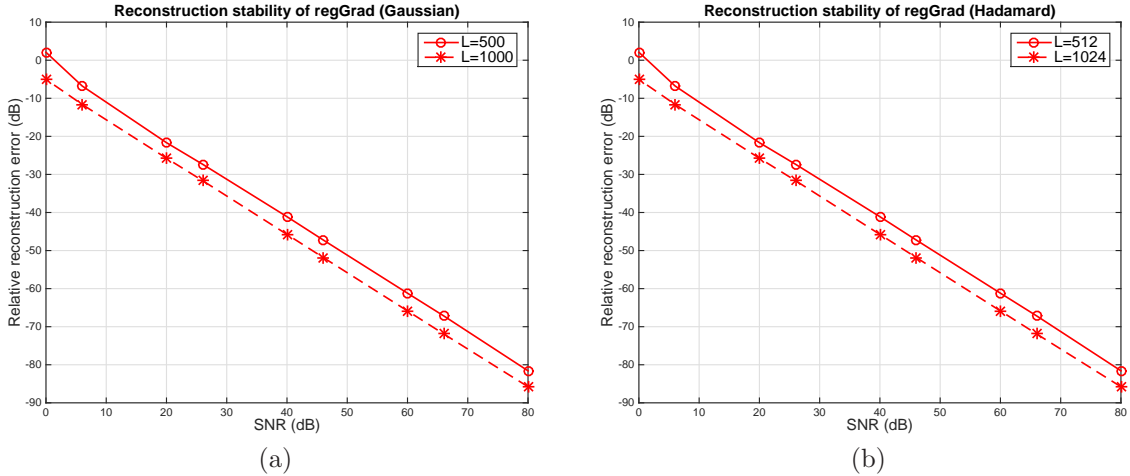

\centering
\subfloat[]{\includegraphics[width=70mm]{stability_of_regGrad_Gaussian.eps}
\label{fig_stability_Gaussian}}
\hfil
\subfloat[]{\includegraphics[width=70mm]{stability_of_regGrad_Hadamard.eps}
\label{fig_stability_Hadamard}}
\caption{Performance of Algorithm~\ref{AGD} under different SNR when (a) $\BA$ is random Gaussian (b) $\BA$ is partial Hadamard.}
\label{fig_stability}
\end{figure*}
\subsection{An example from communications}

In order to demonstrate the effectiveness of Algorithm~\ref{AGD} for real world applications, we first test the algorithm on a blind deconvolution problem arising in communications. Indeed, blind deconvolution problems and their efficient numerical solution are expected to play an increasingly important role in connection  with the emerging Internet-of-Things~\cite{WBSJ14}. 
Assume we want to transmit a signal from one place to another over a so-called time-invariant multi-path communication channel, but the receiver has no information about the channel, except its {\em delay spread} (i.e., the support of the impulse response). In many communication settings it is reasonable to assume that the receiver has information about the signal encoding matrix---in other words, we know the subspace $\BA$ to which $\bx_0$ belongs to.
Translating this communications jargon into mathematical terminology, this simply means that we are dealing with a blind deconvolution problem
of the form~\eqref{eq:model}.

These encoding matrices (or so-called spreading matrices) are often chosen to have a convenient structure that lends itself to fast computations
and minimal memory requirements. One such choice is to let $\BA:=\BD \BH$, where the $L \times K$ matrix $\BH$ consists of $K$ (randomly or not randomly) chosen columns of an $L \times L$ Hadamard matrix, premultiplied with an $L \times L$ diagonal random sign matrix $\BD$.  Instead of a partial Hadamard matrix we could also use a partial Fourier matrix; the resulting setup 
would then closely resemble an OFDM transmission format, which is part of every modern wireless communication scheme.
 For the signal $\bx_0$ we choose a so-called QPSK scheme, i.e., each entry of  $\bx_0\in\CC^{123\times 1}$ takes 
a value from $\{1,-1,\mi,-\mi\}$ with equal probability.  The actual transmitted signal is then $\bz  = \BA \bx_0$. 
For the channel $\bh_0$ (the blurring function) we choose a real-world channel, courtesy of Intel Corporation. 
Here, $\bh_0$ represents the impulse response of a multipath time-invariant channel, it consists of 123 time samples, hence 
$\bh_0\in\CC^{123\times 1}$. 
Otherwise, the setup follows closely that in Sec.~\ref{sec:simulation,phase}. For comparison we also include the case when $\BA$ is a random Gaussian matrix.

The plots of successful recovery probability are presented in Figure~\ref{fig_phase_real}, which shows that  Algorithm~\ref{AGD} can successfully recover the real channel $\bh_0$ with high probability if $L\gtrsim 2.5(L+N)$ when $\BA$ is a random Gaussian matrix and if $L\gtrsim 4.5(L+N)$ when $\BA$ is a partial Hadamard matrix. Thus our theory seems a bit pessimistic. It is gratifying to see that very little additional measurements are required compared to the number of unknowns in order
to recover the transmitted signal.

\begin{figure*}[!t]
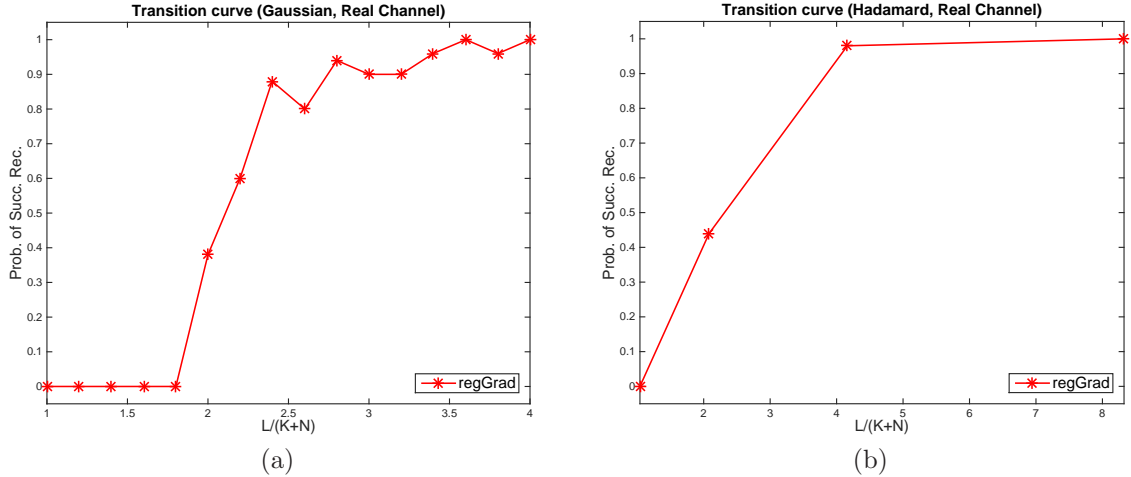

\centering
\subfloat[]{\includegraphics[width=70mm]{transition_plot_Gaussian_Real.eps}
\label{fig_phase_Gaussian_Real}}
\hfil
\subfloat[]{\includegraphics[width=70mm]{transition_plot_Hadamard_Real.eps}
\label{fig_phase_Hadamard_Real}}
\caption{Empirical phase transition curves when (a) $\BA$ is random Gaussian (b) $\BA$ is partial Hadamard. Horizontal axis $L/(K+N)$ and 
vertical axis probability of successful recovery out of $50$ random tests.}
\label{fig_phase_real}
\end{figure*}


\begin{figure*}[!t]
\centering
\subfloat[]{\includegraphics[width=2.3in]{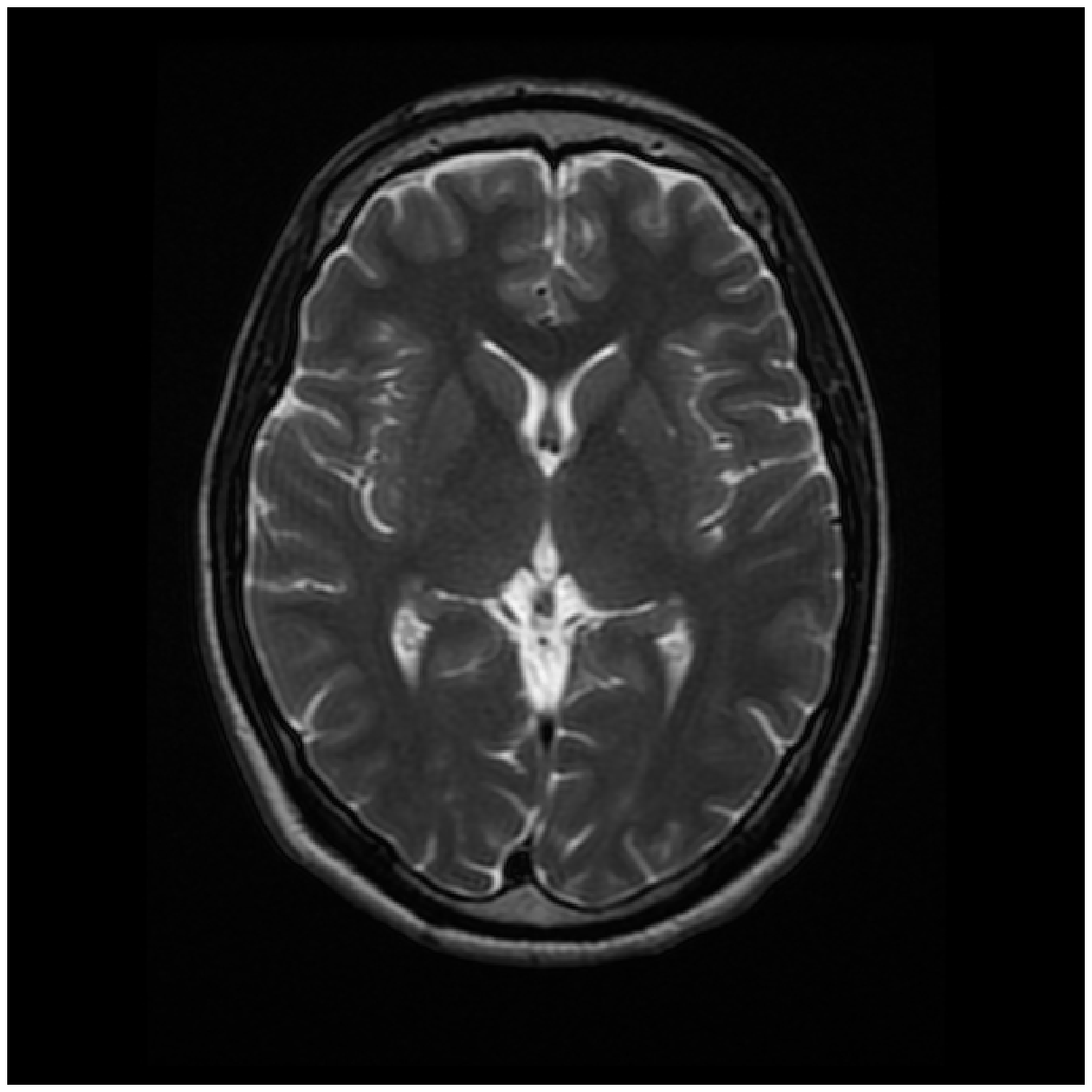}
\label{mri_original}}\hspace{-0.3cm}
\subfloat[]{{\includegraphics[width=1.8in,height=1.4in,trim=-0.2in -0.7in 0.5in 0.8in]{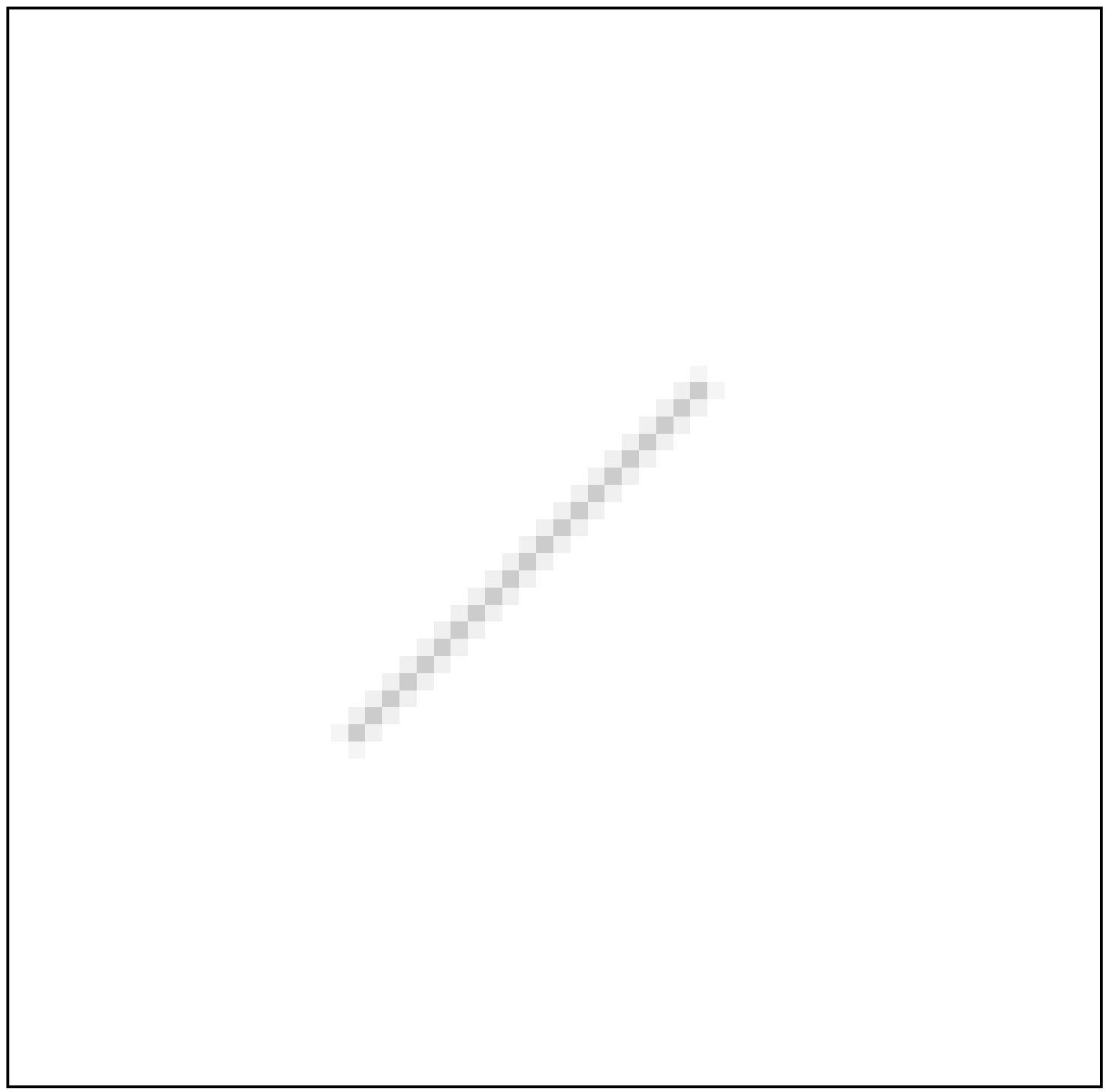}}
\label{kernel}}
\subfloat[]{\includegraphics[width=2.3in]{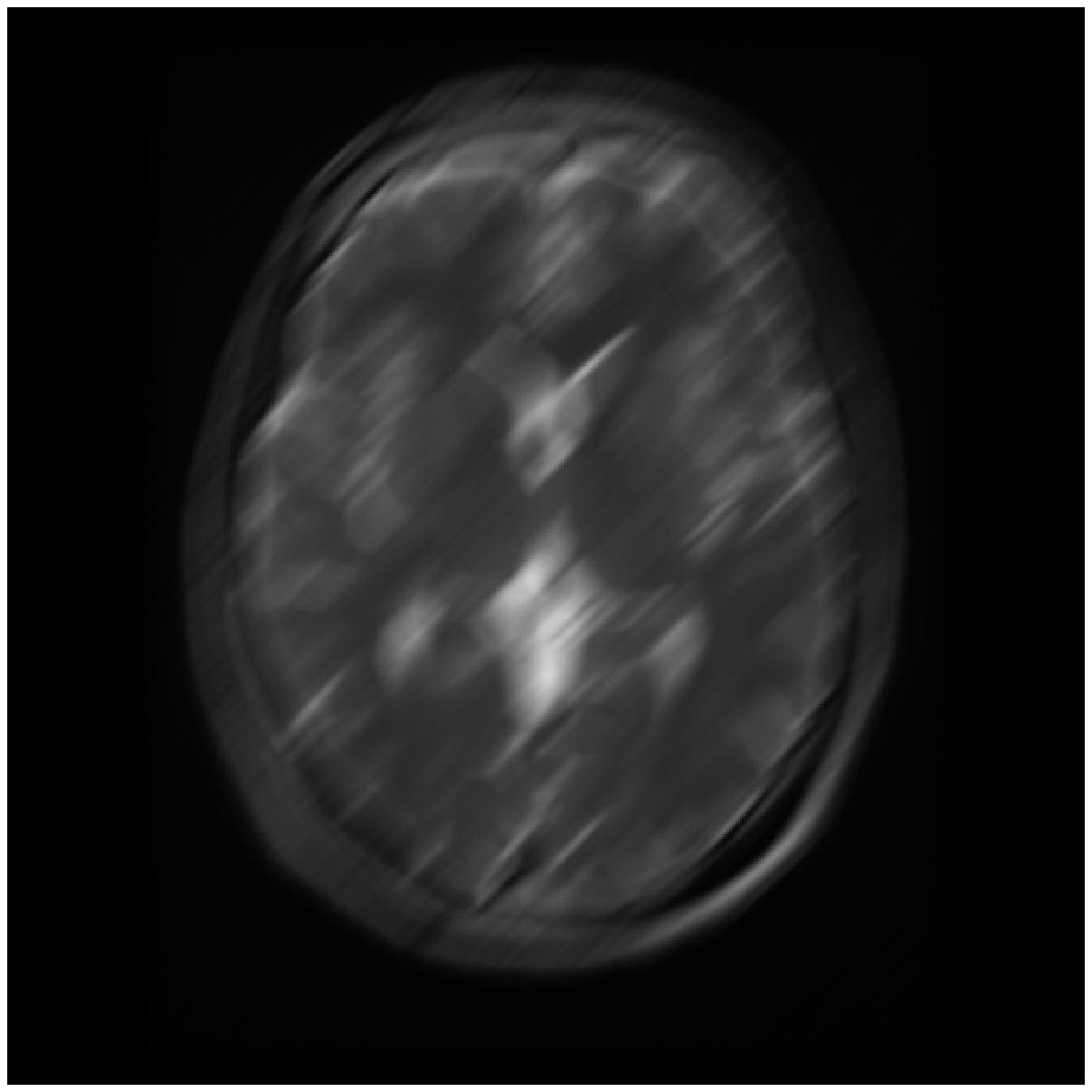}
\label{mri_blurred}}

\subfloat[]{\includegraphics[width=2.3in]{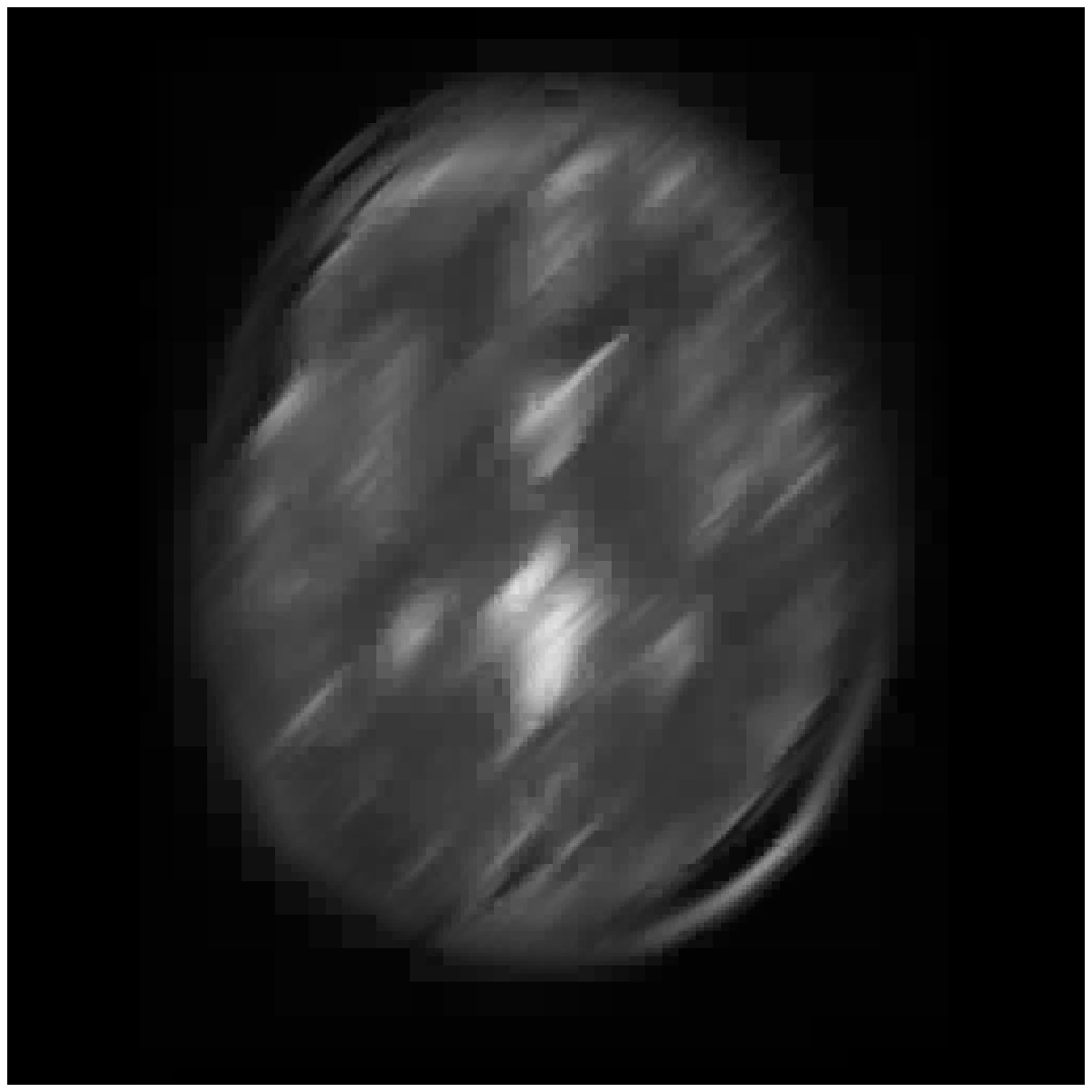}\label{mri_initial_guess}\hspace{-0.6cm}}
\subfloat[]{\includegraphics[width=2.3in]{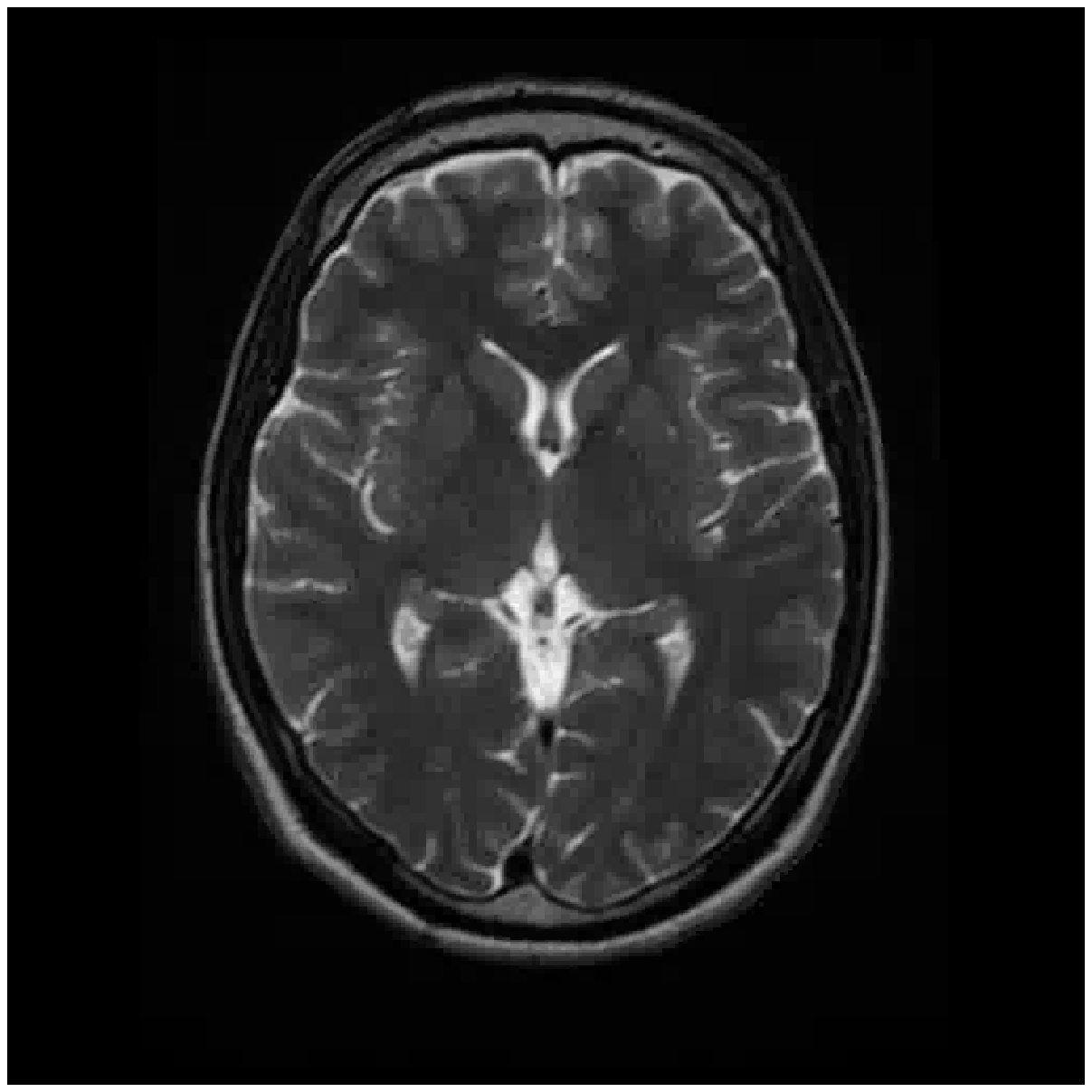}
\label{mri_recovered_both_known}\hspace{-0.8cm}}
\subfloat[]{\includegraphics[width=2.3in]{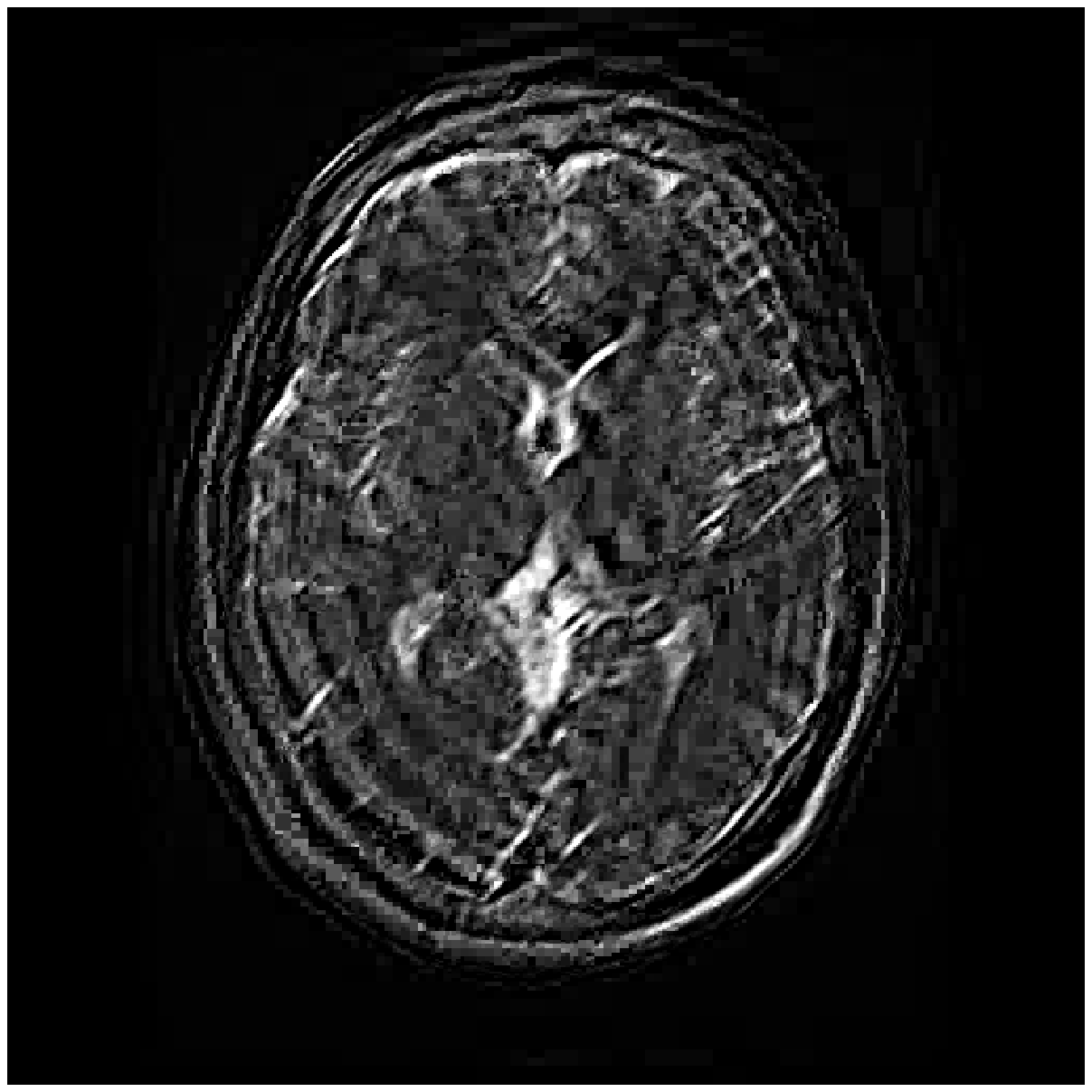}
\label{mri_recovered_both_unknown}}
\caption{MRI image deblurring: (a) Original $512\times 512$ MRI image; (b) Blurring kernel;  (c) Blurred image; (d) Initial guess; (e) Reconstructed image when the subspaces are known; (f) Reconstructed image without knowing the subspaces exactly.}
\label{mri_deblurring}
\end{figure*}

\subsection{An example from image processing}

Next, we test Algorithm~\ref{AGD} on an image deblurring problem, inspired by~\cite{RR12}. The observed image (Figure~\ref{mri_blurred}) is a convolution of a $512\times 512$ MRI image (Figure~\ref{mri_original}) with a motion blurring kernel (Figure~\ref{kernel}). Since the MRI image is approximately sparse in the Haar wavelet basis, we can assume it belongs to a low dimensional subspace; that is, $\bg=\BC\bx_0$, where $\bg\in\mathbb{C}^{L}$ with $L=262,144$  denotes the MRI image reshaped into a vector, $\BC\in\CC^{L\times N}$ represents the wavelet subspace and $\bx_0\in\CC^N$ is the vector of wavelet coefficients. The blurring kernel $\bff\in\CC^{L\times 1}$ is supported on a low frequency region. Therefore $\widehat{\bff}=\BB\bh_0$, where $\BB\in\CC^{L\times K}$ is a reshaped $2$D low frequency Fourier matrix and $\bh_0\in\CC^{K}$ is a short vector. 

Figure~\ref{mri_initial_guess} shows the initial guess for Algorithm~\ref{AGD} in the image domain, which is obtained by running the power method for fifty iterations.
While this initial guess is clearly not a good approximation to the true solution, it suffices as a starting point for gradient descent.
In the first experiment, we take $\BC$ to be the wavelet subspace corresponding to the $N=20000$ largest Haar wavelet coefficients of the original MRI image, and  we also assume the locations of the $K=65$ nonzero entries of the kernel are known. Figure~\ref{mri_recovered_both_known} shows the reconstructed image in this ideal setting. It can be observed that the recovered image is visually indistinguishable from the ground truth MRI image.
 In the second experiment, we test  a more realistic setting, where both the support of the MRI image is the wavelet domain and the support of the kernel are not known. We take the Haar wavelet transform of the blurred image (Figure~\ref{mri_blurred}) and select $\BC$ to be the wavelet subspace corresponding to the $N=35000$ largest wavelet coefficients. We do not assume the exact support of the kernel is known, but assume that its support is contained in a small box region. The reconstructed image in this setting is shown in Figure~\ref{mri_recovered_both_unknown}. Despite not knowing the subspaces exactly, Algorithm~\ref{AGD} is still able to return a reasonable reconstruction.
 
 Yet, this second experiment also demonstrates that there is clearly room for improvement in the case when the subspaces are unknown. One natural
 idea to improve upon the result depicted in Figure~\ref{mri_recovered_both_unknown} is to include an additional total-variation penalty in the reconstruction
 algorithm. We leave this line of work for future research.

\section{Proof of the main theorem}
\label{s:proof}

This section is devoted to the proof of Theorems~\ref{thm:init} and~\ref{thm:main}. Since proving  Theorem~\ref{thm:main} is a bit more involved,
 we briefly describe the architecture of its proof.
In Subsection~\ref{ss:conditions} we state four key conditions:  The {\em Local Regularity Condition} will allow us to show that the objective function decreases; the {\em Local Restricted Isometry Property} enables us to transfer the decrease in the objective function to a decrease of the 
error between the iterates and the true solution; the {\em Local Smoothness Condition} yields the actual rate of convergence, and finally, the
{\em Robustness Condition} establishes robustness of the proposed algorithm against additive noise. Armed with these conditions, 
we will show how the three regions defined in Section~\ref{s:Algo} characterize the convergence neighborhood
of the solution,  i.e., if the initial guess is inside this neighborhood, the sequence generated via gradient descent will always stay inside this neighborhood as well. 
In Subsections~\ref{s:lemmata}--\ref{s:init}  we justify the aforementioned four conditions, show that they are valid under the assumptions stated in Theorem~\ref{thm:main}, and conclude with a proof of Theorem~\ref{thm:init}.

\subsection{Four key conditions and the proof of Theorem~\ref{thm:main}}
\label{ss:conditions}


\begin{condition}[{\textbf{\textit{Local RIP condition}}}]\label{cond:rip}
The following local Restricted Isometry Property (RIP)  for $\A$ holds uniformly for all $(\bh, \bx)\in\Kint:$

\begin{equation}\label{eq:LRIP}
\frac{3}{4} \|\bh\bx^* - \bh_0\bx_0^*\|_F^2 \leq \|\A(\bh\bx^* - \bh_0\bx_0^*)\|^2 \leq \frac{5}{4}\|\bh\bx^* - \bh_0\bx_0^*\|_F^2
\end{equation}
\end{condition}
Condition~\ref{cond:rip} states that $\A$ almost preserves the $\ell_2$-distance between $\bh\bx^* - \bh_0\bx_0^*$ over a ``local" region around the ground truth $\bh_0\bx_0^*$. The proof of Condition~\ref{cond:rip} is given in Lemma~\ref{lem:rip}.

\begin{condition}[{\textbf{\textit{Robustness condition}}}]\label{cond:AE} For the noise $\be\sim\mathcal{N}(\bzero, \frac{\sigma^2 d_0^2}{2L}\I_L) + \mi \mathcal{N}(\bzero, \frac{\sigma^2 d_0^2}{2L}\I_L)$, with high probability there holds, 
\begin{equation}\label{eq:AEnorm}
\| \A^*(\be) \| \leq \frac{\eps d_0}{10\sqrt{2}}
\end{equation}
if $L \geq C_{\gamma}(\frac{\sigma^2}{\eps^2} + \frac{\sigma}{\eps})\max\{K, N\}\log L$. 
\end{condition}
This condition follows directly from~\eqref{eq:Ae-1}. It is quite essential when we analyze the behavior of Algorithm~\ref{AGD} under Gaussian noise. 
With those two conditions above in hand, the lower and upper bounds of $F(\bh, \bx)$  are well approximated over $\Kint$ by two quadratic functions of $\delta$, where $\delta$ is defined in~\eqref{def:delta}. A similar approach towards noisy matrix completion problem can be found in~\cite{KMO09noise}.
For any $(\bh, \bx)\in \Kint$, applying Condition~\ref{cond:rip}  to~\eqref{eq:F-decom} leads to
\begin{equation}\label{eq:Fuppbd-org}
F(\bh, \bx) \leq \|\be\|^2 + F_0(\bh, \bx) + 2\sqrt{2}\|\A^*(\be)\|\delta d_0 \leq   \|\be\|^2 + \frac{5}{4}\delta^2 d_0^2 + 2\sqrt{2} \|\A^*(\be)\| \delta d_0
\end{equation}
and similarly
\begin{equation}\label{eq:Flowbd-org}
F(\bh, \bx) \geq \|\be\|^2 +  \frac{3}{4}\delta^2 d_0^2 - 2\sqrt{2} \|\A^*(\be)\| \delta d_0
\end{equation}
where
\begin{equation*}
|\lag \A^*(\be), \bh\bx^* - \bh_0\bx_0^* \rag| \leq \|\A^*(\be)\| \|\bh\bx^* - \bh_0\bx_0^*\|_* \leq  \sqrt{2}\|\A^*(\be)\| \delta d_0,
\end{equation*}
because $\|\cdot\|$ and $\|\cdot\|_*$ is a pair of dual norm and $\rank(\bh\bx^* - \bh_0\bx_0^*) \leq 2$.
Moreover, with the Condition~\ref{cond:AE}, ~\eqref{eq:Fuppbd-org} and~\eqref{eq:Flowbd-org} yield the followings:
\begin{equation}\label{eq:Fuppbd}
F(\bh, \bx) \leq  \|\be\|^2 + \frac{5}{4}\delta^2d_0^2 + \frac{\eps \delta d_0^2}{5}
\end{equation}
and 
\begin{equation}\label{eq:Flowbd}
F(\bh, \bx) \geq \|\be\|^2 +  \frac{3}{4}\delta^2 d_0^2 - \frac{\eps \delta d_0^2}{5}.
\end{equation}

The third condition is about the regularity condition of $\tF(\bh, \bx)$, which is the key to establishing linear convergence later.  The proof will be given in Lemma~\ref{lem:reg}. 
\begin{condition}[{\textbf{\textit{Local regularity condition}}}]\label{cond:reg}
Let $\tF(\bh, \bx)$ be as defined in~\eqref{def:FG} and $\nabla \tF(\bh, \bx) := (\nabla \tF_{\bh}, \nabla \tF_{\bx})\in\CC^{K + N}$. Then there exists a {\em regularity constant} $\omega = \frac{d_0}{5000}>0$ such that
\begin{equation}\label{eq:reg}
\|\nabla \tF(\bh, \bx)\|^2 \geq \omega \left[ \tF(\bh, \bx) - c\right]_+
\end{equation}
for all $(\bh, \bx) \in \Kint$ where $c = \|\be\|^2 + a \|\A^*(\be)\|^2 $ with $a = 1700$. In particular, in the noiseless case, i.e., $\be = \bzero$, we have
\begin{equation*}
\|\nabla \tF(\bh, \bx)\|^2 \geq \omega \tF(\bh, \bx).
\end{equation*}
\end{condition}

\bigskip

Besides the three regions defined in~\eqref{def:Kd} to~\eqref{def:Kmu}, we define another region $\KF$ via
\begin{equation}
\KF := \left\{(\bh, \bx): \tF(\bh, \bx) \leq \frac{1}{3}\eps^2 d_0^2 + \|\be\|^2\right\}
\end{equation}
for proof technical purposes.
$\KF$ is actually the sublevel set of the nonconvex function $\tF$.

Finally we introduce the last condition called \textit{Local smoothness condition} and its corresponding quantity $C_L$ which characterizes the choice of stepsize $\eta$ and the rate of linear convergence. 
\begin{condition}[{\textbf{\textit{Local smoothness condition}}}]\label{cond:smooth}
Denote $\bz := (\bh, \bx)$. There exists a constant $C_L $ such that
\begin{equation}\label{def:cl}
\|\nabla f(\bz + t \Delta\bz) - \nabla f(\bz)\| \leq C_L t\|\Delta\bz\|, \quad \forall 0\leq t\leq 1,
\end{equation}
for all $\{(\bz, \Delta \bz) | \bz +t \Delta \bz \in \Keps\bigcap \KF, \forall 0\leq t\leq 1\}$, i.e., the whole segment connecting $\bz$ and $\bz + \Delta \bz$, which is parametrized by $t$, belongs to the nonconvex set $\Keps\bigcap \KF.$ 

\end{condition}
The upper bound of $C_L$, which scales with $\mathcal{O}(d_0(1 + \sigma^2)(K + N)\log^2L )$, will be given in Section~\ref{s:smooth}. We will show later in Lemma~\ref{lem:induction} that the stepsize $\eta$ is chosen to be smaller than $\frac{1}{C_L}.$
Hence $\eta = \mathcal{O}((d_0(1 + \sigma^2)(K + N)\log^2L )^{-1}).$
\begin{lemma}
\label{lem:betamu}
There holds $\KF \subset \Kd \cap \Kmu$; under Condition~\ref{cond:rip} and~\ref{cond:AE}, we have $\KF \cap \Keps \subset \MN_{\frac{9}{10}\epsilon}$.
\end{lemma}

\begin{proof}
If $(\vct{h}, \vct{x}) \notin \Kd \cap \Kmu$, by the definition of $G$ in~\eqref{def:G}, at least one component in $G$ exceeds $\rho G_0\left(\frac{2d_0}{d}\right)$. We have 
\begin{eqnarray*}
\tF(\bh, \bx) & \geq & \rho G_0\left(\frac{2d_0}{d}\right) \geq  (d^2 + 2\|\be\|^2) \left( \frac{2d_0}{d} - 1\right)^2 \\
& \geq & (2d_0 - d)^2 + 2\|\be\|^2\left( \frac{2d_0}{d} - 1\right)^2 \\
& \geq & 0.81  d_0^2 + \|\be\|^2 > \frac{1}{3}\eps^2 d_0^2 + \|\be\|^2,
\end{eqnarray*}
where $\rho \geq d^2 + 2\|\be\|^2$ and $0.9d_0 \leq d \leq 1.1d_0.$
This implies $(\bh, \bx) \notin \KF$ and hence $\KF \subset \Kd \cap \Kmu$. \\

For any $(\bh, \bx) \in \KF \cap \Keps$, we have $(\bh, \bx)\in \Kint$ now. By~\eqref{eq:Flowbd},
\begin{equation*}
\|\be\|^2 + \frac{3}{4} \delta^2 d_0^2 - \frac{\eps\delta d_0^2}{5} \leq F(\bh, \bx)\leq \tF(\bh, \bx)\leq \|\be\|^2 + \frac{1}{3}\eps^2 d_0^2.
\end{equation*}
Therefore, $(\vct{h}, \vct{x}) \in \MN_{\frac{9}{10}\epsilon}$ and $\KF \cap \Keps \subset \MN_{\frac{9}{10}\epsilon}$.
\end{proof}

This lemma implies that the intersection of $\KF$ and the boundary of $\Keps$ is empty. One might believe this suggests that $\KF \subset \Keps$. This may not be true. A more reasonable interpretation is that $\KF$ consists of several disconnected regions due to the non-convexity of $\tF(\bh, \bx)$, and one or several of them are contained in $\Keps$.

\begin{lemma}
\label{lem:line_section}
Denote $\vct{z}_1 = (\vct{h}_1, \vct{x}_1)$ and $\vct{z}_2 = (\vct{h}_2, \vct{x}_2)$. Let $\vct{z}(\lambda):=(1-\lambda)\vct{z}_1 + \lambda \vct{z}_2$. If $\vct{z}_1 \in \Keps$ and $\vct{z}(\lambda) \in \KF$ for all $\lambda \in [0, 1]$, we have $\vct{z}_2 \in \Keps$.
\end{lemma}

\begin{proof}
Let us prove the claim by contradiction. If $\vct{z}_2 \notin \Keps$, since $\vct{z}_1 \in \Keps$, there exists $\vct{z}(\lambda_0):=(\vct{h}(\lambda_0), \vct{x}(\lambda_0)) \in \Keps$ for some $\lambda_0 \in [0, 1]$, such that $\|\vct{h}(\lambda_0)\vct{x}(\lambda_0)^* - \bh_0\bx_0^*\|_F = \epsilon d_0$. However, since $\vct{z}(\lambda_0) \in \KF$, by \prettyref{lem:betamu}, we have $\|\vct{h}(\lambda_0)\vct{x}(\lambda_0)^* - \bh_0\bx_0^*\|_F \leq \frac{9}{10}\epsilon d_0$. This leads to a contradiction.
\end{proof}
\begin{remark}
Lemma~\ref{lem:line_section} tells us that if one line segment is completely inside $\KF$ with one end point in $\Keps$, then this whole line segment lies in $\Kint.$
\end{remark}

\begin{lemma}
\label{lem:induction}
Let the stepsize $\eta \leq \frac{1}{C_L}$, $\bz_t : = (\bu_t, \bv_t)\in\CC^{K + N}$ and $C_L$ be the constant defined in~\eqref{def:cl}. Then, as long as $\vct{z}_t \in \Keps \cap \KF$, we have $\vct{z}_{t+1} \in \Keps \cap \KF$ and
\begin{equation}
\label{eq:decreasing2}
\tF(\bz_{t+1}) \leq  \tF(\bz_t) - \eta \|\nabla \tF(\bz_t)\|^2.
\end{equation}
\end{lemma}

\begin{proof}
It suffices to prove \prettyref{eq:decreasing2}. If $\nabla \tF(\bz_t) = \vct{0}$, then $\bz_{t+1} = \bz_t$, which implies \prettyref{eq:decreasing2} directly. So we only consider the case when $\nabla \tF(\bz_t) \neq \vct{0}$. Define the function
\[
\varphi(\lambda) := \tF(\bz_t - \lambda \nabla \tF(\bz_t)). 
\]
Then 
\[
\varphi'(\lambda)|_{\lambda = 0} = - 2\| \nabla \tF(\bz_t) \|^2 <0.
\]
since $\varphi(\lambda)$ is a real-valued function with complex variables (See~\eqref{eq:chain} for details).
By the definition of derivatives, we know there exists $\eta_0 > 0$, such that $\varphi(\lambda) < \varphi(0)$ for all $0<\lambda \leq \eta_0$. Now we will first prove that $\varphi(\lambda) \leq \varphi(0)$ for all $0\leq \lambda\leq \eta$ by contradiction. Assume there exists some $\eta_1 \in (\eta_0, \eta]$ such that $\varphi(\eta_1) > \varphi(0)$. Then there exists $\eta_2 \in (\eta_0, \eta_1)$, such that $\varphi(\eta_2) = \varphi(0)$ and $\varphi(\lambda) < \varphi(0)$ for all $0 < \lambda < \eta_2$, since $\varphi(\lambda)$ is a continuous function. This implies
\[
\bz_t - \lambda \nabla \tF(\bz_t) \in \KF, \quad \forall 0 \leq \lambda \leq \eta_2
\]
since $\tF(\bz_t - \lambda \nabla \tF(\bz_t)) \leq \tF(\bz_t)$ for $0\leq \lambda\leq \eta_2.$
By \prettyref{lem:line_section} and the assumption $\bz_t \in \Keps$, we have
\[
\bz_t - \lambda \nabla \tF(\bz_t) \in \Keps \cap \KF, \quad \forall 0 \leq \lambda \leq \eta_2.
\]
Then, by using the modified descent lemma (Lemma~\ref{lem:DSL}),
\begin{align*}
\tF(\bz_t - \eta_2 \nabla \tF(\bz_t)) &\leq \tF(\bz_t) - 2 \eta_2 \|\nabla \tF(\bz_t)\|^2 + C_L\eta_2^2  \|\nabla \tF(\bz_t)\|^2 \nonumber
\\
& =\tF(\bz_t) + (C_L\eta_2^2 - 2\eta_2) \|\nabla \tF(\bz_t)\|^2 < \tF(\bz_t),
\end{align*}
where the final inequality is due to $\eta_2/\eta <1$, $\eta_2 > \eta_0 \geq 0$ and $\nabla \tF(\bz_t) \neq \vct{0}$. This contradicts $ \tF(\bz_t - \eta_2 \nabla \tF(\bz_t)) = \varphi(\eta_2) = \varphi(0) = \tF(\bz_t)$.

Therefore, there holds $\varphi(\lambda) \leq \varphi(0)$ for all $0 \leq \lambda \leq \eta$. Similarly, we can prove
\[
\bz_t - \lambda \nabla \tF(\bz_t) \in \Keps \cap \KF, \quad \forall 0 \leq \lambda \leq \eta,
\]
which implies $\bz_{t+1} = \bz_t - \eta \nabla \tF(\bz_t) \in \Keps \cap \KF$. Again, by using Lemma~\ref{lem:DSL} we can prove
\[
\tF(\bz_{t+1}) = \tF(\bz_t - \eta \nabla \tF(\bz_t)) \leq \tF(\bz_t) - 2 \eta \|\nabla \tF(\bz_t)\|^2 + C_L\eta^2  \|\nabla \tF(\bz_t)\|^2 \leq \tF(\bz_t)  - \eta \|\nabla \tF(\bz_t)\|^2,
\]
where the final inequality is due to $ \eta \leq \frac{1}{C_L}$.
\end{proof}

We conclude this subsection by proving  Theorem~\ref{thm:main} under the \textit{Local regularity condition}, the
\textit{Local RIP condition}, the \textit{Robustness condition},  and the \textit{Local smoothness condition}. 
The next subsections are devoted to justifying  these conditions and showing that they hold under the assumptions of Theorem~\ref{thm:main}. 

\begin{proof}{[\textbf{of Theorem~\ref{thm:main}}]}
Suppose that the initial guess  $\bz_0 := (\bu_0, \bv_0)\in \frac{1}{\sqrt{3}}\Kd \bigcap \frac{1}{\sqrt{3}}\Kmu \bigcap \MN_{\frac{2}{5}\eps}$, 
we have $G(\bu_0, \bv_0) = 0$. This holds, because
\begin{equation*}
\frac{\|\bu_0\|^2}{2d} \leq \frac{2d_0}{3d} < 1, \quad \frac{L|\bb_l^* \bu_0|^2}{8d\mu^2} \leq \frac{L}{8d\mu^2} \cdot\frac{16d_0\mu^2}{3L} \leq \frac{2d_0}{3d} < 1,
\end{equation*}
where $\|\bu_0\| \leq \frac{2\sqrt{d_0}}{\sqrt{3}}$, $\sqrt{L}\|\BB\bu_0\|_{\infty} \leq \frac{4 \sqrt{d_0}\mu}{\sqrt{3}}$ and $\frac{9}{10}d_0 \leq d\leq \frac{11}{10}d_0.$ Therefore $G_0\left( \frac{\|\bu_0\|^2}{2d}\right) = G_0\left( \frac{\|\bv_0\|^2}{2d}\right) = G_0\left(\frac{L|\bb_l^*\bu_0|^2}{8d\mu^2}\right) = 0$ for all $1\leq l\leq L$ and $G(\bu_0, \bv_0) = 0.$ Since $(\bu_0, \bv_0)\in\Kint$, ~\eqref{eq:Fuppbd} combined with $\delta(\bz_0) := \frac{\|\bu_0\bv_0^* - \bh_0\bx_0^*\|_F}{d_0} \leq \frac{2\eps}{5}$ imply that
\begin{equation*}
\tF(\bu_0, \bv_0) = F(\bu_0, \bv_0) \leq \|\be\|^2 + \frac{5}{4}\delta^2(\bz_0) d_0^2 + \frac{1}{5}\eps\delta(\bz_0) d_0^2 < \frac{1}{3}\eps^2 d_0^2 + \|\be\|^2
\end{equation*}
and hence $\bz_0 = (\bu_0, \bv_0)\in \Keps\bigcap \KF.$ 
Denote $\bz_t : = (\bu_t, \bv_t).$
Combining Lemma~\ref{lem:induction} by choosing $\eta \leq \frac{1}{C_L}$ with Condition~\ref{cond:reg}, we have
\begin{equation*}
\tF(\bz_{t + 1}) \leq \tF(\bz_t) - \eta\omega \left[  \tF(\bz_t)  - c \right]_+
\end{equation*}
with $c = \|\be\|^2 + a\|\A^*(\be)\|^2$, $a = 1700$ and $\bz_t \in \Kint$ for all $t\geq 0.$
Obviously, the inequality above implies 
\begin{equation*}
 \tF(\bz_{t+1})  - c  \leq (1 - \eta\omega)  \left[  \tF(\bz_t)  - c \right]_+ ,
\end{equation*}
and by monotonicity of $z_+ = \frac{z + |z|}{2}$, there holds
\begin{equation*}
 \left[ \tF(\bz_{t+1})  - c\right]_+  \leq (1 - \eta\omega)  \left[  \tF(\bz_t)  - c \right]_+ .
\end{equation*}

Therefore, by induction, we have
\begin{eqnarray*}
\left[ \tF(\bz_t) - c\right]_+ & \leq & \left(1 - \eta\omega\right)^t \left[ \tF(\bz_0) - c\right]_+ \leq \frac{1}{3} (1 - \eta\omega)^{t}  \eps^2 d_0^2
\end{eqnarray*}
where $\tF(\bz_0) \leq \frac{1}{3}\eps^2d_0^2 + \|\be\|^2$ and hence $\left[ \tF(\bz_0) - c \right]_+ \leq   \left[ \frac{1}{3}\eps^2 d_0^2 - a\|\A^*(\be)\|^2 \right]_+ \leq \frac{1}{3}\eps^2 d_0^2.$
Now we can conclude that  $\left[ \tF(\bz_t) - c\right]_+$ converges to $0$ geometrically. Note that over $\Kint$,
\begin{equation*}
\tF(\bz_t) - \|\be\|^2 \geq F_0(\bz_t) - 2\Real(\lag \A^*(\be), \bu_t\bv_t^* - \bh_0\bx_0^* \rag ) 
\geq \frac{3}{4} \delta^2(\bz_t)d_0^2 - 2\sqrt{2}\|\A^*(\be)\| \delta(\bz_t)d_0
\end{equation*}
where $ \delta(\bz_t) := \frac{\|\bu_t\bv_t^* - \bh_0\bx_0^*\|_F}{d_0}$, $F_0$ is defined in~\eqref{def:F0} and $G(\bz_t) \geq 0$.
There holds
\begin{equation*}
\frac{3}{4} \delta^2(\bz_t)d_0^2 - 2\sqrt{2}\|\A^*(\be)\| \delta(\bz_t)d_0  - a\|\A^*(\be)\|^2   \leq \left[ \tF(\bz_t) - c \right]_+ \leq \frac{1}{3}(1 - \eta\omega)^t  \eps^2 d_0^2
\end{equation*}
and equivalently, 
\begin{equation*}
\left|\delta(\bz_t)d_0 - \frac{4\sqrt{2}}{3} \|\A^*(\be)\| \right|^2 \leq  \frac{4}{9} (1 - \eta\omega)^t \eps^2 d_0^2 + \left(\frac{4}{3}a + \frac{32}{9}\right)\|\A^*(\be)\|^2.
\end{equation*}
Solving the inequality above for $\delta(\bz_t)$, we have 
\begin{eqnarray}
\delta(\bz_t) d_0 & \leq &  \frac{2}{3}(1 - \eta\omega)^{t/2}  \eps d_0 +\left(\frac{4\sqrt{2}}{3} + \sqrt{\frac{4}{3}a + \frac{32}{9}} \right)\|\A^*(\be)\| \nonumber \\
& \leq &  \frac{2}{3}(1 - \eta\omega)^{t/2}\eps d_0 + 50 \|\A^*(\be)\|. \label{eq:main-res-2}
\end{eqnarray}
Let $d_t : = \|\bu_t\|\|\bv_t\|$, $t\geq 1.$
By~\eqref{eq:main-res-2} and triangle inequality, we immediately conclude that
\begin{equation*}
|d_t - d_0| \leq  \frac{2}{3}(1 - \eta\omega)^{t/2}\eps d_0 + 50 \|\A^*(\be)\|.
\end{equation*}

Now we derive the upper bound for $\sin \angle (\bu_t, \bh_0)$ and $\sin \angle(\bv_t, \bx_0).$ Due to symmetry, it suffices to consider $\sin \angle (\bu_t, \bh_0)$. The bound follows from standard linear algebra arguments:
\begin{eqnarray*}
\sin \angle(\bu_t, \bh_0) & = & \frac{1}{\|\bu_t\|}\left\| \left(\I  - \frac{\bh_0\bh_0^*}{d_0}\right)\bu_t\right\| \\
& = & \frac{1}{\|\bu_t\| \|\bv_t\|}\left\| \left(\I  - \frac{\bh_0\bh_0^*}{d_0}\right)(\bu_t \bv_t^* - \bh_0\bx_0^* )\right\|  \\
& \leq & \frac{1}{d_t} \|\bu_t\bv_t^* - \bh_0\bx_0^*\|_F \\
& \leq & \frac{1}{d_t}\left(  \frac{2}{3}(1 - \eta\omega)^{t/2}\eps d_0 + 50 \|\A^*(\be)\|\right),
\end{eqnarray*}
where the second equality uses $\left(\I  - \frac{\bh_0\bh_0^*}{d_0}\right) \bh_0 = \bzero.$
\end{proof}

\subsection{Supporting lemmata}
\label{s:lemmata}

This subsection introduces several lemmata, especially Lemma~\ref{lem:orth_decomp}, ~\ref{lem:ripu} and~\ref{lem:key}, which are central for  justifying Conditions~\ref{cond:rip} and~\ref{cond:reg}. After that, we will prove the \textit{Local RIP Condition} in Lemma~\ref{lem:rip} based on those three lemmata. We start with defining a linear space $T$, which contains $\bh_0\bx_0^*$, via
\begin{equation}\label{def:T}
T := \lc \frac{1}{\sqrt{d_0}}\bh_0\bv^* + \frac{1}{\sqrt{d_0}}\bu\bx_0^*,~ \bu \in\CC^K, \bv\in\CC^N \rc \subset \CC^{K\times N}.
\end{equation}
Its orthogonal complement is given by
\begin{equation*}
\TB : = \lc \left(\I - \frac{1}{d_0}\bh_0\bh_0^*\right)\BZ\left(\I - \frac{1}{d_0}\bx_0\bx_0^*\right),~ \BZ\in\CC^{K\times N} \rc.
\end{equation*}
Denote $\PP_T$ to be the projection operator from $\CC^{K\times N}$ onto $T$. 

For any $\vct{h}$ and $\vct{x}$, there are unique orthogonal decompositions
\begin{equation}
\label{eq:orth}
{\vct{h} = \alpha_1 \vct{h}_0 + \tilde{\vct{h}}, \quad \vct{x} = \alpha_2 \vct{x}_0 + \tilde{\vct{x}}},
\end{equation}
where $\vct{h}_0 \perp \tilde{\vct{h}}$ and $\vct{x}_0 \perp \tilde{\vct{x}}$. More precisely, $
\alpha_1 = \frac{\bh_0^*\bh}{d_0} = \frac{\lag \bh_0, \bh\rag}{d_0}$ and $\alpha_2 =  \frac{\lag \bx_0, \bx\rag}{d_0}.$
 We thereby have the following matrix orthogonal decomposition
\begin{align}\label{eq:decomposition}
\vct{h}\vct{x}^* - \vct{h}_0 \vct{x}_0^* = (\alpha_1 \overline{\alpha_2} - 1)\vct{h}_0\vct{x}_0^* + \overline{\alpha_2} \tilde{\vct{h}} \vct{x}_0^* + \alpha_1 \vct{h}_0 \tilde{\vct{x}}^*  + \tilde{\vct{h}} \tilde{\vct{x}}^*
\end{align}
where the first three components are in $T$ while $\tilde{\bh}\tilde{\bx}^*\in T^{\bot}$.

\begin{lemma}
\label{lem:orth_decomp}
Recall that $\|\vct{h}_0\| = \|\vct{x}_0\| = \sqrt{d_0}$. If $\delta := \frac{\|\vct{h}\vct{x}^* - \vct{h}_0 \vct{x}_0^*\|_F}{d_0}<1$, we have the following useful bounds
\[
|\alpha_1|\leq \frac{\|\vct{h}\|}{\|\vct{h}_0\|}, \quad |\alpha_1\overline{\alpha_2} - 1|\leq \delta,
\]
and
\[
\|\tilde{\vct{h}}\| \leq \frac{\delta}{1 - \delta}\|\vct{h}\|,\quad \|\tilde{\vct{x}}\| \leq \frac{\delta}{1 - \delta}\|\vct{x}\|,\quad \|\tilde{\vct{h}}\| \|\tilde{\vct{x}}\| \leq \frac{\delta^2}{2(1 - \delta)} d_0.
\]
Moreover, if $\|\vct{h}\| \leq 2\sqrt{d_0}$ and $\sqrt{L}\|\mtx{B} \vct{h}\|_\infty \leq 4\mu \sqrt{d_0}$, we have $\sqrt{L}\|\mtx{B} \tilde{\vct{h}}\|_\infty \leq 6 \mu \sqrt{d_0}$.
\end{lemma}
\begin{remark}
This lemma is actually a simple version of singular value/vector perturbation. It says that if $\frac{\|\bh\bx^* - \bh_0\bx_0^*\|_F}{d_0}$ is of $\mathcal{O}(\delta)$, then the individual vectors $(\bh, \bx)$ are also close to $(\bh_0, \bx_0)$, with the error of order $\mathcal{O}(\delta).$
\end{remark}

\begin{proof}
The equality \prettyref{eq:orth} implies that $\|\alpha_1 \vct{h}_0\| \leq \|\vct{h}\|$, so there holds $|\alpha_1|\leq \frac{\|\vct{h}\|}{\|\vct{h}_0\|}$. Since $\|\vct{h}\vct{x}^* - \vct{h}_0 \vct{x}_0^*\|_F = \delta d_0$, by \prettyref{eq:decomposition}, we have
\begin{equation}
\label{eq:pythag}
\delta^2d_0^2 = (\alpha_1 \overline{\alpha_2} -1)^2 d_0^2 +  |\overline{\alpha_2}|^2 \|\tilde{\vct{h}}\|^2d_0 + |\alpha_1|^2 \|\tilde{\vct{x}}\|^2d_0+ \|\tilde{\vct{h}}\|^2\|\tilde{\vct{x}}\|^2.
\end{equation}
This implies that
\[
\|\vct{h}\|^2\|\tilde{\vct{x}}\|^2 = (\alpha_1^2d_0 + \|\tilde{\vct{h}}\|^2)\|\tilde{\vct{x}}\|^2 \leq  \delta^2d_0^2.
\]
On the other hand,
\[
\|\vct{h}\|\|\vct{x}\| \geq \|\vct{h}_0\|\|\vct{x}_0\|- \|\vct{h}\vct{x}^* - \vct{h}_0 \vct{x}_0^*\|_F = (1 - \delta) d_0.
\]
The above two inequalities imply that $\|\tilde{\vct{x}}\| \leq \frac{\delta}{1 - \delta}\|\vct{x}\|$, and similarly we have $\|\tilde{\vct{h}}\| \leq \frac{\delta}{1 - \delta}\|\vct{h}\|$. The equality \prettyref{eq:pythag} implies that $|\alpha_1 \overline{\alpha_2} -1| \leq \delta$ and hence $|\alpha_1 \overline{\alpha_2}| \geq 1 - \delta$. Moreover, \prettyref{eq:pythag} also implies
\[
\|\tilde{\vct{h}}\|\|\tilde{\vct{x}}\| |\alpha_1| |\overline{\alpha_2}| \leq \frac{1}{2}( |\overline{\alpha_2}|^2 \|\tilde{\vct{h}}\|^2 + |\alpha_1|^2\|\tilde{\vct{x}}\|^2) \leq \frac{\delta^2 d_0}{2},
\] 
which yields $\|\tilde{\vct{h}}\|_2 \|\tilde{\vct{x}}\|_2 \leq \frac{\delta^2}{2(1 - \delta)} d_0$.\\
~\\
If $\|\vct{h}\| \leq 2\sqrt{d_0}$ and $\sqrt{L}\|\mtx{B} \vct{h}\|_\infty \leq 4\mu \sqrt{d_0}$, there holds $|\alpha_1| \leq \frac{\|\vct{h}\|}{\|\vct{h}_0\|} \leq 2$. Then
\begin{eqnarray*}
\sqrt{L}\|\mtx{B} \tilde{\vct{h}}\|_\infty & \leq & \sqrt{L}\|\mtx{B} \vct{h}\|_\infty + \sqrt{L}\|\mtx{B}(\alpha_1\vct{h}_0)\|_\infty \leq \sqrt{L}\|\mtx{B} \vct{h}\|_\infty + 2\sqrt{L}\|\mtx{B}\vct{h}_0\|_\infty  \\
& \leq & 4 \mu \sqrt{d_0}  + 2 \mu_h \sqrt{d_0}  \leq 6\mu \sqrt{d_0}
\end{eqnarray*}
where $\mu_h \leq \mu$.
\end{proof}

In the following, we introduce and prove a series of local and global properties of $\A$:

\begin{lemma}[Lemma 1 in \cite{RR12}]
\label{lem:A-UPBD}
For $\A$ defined in~\eqref{def:A},
\begin{equation}\label{eq:A-UPBD}
\|\A\| \leq \sqrt{N\log(NL/2) + \gamma \log L}
\end{equation}
with probability at least $1 - L^{-\gamma}.$
\end{lemma}

\begin{lemma}[Corollary 2 in~\cite{RR12}]
\label{lem:ripu}
Let $\A$ be the operator defined in~\eqref{def:A}, then on an event $E_1$ with probability at least $1 - L^{-\gamma}$, $\A$ restricted on $T$ is well-conditioned, i.e., 
\begin{equation*}
\|\PP_T\A^*\A\PP_T - \PP_T\| \leq \frac{1}{10}
\end{equation*}
where $\PP_T$ is the projection operator from $\CC^{K\times N}$ onto $T$, provided $L \geq C_{\gamma} \max\{K, \mu_h^2 N\}\log^2(L)$. 
\end{lemma}
Now we introduce a property of $\A$ when restricted on rank-one matrices.
\begin{lemma}
\label{lem:key}
On an event $E_2$ with probability at least $1 - L^{-\gamma} - \frac{1}{\gamma}\exp(-(K+N))$, we have
\begin{equation*}
\|\A(\vct{u}\vct{v}^*)\|^2 \leq \left(\frac{4}{3}\|\vct{u}\|^2 + 2\|\vct{u}\|\|\mtx{B}\vct{u}\|_\infty \sqrt{2(K+N)\log L} + 8\|\mtx{B}\vct{u}\|_\infty^2(K+N) \log L \right)\|\vct{v}\|^2,
\end{equation*}
uniformly for any $\vct{u}$ and $\vct{v}$, provided $L\geq C_{\gamma}(K+N)\log L$.
\end{lemma}
\begin{proof}
Due to the homogeneity, without loss of generality we can assume $\|\vct{u}\|=\|\vct{v}\|=1$. Define
\[
f(\vct{u}, \vct{v}):=\|\A(\vct{u}\vct{v}^*)\|^2 - 2\|\mtx{B}\vct{u}\|_\infty \sqrt{2(K+N)\log L} - 8\|\mtx{B}\vct{u}\|_\infty^2(K+N)\log L.
\]
It suffices to prove that $f(\vct{u}, \vct{v})\leq \frac{4}{3}$ uniformly for all $(\vct{u}, \vct{v}) \in \MS^{K-1}\times \MS^{N-1}$ with high probability, where $\mathcal{S}^{K-1}$ is the unit sphere in $\CC^K.$ For fixed $(\vct{u}, \vct{v}) \in \MS^{K-1}\times \MS^{N-1}$, notice that
\begin{equation*}
\|\A(\vct{u}\vct{v}^*)\|^2 = \sum\limits_{l=1}^L |\vct{b}_l^* \vct{u}|^2 |\vct{a}_l^* \vct{v}|^2
\end{equation*}
is the sum of subexponential variables with expectation $\mathbb{E} \|\A(\vct{u}\vct{v}^*)\|^2 =  \sum\limits_{l=1}^L |\vct{b}_l^* \vct{u}|^2 = 1$. For any generalized $\chi_n^2$ variable $Y \sim \sum_{i=1}^n c_i \xi_i^2$ satisfies
\begin{equation}\label{ineq:bern}
\mathbb{P}(Y - \mathbb{E}(Y) \geq t) \leq \exp\left(- \frac{t^2}{8\|\vct{c}\|_2^2}\right) \vee \exp\left(- \frac{t}{8\|\vct{c}\|_\infty}\right), 
\end{equation}
 where $\{\xi_i\}$ are i.i.d. $\chi^2_1$ random variables and $\bc = (c_1, \cdots, c_n)^T\in \RR^n$. Here we set $|\ba_l^*\bv|^2 = \frac{1}{2} \xi_{2l-1}^2 + \frac{1}{2}\xi_{2l}^2$, $c_{2l-1} = c_{2l} = \frac{|\bb_l^*\bu|^2}{2}$ and $n = 2L$.
Therefore, 
\begin{equation*}
\|\bc\|_{\infty} = \frac{\|\BB\bu\|_{\infty}^2}{2}, \quad \|\bc\|_2^2 = \frac{1}{2}\sum_{l=1}^L |\bb_l^*\bu|^4 \leq \frac{\|\BB\bu\|^2_{\infty}}{2}
\end{equation*}
and we have
\begin{equation*}
\mathbb{P}(\|\A(\vct{u}\vct{v}^*)\|^2 \geq 1 + t)\leq  \exp\left(- \frac{t^2}{4 \|\mtx{B}\vct{u}\|_\infty^2}\right) \vee \exp\left(- \frac{t}{4\|\mtx{B}\vct{u}\|_\infty^2}\right).
\end{equation*}
Applying~\eqref{ineq:bern} and setting 
\[
t = g(\vct{u}):= 2\|\mtx{B}\vct{u}\|_\infty \sqrt{2(K+N)\log L} + 8\|\BB\bu\|^2_{\infty}(K + N)\log L,
\]  
there holds
\[
\mathbb{P}\left(\|\A(\vct{u}\vct{v}^*)\|^2 \geq 1 + g(\vct{u})\right) \leq \exp\left( - 2 (K+N)(\log L) \right).
\]
That is, $f(\vct{u}, \vct{v}) \leq 1$ with probability at least $1 - \exp\left( - 2 (K+N)(\log L) \right)$.
We define $\mathcal{K}$ and $\mathcal{N}$ as $\epsilon_0$-nets of $\MS^{K-1}$ and $\MS^{N-1}$, respectively. Then, $|\mathcal{K}|\leq (1+\frac{2}{\epsilon_0})^{2K}$ and $|\mathcal{N}|\leq (1+\frac{2}{\epsilon_0})^{2N}$ follow from the covering numbers of the sphere (Lemma 5.2 in~\cite{Ver10}).

By taking the union bounds over $\mathcal{K}\times \mathcal{N},$ we have $f(\vct{u}, \vct{v})\leq 1$ holds uniformly for all $(\vct{u}, \vct{v}) \in \mathcal{K} \times \mathcal{N}$ with probability at least 
\[
1- \left(1+\frac{2}{\epsilon_0}\right)^{2(K + N)} e^{ - 2 (K+N)\log L } = 1- \exp\left( -2(K + N)\left(\log L - \log \left(1 + \frac{2}{\epsilon_0}\right)\right) \right).
\] 
Our goal is to show that $f(\vct{u}, \vct{v})\leq \frac{4}{3}$ uniformly for all $(\vct{u}, \vct{v}) \in \MS^{K-1}\times \MS^{N-1}$ with the same probability. For any $(\vct{u}, \vct{v}) \in \MS^{K-1}\times \MS^{N-1}$, we can find its closest $(\vct{u}_0, \vct{v}_0) \in \mathcal{K} \times \mathcal{N}$ satisfying $\| \bu - \bu_0 \| \leq \eps_0$ and $\|\bv - \bv_0\| \leq \eps_0$. By \prettyref{lem:A-UPBD}, with probability at least $1 - L^{-\gamma}$, we have $\|\A\| \leq \sqrt{(N+\gamma)\log L}$. Then straightforward calculation gives
\begin{eqnarray*}
 | f(\vct{u}, \vct{v}) - f(\vct{u}_0, \vct{v})| & \leq &
\| \A((\vct{u} - \vct{u}_0)\vct{v}^*)\|\|\A((\vct{u} + \vct{u}_0)\vct{v}^*)\| \\
&& + 2\|\mtx{B}(\vct{u} - \vct{u}_0)\|_\infty \sqrt{2(K+N)\log L}\\
&& + 8(K+N) (\log L) \|\mtx{B}(\vct{u} - \vct{u}_0)\|_\infty(\|\mtx{B}\vct{u}\|_\infty + \|\mtx{B}\vct{u}_0\|_\infty) \\
& \leq & 2\|\A\|^2 \eps_0 + 2\sqrt{2(K+N)\log L}\eps_0 + 16(K + N)(\log L) \eps_0 \\
&\leq & (21N + 19K + 2\gamma )(\log L) \epsilon_0
\end{eqnarray*}
where the first inequality is due to $||z_1|^2 - |z_2|^2| \leq |z_1 - z_2||z_1 + z_2|$ for any $z_1, z_2 \in \mathbb{C}$, and the second inequality is due to $\|\BB \vct{z} \|_{\infty} \leq \|\BB \vct{z} \| = \|\vct{z}\|$ for any $\vct{z} \in \mathbb{C}^K$. Similarly,
\begin{eqnarray*}
|f(\vct{u}_0, \vct{v}) - f(\vct{u}_0, \vct{v}_0)| & = & \| \A( \vct{u}_0 (\vct{v} + \vct{v}_0)^*)\|\| \A( \vct{u}_0 (\vct{v} - \vct{v}_0)^*)\| \\
& \leq & 2\|\A\|^2\eps_0 \leq 2( N + \gamma)(\log L)\epsilon_0 .
\end{eqnarray*}
Therefore, if $\epsilon_0 = \frac{1}{70(N + K + \gamma)\log L}$, there holds
\[
|f(\vct{u}_0, \vct{v}) - f(\vct{u}_0, \vct{v}_0)| \leq \frac{1}{3}.
\]
Therefore, if $L \geq C_{\gamma}(K+N)\log L$ with $C_{\gamma}$ reasonably large and $\gamma \geq 1$, we have $\log L - \log\left(1 + \frac{2}{\eps_0}\right) \geq \frac{1}{2}(1 + \log(\gamma))$ and
$f(\vct{u}, \vct{v})\leq \frac{4}{3}$ uniformly for all $(\vct{u}, \vct{v}) \in \MS^{K-1}\times \MS^{N-1}$ with probability at least $1- L^{-\gamma} - \frac{1}{\gamma}\exp(-(K+N))$.
\end{proof}

 Finally, we introduce a local RIP property of $\A$ conditioned on the event $E_1\cap E_2$, where $E_1$ and $E_2$ are defined in Lemma~\ref{lem:ripu} and Lemma~\ref{lem:key}
\begin{lemma}
\label{lem:rip}
Over $\Kint$ with $\mu \geq \mu_h$ and $\epsilon \leq \frac{1}{15}$, the following RIP type of property holds for $\A$:
\begin{equation*}
\frac{3}{4} \|\bh\bx^* - \bh_0\bx_0^*\|_F^2 \leq \|\A(\bh\bx^* - \bh_0\bx_0^*)\|^2 \leq \frac{5}{4}\|\bh\bx^* - \bh_0\bx_0^*\|_F^2
\end{equation*}
provided $L \geq C\mu^2 (K+N)\log^2 L$ for some numerical constant $C$ and conditioned on $E_1\bigcap E_2.$
\end{lemma}

\begin{proof}
Let $\delta : = \frac{\|\bh\bx^* - \bh_0\bx_0^*\|_F}{d_0} \leq \epsilon \leq \frac{1}{15}$, and 
\[
\bh\bx^* - \bh_0\bx_0^* := \BU + \BV.
\]
where 
\begin{equation}
\label{eq:UV}
\BU = (\alpha_1 \overline{\alpha_2} - 1)\vct{h}_0\vct{x}_0^* + \overline{\alpha_2} \tilde{\vct{h}} \vct{x}_0^* + \alpha_1 \vct{h}_0 \tilde{\vct{x}}^* \in T,\quad \BV = \tilde{\vct{h}} \tilde{\vct{x}}^* \in  T^\perp.
\end{equation}
By \prettyref{lem:orth_decomp}, we have $\|\BV\|_F \leq \frac{\delta^2}{2(1 - \delta)} d_0$ and hence 
\[
\left(\delta - \frac{\delta^2}{2(1 - \delta)}\right)d_0 \leq \|\BU \|_F \leq \left(\delta + \frac{\delta^2}{2(1 - \delta)}\right)d_0.
\]
Since $\BU \in T$, by \prettyref{lem:ripu}, we have
\begin{equation}
\label{eq:AU}
\sqrt{\frac{9}{10}}\left(\delta - \frac{\delta^2}{2(1 - \delta)}\right)d_0 \leq \|\A(\BU) \| \leq \sqrt{\frac{11}{10}}\left(\delta + \frac{\delta^2}{2(1 - \delta)}\right)d_0.
\end{equation}
By \prettyref{lem:key}, we have
\begin{equation}\label{ineq:AV}
\|\A(\BV)\|^2 \leq \left(\frac{4}{3}\|\tilde{\vct{h}}\|_2^2 + 2\|\tilde{\vct{h}}\|\|\mtx{B}\tilde{\vct{h}}\|_\infty \sqrt{2(K+N)\log L} + 8\|\mtx{B}\tilde{\vct{h}}\|_\infty^2(K+N) (\log L)\right) \|\tilde{\vct{x}}\|^2. 
\end{equation}
By \prettyref{lem:orth_decomp}, we have $\|\tilde{\vct{h}}\| \|\tilde{\vct{x}}\| \leq \frac{\delta^2}{2(1 - \delta)} d_0$, $\|\tilde{\vct{x}}\| \leq \frac{\delta}{1 - \delta}\|\vct{x}\| \leq \frac{2\delta}{1 - \delta} \sqrt{d_0}$, $\|\tilde{\vct{h}}\| \leq \frac{\delta}{1 - \delta}\|\vct{h}\| \leq \frac{2\delta}{1 - \delta} \sqrt{d_0}$, and $\sqrt{L}\|\mtx{B} \tilde{\vct{h}}\|_\infty \leq 6 \mu \sqrt{d_0}$.
By substituting all those estimations into~\eqref{ineq:AV}, it ends up with
\begin{equation}
\label{eq:AV}
\|\A(\BV)\|^2 \leq \frac{\delta^4}{3(1-\delta)^2} d_0^2 + C'\left(\frac{\delta^3}{\sqrt{C\log L}}  + \frac{\delta^2}{C\log L} \right)d_0^2,
\end{equation}
where $C'$ is a numerical constant and $L \geq C\mu^2 (K + N)\log^2 L$. Combining \prettyref{eq:AV} and \prettyref{eq:AU} together with $C$ sufficiently large, numerical computation gives
\[
\frac{3}{4}\delta d_0 \leq \|\A(\BU)\| -  \|\A(\BV)\| \leq \|\A(\BU + \BV)\| \leq \|\A(\BU)\| + \|\A(\BV)\| \leq \frac{5}{4}\delta d_0.
\]
for all $(\bh, \bx)\in \Kint$ given $\eps \leq \frac{1}{15}$.
\end{proof}


\subsection{Local regularity} 
\label{s:LRC}
In this subsection, we will prove Condition~\ref{cond:reg}. Throughout this section, we assume $E_1$ and $E_2$ holds where $E_1$ and $E_2$ are mentioned in Lemma~\ref{lem:ripu} and Lemma~\ref{lem:key}. 
For all $(\vct{h}, \vct{x}) \in \Kd \cap \Keps$, consider $\alpha_1, \alpha_2, \tilde{\vct{h}}$ and $\tilde{\vct{x}}$ defined in \prettyref{eq:orth} and let
\begin{equation*}
\Dh = \bh - \alpha \bh_0, \quad \Dx = \bx - \overline{\alpha}^{-1}\bx_0.
\end{equation*}
where 
\begin{equation*}
\alpha (\vct{h}, \vct{x})= 
\begin{cases} 
(1 - \delta_0)\alpha_1, & \text{~if~} \|\vct{h}\|_2 \geq \|\vct{x}\|_2 \\ 
\frac{1}{(1 - \delta_0)\overline{\alpha_2}}, & \text{~if~} \|\vct{h}\|_2 < \|\vct{x}\|_2\end{cases}
\end{equation*}
with $\delta_0 := \frac{\delta}{10}$. The particular form of  $\alpha(\bh, \bx)$ serves primarily for proving the local regularity condition of $G(\bh, \bx)$, which will be evident in Lemma~\ref{lem:regG}.
The following lemma gives bounds of $\Dx$ and $\Dh$.\\

\begin{lemma}
\label{lem:DxDh}
For all $(\vct{h}, \vct{x}) \in \Kd \cap \Keps$ with $\epsilon \leq \frac{1}{15}$, there holds $\|\Dh\|_2^2 \leq 6.1 \delta^2 d_0$, $\|\Dx\|_2^2 \leq 6.1 \delta^2 d_0$, and $\|\Dh\|_2^2 \|\Dx\|_2^2 \leq 8.4 \delta^4 d_0^2$. 
Moreover, if we assume $(\vct{h}, \vct{x}) \in \Kmu$ additionally, we have $ \sqrt{L}\|\mtx{B}(\Dh)\|_\infty \leq 6\mu\sqrt{d_0}$.
\end{lemma}

\begin{proof}
We first prove that $\|\Dh\|_2^2 \leq 6.1 \delta^2 d_0$, $\|\Dx\|_2^2 \leq 6.1 \delta^2 d_0$, and $\|\Dh\|_2^2 \|\Dx\|_2^2 \leq 8.4 \delta^4 d_0^2$:\\
Case 1: $\|\vct{h}\|_2 \geq \|\vct{x}\|_2$ and $\alpha = (1 - \delta_0) \alpha_1$. In this case, we have
\[
\Dh = \tilde{\vct{h}} + \delta_0 \alpha_1 \vct{h}_0,\quad \Dx = \vct{x} - \frac{1}{(1 - \delta_0)\overline{\alpha}_1} \vct{x}_0 = \left(\alpha_2 - \frac{1}{(1 - \delta_0)\overline{\alpha}_1}\right)\vct{x}_0 + \tilde{\vct{x}}. 
\]
First, notice that $\|\vct{h}\|_2^2 \leq 4d_0$ and $\|\alpha_1 \vct{h}_0\|_2^2\leq \|\vct{h}\|_2^2$. By \prettyref{lem:orth_decomp}, we have 
\begin{equation}\label{eq:Dh-est}
\|\Dh\|_2^2 = \|\tilde{\vct{h}}\|_2^2 + \delta_0^2\|\alpha_1 \vct{h}_0\|_2^2 \leq \left(\left(\frac{\delta}{1-\delta}\right)^2 + \delta_0^2\right)\|\vct{h}\|_2^2 \leq 4.7 \delta^2 d_0.
\end{equation}
Secondly, we estimate $\|\Delta \bx\|.$
Note that $\|\vct{h}\|_2 \|\vct{x}\|_2 \leq (1+\delta)d_0$. By $\|\vct{h}\|_2 \geq \|\vct{x}\|_2$, we have $\|\vct{x}\|_2 \leq \sqrt{(1+\delta)d_0}$. By $|\alpha_2| \|\vct{x}_0\|_2 \leq \|\vct{x}\|_2$, we get $|\alpha_2| \leq \sqrt{1+\delta}$. By \prettyref{lem:orth_decomp}, we have $|\overline{\alpha_1} \alpha_2 -1|=|\alpha_1 \overline{\alpha_2} -1|\leq \delta$, so 
\[
\left|\alpha_2 - \frac{1}{(1 - \delta_0)\overline{\alpha_1}}\right| = |\alpha_2| \left|\frac{(1 - \delta_0)(\overline{\alpha_1} \alpha_2- 1) - \delta_0}{(1 - \delta_0)\overline{\alpha_1} \alpha_2}\right|  \leq \frac{\delta \sqrt{1+ \delta}}{1 - \delta} + \frac{\sqrt{1+\delta}\delta_0}{(1 - \delta_0)(1 - \delta)}\leq 1.22 \delta
\]
where $|\overline{\alpha_1}\alpha_2| \leq \frac{1}{1 - \delta}.$
Moreover, by \prettyref{lem:orth_decomp}, we have $\|\tilde{\vct{x}}\|_2 \leq \frac{\delta}{1-\delta}\|\vct{x}\|_2 \leq \frac{2\delta}{1-\delta} \sqrt{d_0}$. Then we have
\begin{equation}\label{eq:Dx-est}
\|\Dx\|_2^2 = \left|\alpha_2 - \frac{1}{(1 - \delta_0)\overline{\alpha}_1}\right|^2 d_0 + \|\tilde{\vct{x}}\|_2^2 \leq \left(1.22^2+ \frac{4}{(1 - \delta)^2}\right) \delta^2  d_0 \leq 6.1 \delta^2 d_0.
\end{equation}
Finally, \prettyref{lem:orth_decomp} gives $\|\tilde{\vct{h}}\|_2 \|\tilde{\vct{x}}\|_2 \leq \frac{\delta^2}{2(1 - \delta)} d_0$ and  $|\alpha_1| \leq 2$. Combining~\eqref{eq:Dh-est} and~\eqref{eq:Dx-est}, we have
\begin{align*}
\|\Dh\|_2^2 \|\Dx\|_2^2 &\leq \|\tilde{\bh}\|_2^2\|\tilde{\bx}\|_2^2 + \delta_0^2 |\alpha_1|^2 \|\vct{h}_0\|_2^2 \|\Dx\|_2^2 + \left|\alpha_2 - \frac{1}{(1 - \delta_0)\overline{\alpha}_1}\right|^2 \|\vct{x}_0\|_2^2 \|\Dh\|_2^2 
\\
& \leq \left(\frac{\delta^4}{4(1 - \delta)^2} d_0^2 + \delta_0^2(4d_0) (6.1 \delta^2 d_0) + (1.22 \delta)^2 d_0 (4.7 \delta^2 d_0 )\right) \leq 8.4 \delta^4 d_0^2.
\end{align*}
where $\|\tilde{\bx}\| \leq \|\Dx\|$. 
~\\
Case 2: $\|\vct{h}\|_2 < \|\vct{x}\|_2$ and $\alpha = \frac{1}{(1-\delta_0)\overline{\alpha_2}}$. In this case, we have
\[
\Dh =  \left(\alpha_1 - \frac{1}{(1 - \delta_0)\overline{\alpha_2}}\right)\vct{h}_0 + \tilde{\vct{h}},\quad \Dx = \tilde{\vct{x}} + \delta_0 \alpha_2 \vct{x}_0. 
\]
By the symmetry of $\Kd \cap \Keps$, we can prove $|\alpha_1| \leq \sqrt{1+\delta}$,
\[
\left|\alpha_1 - \frac{1}{(1 - \delta_0)\overline{\alpha_2}}\right|  \leq \frac{\delta \sqrt{1+ \delta}}{1 - \delta} + \frac{\sqrt{1+\delta}\delta_0}{(1 - \delta_0)(1 - \delta)} \leq 1.22 \delta.
\]
Moreover, we can prove $\|\Dh\|_2^2 \leq 6.1 \delta^2 d_0$, $\|\Dx\|_2^2 \leq 4.7 \delta^2 d_0$ and $\|\Dh\|_2^2 \|\Dx\|_2^2 \leq 8.4 \delta^4 d^2$.\\
~\\
Next, under the additional assumption $(\vct{h}, \vct{x}) \in \Kmu$, we now prove $\sqrt{L}\|\mtx{B}(\Dh)\|_\infty \leq 6\mu\sqrt{d_0}$:\\
Case 1: $\|\vct{h}\|_2 \geq \|\vct{x}\|_2$ and $\alpha = (1 - \delta_0) \alpha_1$. By \prettyref{lem:orth_decomp} gives $|\alpha_1| \leq 2$, which implies
\begin{align*}
\sqrt{L}\|\mtx{B}(\Dh)\|_\infty &\leq \sqrt{L}\|\mtx{B}\vct{h}\|_\infty + (1 - \delta_0) |\alpha_1|\sqrt{L}\|\mtx{B}\vct{h}_0\|_\infty 
\\
&\leq 4\mu\sqrt{d_0} + 2(1 - \delta_0)\mu_h \sqrt{d}_0 \leq 6\mu\sqrt{d_0}.
\end{align*}
Case 2: $\|\vct{h}\|_2 < \|\vct{x}\|_2$ and $\alpha = \frac{1}{(1-\delta_0)\overline{\alpha_2}}$. Notice that in this case we have $|\alpha_1| \leq \sqrt{1+\delta}$, so $\frac{1}{|(1 -\delta_0)\overline{\alpha_2}|}= \frac{|\alpha_1|}{|(1 -\delta_0)\overline{\alpha_2} \alpha_1|} \leq \frac{\sqrt{1+\delta}}{|1 - \delta_0||1 - \delta|}$. Therefore
\begin{align*}
\sqrt{L}\|\mtx{B}(\Dh)\|_\infty &\leq \sqrt{L}\|\mtx{B}(\vct{h})\|_\infty + \frac{1}{(1 - \delta_0) |\overline{\alpha_2}|} \sqrt{L}\|\mtx{B}(\vct{h}_0)\|_\infty
\\
&\leq 4\mu\sqrt{d_0} + \frac{\sqrt{1+\delta}}{|1 - \delta_0||1 - \delta|}\mu_h \sqrt{d}_0 \leq 5.2 \mu\sqrt{d_0}.
\end{align*}
\end{proof}

\begin{lemma}\label{lem:regF}
For any $(\bh, \bx) \in \Kint$ with $\epsilon \leq \frac{1}{15}$,  the  following inequality holds uniformly:
\begin{equation*}
\Real\lp\lag \nabla F_{\bh}, \Dh \rag + \lag \nabla F_{\bx}, \Dx \rag\rp \geq \frac{\delta^2 d_0^2}{8} - 2\delta d_0 \|\A^*(\be)\|,
\end{equation*}
provided 
$L \geq C\mu^2 (K+N)\log^2 L$  for some numerical constant $C$.
\end{lemma}

\begin{proof}
In this section, define $\BU$ and $\BV$ as
\begin{equation}
\label{eq:UV2}
\BU = \alpha\bh_0\Dx^* + \overline{\alpha}^{-1}\Dh\bx_0^* \in T, \quad \BV = \Dh\Dx^*,
\end{equation}
which gives
\begin{equation*}
\bh\bx^* - \bh_0\bx_0^* = \BU + \BV.
\end{equation*}
Notice that generally $\BV \in T^\perp$ does not hold.
Recall that
\begin{equation*}
\nabla F_{\bh} = \A^*(\A(\bh\bx^* - \bh_0\bx_0^*) - \be) \bx, \quad \nabla F_{\bx} = [\A^*(\A(\bh\bx^* - \bh_0\bx_0^*) - \be)]^* \bh.
\end{equation*}
Define $I_0 := \lag \nabla_{\vct{h}} F, \Dh \rag + \overline{\lag \nabla_{\bx} F, \Dx \rag}$ and we have 
$\Real(I_0)=\Real\lp\lag \nabla_{\vct{h}} F, \Dh \rag + \lag \nabla_{\bx} F, \Dx \rag\rp$.  Since 
\begin{eqnarray*}
I_0 
& = &  \lag  \A^*(\A(\bh\bx^* - \bh_0\bx_0^*) - \be), \Dh \bx^* + \bh \Dx^* \rag \\
& = & \lag \A(\bh\bx^* - \bh_0\bx_0^*) - \be, \A(\bh\bx^* - \bh_0\bx_0^* + \Dh\Dx^*) \rag
\\ 
& = & \lag \A(\BU + \BV), \A(\BU + 2\BV)\rag  - \lag \A^*(\be), \BU + 2\BV\rag: = I_{01} + I_{02}
\end{eqnarray*}
where $\Dh\bx^* + \bh \Dx^* = \bh\bx^* - \bh_0\bx_0^* + \Dh\Dx^*.$
By the Cauchy-Schwarz inequality, $\Real(I_{01})$ has the lower bound
\begin{eqnarray}
\label{eq:I_0}
\Real(I_{01})& \geq & \|\A(\BU)\|^2 - 3\|\A(\BU)\|\|\A(\BV)\| + 2\|\A(\BV)\|^2 \nonumber
\\
& \geq & (\|\A(\BU)\| - \|\A(\BV)\|) (\|\A(\BU)\| - 2\|\A(\BV)\|). 
\end{eqnarray}
In the following, we will give an upper bound for $\|\A(\BV)\|$ and a lower bound for $\|\A(\BU)\|$.

Upper bound for $\|\A(\BV)\|$: By \prettyref{lem:DxDh} and \prettyref{lem:key}, we have
\begin{align*}
\|\A(\BV)\|^2 &\leq \left(\frac{4}{3}\|\Dh\|^2 + 2\|\Dh\| \|\mtx{B}\Dh\|_\infty \sqrt{2(K+N)\log L} + 8\|\mtx{B}\Dh\|_\infty^2 (K+N) (\log L)\right)\|\Dx\|^2
\\
&\leq \left(11.2 \delta^2 + C_0\left(\delta \mu \sqrt{\frac{1}{L} (K+N)(\log L)}  +\frac{1}{L}\mu^2 (K+N)(\log L)\right) \right) \delta^2 d_0^2
\\
& \leq \left(11.2\delta^2 + C_0\left( \frac{\delta}{\sqrt{ C\log L}} + \frac{1}{C\log L}\right)\right) \delta^2 d_0^2
\end{align*}
for some numerical constant $C_0$. Then by $\delta  \leq \eps \leq  \frac{1}{15}$ and letting $L \geq C\mu^2 (K + N)\log^2 L$ for a sufficiently large numerical constant $C$, there holds
\begin{equation}\label{eq:AVdelta}
\|\A(\BV)\|^2 <\frac{\delta^2 d_0^2}{16} \implies \|\A(\BV)\| \leq \frac{\delta d_0}{4}.
\end{equation}

Lower bound for $\|\A(\BU)\|$: By \prettyref{lem:DxDh}, we have
\[
\|\BV\|_F = \|\Dh\|_2 \|\Dx\|_2 \leq 2.9 \delta^2 d_0,
\]
and therefore
\[
\|\BU\|_F \geq d_0 \delta - 2.9\delta^2 d_0 \geq \frac{4}{5} d_0 \delta.
\]
if $\epsilon \leq \frac{1}{15}$. Since $\BU \in T$, by \prettyref{lem:ripu}, there holds
\begin{equation}\label{eq:AUdelta}
\|\A(\BU)\| \geq \sqrt{\frac{9}{10}}\|\BU\|_F \geq \frac{3}{4} d_0 \delta.
\end{equation}
With the upper bound of $\A(\BV)$ in~\eqref{eq:AVdelta}, the lower bound of $\A(\BU)$ in \eqref{eq:AUdelta}, and~\eqref{eq:I_0}, we finally arrive at
\[
\Real(I_{01}) \geq \frac{\delta^2 d_0^2}{8}.
\]

Now let us give a lower bound for $\Real(I_{02})$, 
\begin{equation*}
\Real(I_{02}) \geq  - \|\A^*(\be)\| \|\BU + 2\BV\|_* \geq  - \sqrt{2}\|\A^*(\be)\| \|\BU + 2\BV\|_F  \geq -2\delta d_0 \|\A^*(\be)\|
\end{equation*}
where $\|\cdot\|$ and $\|\cdot\|_*$ are a pair of dual norms and 
\begin{equation*}
\|\BU + 2\BV\|_F \leq \|\BU + \BV\|_F + \|\BV\|_F \leq \delta d_0+ 2.9\delta^2 d_0 \leq 1.2\delta d_0
\end{equation*}
if $\delta \leq \eps \leq \frac{1}{15}.$
Combining the estimation of $\Real(I_{01})$ and $\Real(I_{02})$ above leads to 
\begin{equation*}
\Real( \lag \nabla F_{\bh}, \Dh\rag + \lag \nabla F_{\bx}, \Dx\rag 
\geq  \frac{\delta^2 d_0^2}{8} - 2\delta d_0 \|\A^*(\be)\|
\end{equation*}
as we desired.
\end{proof}

\begin{lemma}\label{lem:regG}
For any $(\bh, \bx) \in \Kd \bigcap \Keps$ with $\epsilon \leq \frac{1}{15}$ and $\frac{9}{10}d_0 \leq d\leq \frac{11}{10}d_0$,   the following inequality holds uniformly
\begin{equation}
\label{eq:G_regularity}
\Real\lp\lag \nabla G_{\bh}, \Dh \rag + \lag \nabla G_{\bx}, \Dx \rag\rp  \geq {\frac{\delta }{5}}\sqrt{ \rho G(\bh, \bx)},
\end{equation}
where $\rho \geq d^2 + 2\|\be\|^2.$
\end{lemma}

\begin{proof}
Recall that $G_0'(z) = 2\max\{z - 1, 0\} = 2\sqrt{G_0(z)}$. Using the Wirtinger derivative of $G$ in~\eqref{eq:WGh} and~\eqref{eq:WGx}, we have
\[
\lag \nabla G_{\bh}, \Dh \rag + \lag \nabla G_{\bx}, \Dx \rag :=  \frac{\rho}{2d} \left(H_1 +H_2 +H_3\right),
\]
where 
\begin{equation}
H_1 = G'_0\left( \frac{\|\bh\|^2}{2d}\right) \lag  \bh, \Dh \rag, \quad H_2 =  G'_0\left( \frac{\|\bx\|^2}{2d}\right)  \lag\bx, \Dx \rag,
\end{equation}
and
\begin{equation*}
H_3 =  \frac{L}{4\mu^2} \sum_{l=1}^L G'_0\left(\frac{L|\bb_l^*\bh|^2}{8d\mu^2}\right) \lag \bb_l\bb_l^*\bh, \Dh\rag.
\end{equation*}
We will give lower bounds for $H_1$, $H_2$ and $H_3$ for two cases.\\

\paragraph{Case 1} $\|\vct{h}\|_2 \geq \|\vct{x}\|_2$ and $\alpha = (1 - \delta_0) \alpha_1$. 

\paragraph{Lower bound of $H_1$:} Notice that
\[
\Dh = \bh - \alpha \bh_0 = \bh - (1 - \delta_0)\alpha_1 \bh_0 = \bh - (1-\delta_0)(\bh - \tilde{\bh}) = \delta_0 \bh + (1 - \delta_0) \tilde{\bh}.
\]
We have $\lag \bh, \Dh \rag =  \delta_0 \|\bh\|_2^2 + (1 - \delta_0) \lag \bh, \tilde{\bh}\rag = \delta_0 \|\bh\|_2^2 + (1 - \delta_0) \|\tilde{\bh}\|_2^2 \geq \delta_0 \|\bh\|_2^2$, which implies that $H_1 \geq G'_0\left( \frac{\|\bh\|^2}{2d}\right) \frac{\delta}{10} \|\bh\|_2^2$. We claim that 
\begin{equation}
\label{eq:H1_lower}
H_1 \geq \frac{\delta d}{5}G'_0\left( \frac{\|\bh\|^2}{2d}\right).
\end{equation}
 In fact, if $\|\bh\|_2^2 \leq 2d$, we get $H_1 = 0 =  \frac{\delta d}{5}G'_0\left( \frac{\|\bh\|^2}{2d}\right)$; If $\|\bh\|_2^2 > 2d$, we get \eqref{eq:H1_lower} straightforwardly.

\paragraph{Lower bound of $H_2$:} The assumption $\|\vct{h}\|_2 \geq \|\vct{x}\|_2$ gives
\[
\|\vct{x}\|_2^2 \leq \|\vct{x}\|_2 \|\vct{h}\|_2 \leq (1+\delta) d_0 \leq 1.1(1 + \delta)d_0 < 2d,
\] 
which implies that
\[
H_2 =  G'_0\left( \frac{\|\bx\|^2}{2d}\right)  \lag\bx, \Dx \rag = 0 = \frac{\delta d}{5}G'_0\left( \frac{\|\bx\|^2}{2d}\right).
\] 

\paragraph{Lower bound of $H_3$:} When $L|\bb_l^*\bh|^2 \leq 8d\mu^2$,
\[
\frac{L}{4\mu^2} G'_0\left(\frac{L|\bb_l^*\bh|^2}{8d\mu^2}\right) \lag \bb_l\bb_l^*\bh, \Dh\rag = 0 = \frac{d}{2} G'_0\left(\frac{L|\bb_l^*\bh|^2}{8d\mu^2}\right).
\]
When $L|\bb_l^*\bh|^2 > 8d\mu^2$, by \prettyref{lem:orth_decomp}, there holds $|\alpha_1| \leq 2$. Then by $\mu_h \leq \mu$, we have 
\begin{eqnarray*}
\Real(\lag \bb_l\bb_l^*\bh, \Dh\rag) & = & \Real(|\bb_l^*\bh|^2- \alpha\lag \bb_l^*\bh,\bb_l^*\bh_0\rag)\\
& \geq & |\bb_l^*\bh| (|\bb_l^*\bh| - (1-\delta_0)|\alpha_1| |\bb_l^*\bh_0|)
\\
&\geq & |\bb_l^*\bh| (|\bb_l^*\bh| - 2\mu\sqrt{d_0/L})  \\
& \geq & \sqrt{\frac{8d\mu^2}{L}} \left(\sqrt{\frac{8d\mu^2}{L}} - 2\mu\sqrt{\frac{10d}{9L}}\right) \geq \frac{2d\mu^2}{L},
\end{eqnarray*}
where $(1 - \delta_0)|\alpha_1| |\bb_l^*\bh_0| \leq \frac{2\mu_h\sqrt{d_0}}{\sqrt{L}} \leq \frac{2\mu\sqrt{10d}}{\sqrt{9L}}.$
This implies that 
\[
\frac{L}{4\mu^2} G'_0\left(\frac{L|\bb_l^*\bh|^2}{8d\mu^2}\right)\Real(\lag \bb_l\bb_l^*\bh, \Dh\rag) \geq \frac{d}{2} G'_0\left(\frac{L|\bb_l^*\bh|^2}{8d\mu^2}\right).
\]
So we always have
\[
\Real(H_3) \geq \sum_{l=1}^L \frac{d}{2} G'_0\left(\frac{L|\bb_l^*\bh|^2}{8d\mu^2}\right).
\]

\paragraph{Case 2:} $\|\vct{h}\|_2 < \|\vct{x}\|_2$ and $\alpha = \frac{1}{(1-\delta_0)\overline{\alpha_2}}$. 

\paragraph{Lower bound of $H_1$:} The assumption $\|\vct{h}\|_2 < \|\vct{x}\|_2$ gives
\[
\|\vct{h}\|_2^2 \leq \|\vct{x}\|_2 \|\vct{h}\|_2 \leq (1+\delta) d_0 < 2d,
\] 
which implies that
\[
H_1 =  G'_0\left( \frac{\|\bh\|^2}{2d}\right)  \lag\bh, \Dh \rag = 0 = \frac{\delta d}{5}G'_0\left( \frac{\|\bh\|^2}{2d}\right).
\] 

\paragraph{Lower bound of $H_2$:}
Notice that
\[
\Dx = \bx - \overline{\alpha}^{-1} \bx_0 = \bx - (1 - \delta_0)\alpha_2 \bx_0 = \bx - (1-\delta_0)(\bx - \tilde{\bx}) = \delta_0 \bx + (1 - \delta_0) \tilde{\bx}.
\]
We have $\lag \bx, \Dx \rag =  \delta_0 \|\bx\|_2^2 + (1 - \delta_0) \lag \bx, \tilde{\bx}\rag = \delta_0 \|\bx\|_2^2 + (1 - \delta_0) \|\tilde{\bx}\|_2^2 \geq \delta_0 \|\bx\|_2^2$, which implies that $H_2 \geq G'_0\left( \frac{\|\bx\|^2}{2d}\right) \frac{\delta}{10} \|\bx\|_2^2$. We claim that 
\begin{equation}
\label{eq:H2_lower}
H_2 \geq \frac{\delta d}{5}G'_0\left( \frac{\|\bx\|^2}{2d}\right).
\end{equation}
 In fact, if $\|\bx\|_2^2 \leq 2d$, we get $H_2 = 0 =  \frac{\delta d}{5}G'_0\left( \frac{\|\bx\|^2}{2d}\right)$; If $\|\bx\|_2^2 > 2d$, we get \eqref{eq:H2_lower} straightforwardly.

\paragraph{Lower bound of $H_3$:} When $L|\bb_l^*\bh|^2 \leq 8d \mu^2$,
\[
\frac{L}{4\mu^2} G'_0\left(\frac{L|\bb_l^*\bh|^2}{8d\mu^2}\right) \lag \bb_l\bb_l^*\bh, \Dh\rag = 0 = \frac{d}{4} G'_0\left(\frac{L|\bb_l^*\bh|^2}{8d\mu^2}\right).
\]
When $L|\bb_l^*\bh|^2 > 8d \mu^2$, by \prettyref{lem:orth_decomp}, there hold $|\alpha_1\overline{\alpha_2} -1|\leq \delta$ and $|\alpha_1| \leq 2$,
which implies that
\[
\frac{1}{(1-\delta_0)|\overline{\alpha_2}|} = \frac{|\alpha_1|}{(1 - \delta_0)|\alpha_1 \overline{\alpha_2}|} \leq \frac{2}{(1-\delta_0)(1-\delta)}.
\]
By $\mu_h \leq \mu$ and $\delta \leq \epsilon \leq \frac{1}{15}$, similarly we have 
\begin{align*}
\Real(\lag \bb_l\bb_l^*\bh, \Dh\rag) &\geq |\bb_l^*\bh| \left(|\bb_l^*\bh| - \frac{2}{(1-\delta_0)(1-\delta)} |\bb_l^*\bh_0|\right)
\\
&\geq \left(8 - 4\sqrt{\frac{20}{9}}\frac{1}{(1-\delta_0)(1-\delta)}\right) \frac{d\mu^2 }{L} > \frac{d\mu^2}{L}.
\end{align*}
This implies that for $1\leq l\leq L$,
\[
\frac{L}{4\mu^2} G'_0\left(\frac{L|\bb_l^*\bh|^2}{8d\mu^2}\right) \Real\lag \bb_l\bb_l^*\bh, \Dh\rag \geq \frac{d}{4} G'_0\left(\frac{L|\bb_l^*\bh|^2}{8d\mu^2}\right).
\]
So we always have
\[
\Real(H_3) \geq \sum_{l=1}^L \frac{d}{4} G'_0\left(\frac{L|\bb_l^*\bh|^2}{8d\mu^2}\right).
\]
To sum up the two cases, we have
\begin{align*}
\Real(H_1 + H_2 + H_3) &\geq \frac{\delta d}{5}G'_0\left( \frac{\|\bh\|^2}{2d}\right) + \frac{\delta d}{5}G'_0\left( \frac{\|\bx\|^2}{2d}\right) + \sum_{l=1}^L \frac{d}{4} G'_0\left(\frac{L|\bb_l^*\bh|^2}{8d\mu^2}\right)
\\
&\geq \frac{2\delta d}{5}\left( \sqrt{G_0\left( \frac{\|\bh\|^2}{2d}\right)} + \sqrt{G_0\left( \frac{\|\bx\|^2}{2d}\right)} + \sum_{l=1}^L  \sqrt{G_0\left(\frac{L|\bb_l^*\bh|^2}{8d\mu^2}\right)}   \right)
\\
&\geq \frac{2\delta d}{5}\left( \sqrt{G_0\left( \frac{\|\bh\|^2}{2d}\right) + G_0\left( \frac{\|\bx\|^2}{2d}\right)+ \sum_{l=1}^L  G_0\left(\frac{L|\bb_l^*\bh|^2}{8d\mu^2}\right)}   \right)
\end{align*}
where $G_0'(z) = 2\sqrt{G_0(z)}$ and it implies \prettyref{eq:G_regularity}.
\end{proof}

\begin{lemma}
\label{lem:reg}
Let $\tF$ be as defined in~\eqref{def:FG}, then there exists a positive constant $\omega$ such that
\begin{equation*}
\|\nabla \tF(\bh, \bx)\|^2 \geq \omega \left[ \tF(\bh, \bx) - c \right]_+
\end{equation*}
with $c =  \|\be\|^2 + 1700 \|\A^*(\be)\|^2$ and $\omega = \frac{d_0}{5000}$ for all $(\bh, \bx) \in \Kint$. 
Here we set  $\rho \geq d^2 + 2\|\be\|^2.$
\end{lemma}

\begin{proof}
Following from Lemma~\ref{lem:regF} and Lemma~\ref{lem:regG}, we have
\begin{eqnarray*}
\Real( \lag \nabla F_{\bh}, \Dh\rag + \lag \nabla F_{\bx}, \Dx\rag)
& \geq & \frac{\delta^2 d_0^2}{8} - 2\delta d_0 \|\A^*(\be)\| \\
 \Real( \lag \nabla G_{\bh}, \Dh \rag + \lag \nabla G_{\bx}, \Dx \rag) 
 & \geq & \frac{\delta d}{5} \sqrt{ G(\bh, \bx)} \geq \frac{9\delta d_0}{50}\sqrt{G(\bh, \bx)}
\end{eqnarray*}
for $\alpha = (1 - \delta_0)\alpha_1$ or $\frac{1}{(1 - \delta)\overline{\alpha_2}}$ and $\forall (\bh, \bx) \in \Kint$ where $\rho \geq d^2 + 2\|\be\|^2  \geq d^2$ and $\frac{9}{10}d_0 \leq d \leq \frac{11}{10}d_0$.
Adding them together gives
\begin{eqnarray}
\frac{\delta^2 d_0^2}{8} + \frac{  9\delta d_0}{50}\sqrt{ G(\bh, \bx)} - 2\delta d_0 \|\A^*(\be) \|
& \leq & \Real\left (\lag  \nabla F_{\bh}+  \nabla G_{\bh}, \Dh \rag + \lag \nabla F_{\bx} +  \nabla G_{\bx} \rag\right) \nonumber \\
& \leq & \|\nabla\tF_{\bh}\| \| \Dh \|   + \| \nabla\tF_{\bx} \| \| \Dx \| \nonumber \\
& \leq & \sqrt{2} \| \nabla \tF(\bh, \bx)\| \max\{ \| \Dh \|, \| \Dx \|\}  \nonumber \\
& \leq & 3.6 \delta \sqrt{d_0}  \| \nabla \tF(\bh, \bx)\| \label{eq:nabla-tF}
\end{eqnarray}
where both $\|\Dh\|$ and $\|\Dx\|$ are bounded by $2.5\delta\sqrt{d_0}$ in  Lemma~\ref{lem:DxDh}.
Note that
\begin{equation}\label{eq:AEHX}
\sqrt{2\left[ \Real(\lag \A^*(\be), \bh\bx^* - \bh_0\bx_0^*\rag) \right]_+} \leq \sqrt{ 2\sqrt{2} \|\A^*(\be)\| \delta d_0} \leq \frac{\sqrt{5}\delta d_0}{4} + \frac{4}{\sqrt{5}}\|\A^*(\be)\|.
\end{equation}
Dividing both sides of~\eqref{eq:nabla-tF} by $\delta d_0$, we obtain
\begin{equation*}
\frac{3.6}{\sqrt{d_0}} \|\nabla \tF(\bh, \bx)\|  \geq 
 \frac{\delta d_0}{12}  + \frac{9}{50}\sqrt{G(\bh, \bx)} + \frac{\delta d_0 }{24} - 2\|\A^*(\be)\|.
\end{equation*}
The Local RIP condition implies $F_0(\bh, \bx) \leq \frac{5}{4}\delta^2 d_0^2$ and hence $\frac{\delta d_0}{12} \geq \frac{1}{6\sqrt{5}}\sqrt{F_0(\bh, \bx)}$, where $F_0$ is defined in~\eqref{def:F0}.  Combining the equation above and~\eqref{eq:AEHX}, 
\begin{eqnarray*}
\frac{3.6}{\sqrt{d_0}} \|\nabla \tF(\bh, \bx)\| 
& \geq & \frac{1}{6\sqrt{5}} \Big[ \left(\sqrt{F_0(\bh, \bx)} +\sqrt{2\left[ \Real(\lag \A^*(\be), \bh\bx^* - \bh_0\bx_0^*\rag) \right]_+} + \sqrt{G(\bh, \bx)}\right) \\
&& + \frac{\sqrt{5}\delta d_0}{4} - \left( \frac{\sqrt{5}\delta d_0}{4} + \frac{4}{\sqrt{5}}\|\A^*(\be)\|\right)\Big] - 2\|\A^*(\be)\| \\
& \geq & \frac{1}{6\sqrt{5}} \left[ \sqrt{ \left[\tF(\bh, \bx) - \|\be\|^2\right]_+} -  29\|\A^*(\be)\|\right]
\end{eqnarray*}
where $\tF(\bh, \bx) - \|\be\|^2 \leq F_0(\bh, \bx) + 2 [\Real(\lag \A^*(\be), \bh\bx^* - \bh_0\bx_0^*\rag)]_+ + G(\bh, \bx)$ follows from definition and~\eqref{eq:F-decom}.
Finally, we have
\begin{equation*}
\|\nabla \tF(\bh, \bx)\|^2 \geq \frac{d_0}{2500} \left[ \sqrt{\left[\tF(\bh, \bx) - \|\be\|^2\right]_+} - 29 \|\A^*(\be)\|\right]_+^2
\end{equation*}
for all $(\bh, \bx)\in \Kint.$ For any nonnegative fixed real numbers $a$ and $b$, we have
\begin{equation*}
[\sqrt{(x - a)_+} - b ]_+ + b \geq \sqrt{(x - a)_+} 
\end{equation*}
and it implies
\begin{equation*}
( x - a)_+ \leq 2 ( [\sqrt{(x - a)_+} - b ]_+^2 + b^2) \Longrightarrow [\sqrt{(x - a)_+} - b ]_+^2  \geq \frac{(x - a)_+}{2} - b^2.
\end{equation*}

Therefore, by setting $ a = \|\be\|$ and $b = 30\|\A^*(\be)\|$, there holds
\begin{eqnarray*}
\|\nabla \tF(\bh, \bx)\|^2 
& \geq & \frac{d_0}{2500} \left[ \frac{\tF(\bh, \bx) - \|\be\|^2 }{2} -  850 \|\A^*(\be)\|^2 \right]_+ \\
& \geq & \frac{d_0}{5000} \left[ \tF(\bh, \bx) - (\|\be\|^2 + 1700 \|\A^*(\be)\|^2) \right]_+.
\end{eqnarray*}
\end{proof}

\subsection{Local smoothness}
\label{s:smooth}

\begin{lemma}
For any $\bz : = (\bh, \bx)$ and $\bw : = (\bu, \bv)$ such that $\vct{z}, \vct{z}+\vct{w} \in \Keps \cap \KF$, there holds
\begin{equation*}
\| \nabla\tF(\bz + \bw) - \nabla\tF(\bz) \| \leq C_L \|\bw\|,
\end{equation*}
with
\begin{equation*}
C_L \leq \sqrt{2}d_0\left[ 10 \|\A\|^2 + \frac{\rho}{d^2} \left(5 + \frac{3L}{2\mu^2}\right)\right] 
\end{equation*}
where $\rho \geq d^2 + 2\|\be\|^2$ 
and $\|\A\|^2\leq \sqrt{N\log(NL/2) + \gamma\log L}$ holds with probability at least $1 - L^{-\gamma}$ from Lemma~\ref{lem:A-UPBD}. In particular, $L = \mathcal{O}((\mu^2 + \sigma^2) (K + N)\log^2 L)$ 
and $\|\be\|^2 = \mathcal{O}(\sigma^2d_0^2)$ follows from $\|\be\|^2 \sim \frac{\sigma^2d_0^2}{2L} \chi^2_{2L}$ and~\eqref{ineq:bern}. Therefore, $C_L$ can be simplified into
\begin{equation*}
C_L = \mathcal{O}(d_0(1 + \sigma^2)(K + N)\log^2 L )
\end{equation*}
by choosing $\rho \approx d^2 + 2\|\be\|^2.$
\end{lemma}

\begin{proof}
By \prettyref{lem:betamu}, we have $\vct{z}=(\vct{h}, \vct{x}), \vct{z}+\vct{w}=(\vct{h}+\vct{u}, \vct{x}+\vct{v}) \in \Kd \cap \Kmu$. 
Note that 
\[
\nabla \tF = (\nabla \tF_{\bh}, \nabla \tF_{\bx}) = (\nabla F_{\bh} + \nabla G_{\bh}, \nabla F_{\bx} + \nabla G_{\bx}),
\]
where 
\begin{equation}
\nabla F_{\bh} = \A^*(\A(\bh\bx^* - \bh_0\bx_0^*) - \be) \bx, \quad  \nabla F_{\bx} = (\A^*(\A(\bh\bx^* - \bh_0\bx_0^*) - \be))^* \bh,
\end{equation}
and
\begin{equation*}
\nabla G_{\bh} 
= \frac{\rho}{2d}\left[G'_0\left(\frac{\|\bh\|^2}{2d}\right) \bh 
+ \frac{L}{4\mu^2} \sum_{l=1}^L G'_0\left(\frac{L|\bb_l^*\bh|^2}{8d\mu^2}\right) \bb_l\bb_l^*\bh \right],\quad
\nabla G_{\bx} = \frac{\rho}{2d}G'_0\left( \frac{\|\bx\|^2}{2d}\right) \bx.
\end{equation*}
\paragraph{Step 1:} we estimate the upper bound of $\|\nabla F_{\bh}(\bz + \bw) - \nabla F_{\bh}(\bz) \|$. A straightforward calculation gives
\begin{equation*}
\nabla F_{\bh}(\bz + \bw) - \nabla F_{\bh}(\bz) = \A^*\A( \bu\bx^* + \bh\bv^* + \bu\bv^*) \bx + \A^*\A((\bh+\bu)(\bx+\bv)^* - \bh_0\bx_0^*) \bv - \A^*(\be)\bv.
\end{equation*}
Note that $\bz, \bz+\bw \in \Kd$ directly implies 
\begin{equation*}
\|\bu\bx^* + \bh\bv^* + \bu\bv^*\|_F \leq \|\vct{u}\| \|\vct{x}\| + \|\vct{h}+\vct{u}\|\|\vct{v}\| \leq 2\sqrt{d_0} (\|\bu\| + \|\bv\|)
\end{equation*}
where $\|\bh + \bu\| \leq 2\sqrt{d_0}.$
Moreover, $\vct{z} + \vct{w} \in \Keps$ implies
\begin{equation*}
\|(\bh+\bu)(\bx+\bv)^* - \bh_0\bx_0^*\|_F \leq \epsilon d_0.
\end{equation*}
Combined with $\|\A^*(\be)\| \leq \eps d_0$ and $\|\bx\| \leq 2\sqrt{d_0}$, we have 
\begin{eqnarray}
\|\nabla F_{\bh}(\bz + \bw) - \nabla F_{\bh}(\bz) \| 
& \leq & 4d_0 \|\A\|^2(\|\bu\| + \|\bv\|) + \eps d_0 \|\A\|^2 \|\bv\| + \eps d_0 \|\bv\|  \nonumber \\
& \leq & 5d_0 \|\A\|^2 ( \|\bu\| + \|\bv\|). 
\label{eq:LipFh}
\end{eqnarray}
\paragraph{Step 2:} we estimate the upper bound of $\|\nabla F_{\bx}(\bz + \bw) - \nabla F_{\bx}(\bz)\|$. Due to the symmetry between $\nabla F_{\bh}$ and $\nabla F_{\bx}$, we have,
\begin{equation}
\label{eq:LipFx}
\|\nabla F_{\bx}(\bz + \bw) - \nabla F_{\bx}(\bz) \| \leq 5d_0\|\A\|^2 ( \|\bu\| + \|\bv\|).\end{equation}

\paragraph{Step 3:} we estimate the upper bound of $\|\nabla G_{\bx}(\bz + \bw) - \nabla G_{\bx}(\bz)\|$. Notice that $G_0'(z) = 2\max\{z - 1, 0\}$, which implies that for any $z_1, z_2, z\in \RR$, there holds
\begin{equation}
\label{eq:LipG}
|G'_0(z_1) - G'_0(z_2)| \leq 2|z_1 - z_2|, \quad G'(z) \leq 2|z|,
\end{equation}
although $G'(z)$ is not differentiable at $z = 1$. Therefore, by~\eqref{eq:LipG}, it is easy to show that 
\[
\left| G'_0\left( \frac{\|\bx + \bv\|^2}{2d}\right) - G'_0\left( \frac{\|\bx\|^2}{2d}\right) \right| \leq \frac{\|\bx + \bv\| + \|\bx\|}{d} \|\bv\| \leq \frac{4\sqrt{d_0}}{d}\|\bv\|
\]
where $\|\bx + \bv\| \leq 2\sqrt{d_0}.$
Therefore, by $\bz, \bz+\bw \in \Kd$, we have
\begin{align}
\| \nabla G_{\bx}(\bz + \bw) - \nabla G_{\bx}(\bz) \| &\leq \frac{\rho}{2d} \left| G'_0\left( \frac{\|\bx + \bv\|^2}{2d}\right) - G'_0\left( \frac{\|\bx\|^2}{2d}\right)\right| \| \bx  + \bv \| + \frac{\rho}{2d}  G'_0\left( \frac{\|\bx\|^2}{2d}\right) \| \bv \| \nonumber 
\\
& \leq \frac {\rho}{2d} \left( \frac{8d_0 \|\bv\|}{d} + \frac{2d_0\|\bv\|}{d} \right) \leq \frac{5d_0 \rho}{d^2} \|\bv\|.   \label{eq:LipGx}
\end{align}

\paragraph{Step 4:} we estimate the upper bound of $\|\nabla G_{\bh}(\bz + \bw) - \nabla G_{\bh}(\bz)\|$. Denote 
\begin{align*}
\nabla G_{\bh}(\bz + \bw) - \nabla G_{\bh}(\bz) &= \frac{\rho}{2d}\left[G'_0\left(\frac{\|\bh + \bu\|^2}{2d}\right) (\bh + \bu) - G'_0\left(\frac{\|\bh\|^2}{2d}\right) \bh\right] 
\\
&+ \frac{\rho L}{8d\mu^2 }\sum_{l=1}^L \left[G'_0\left(\frac{L|\bb_l^*(\bh + \bu)|^2}{8d\mu^2}\right) \bb_l^*(\bh + \bu) - G'_0\left(\frac{L|\bb_l^*\bh|^2}{8d\mu^2}\right) \bb_l^*\bh \right]\bb_l
\\
&:= \vct{j}_1 + \vct{j}_2.
\end{align*}
Similar to \prettyref{eq:LipGx}, we have
\begin{equation}
\label{eq:LipG1}
\|\vct{j}_1\| \leq \frac{5d_0 \rho}{d^2} \|\bu\|.
\end{equation}
Now we control $\|\vct{j}_2\|$. Since $\vct{z}, \vct{z}+\vct{w} \in \Kmu$, we have
\begin{eqnarray}
\left| G'_0\left(\frac{L|\bb_l^*(\bh + \bu)|^2}{8d\mu^2}\right) - G'_0\left(\frac{L|\bb_l^*\bh|^2}{8d\mu^2}\right)\right| 
& \leq & \frac{L}{4d\mu^2}\left(|\vct{b}_l^*(\vct{h} + \vct{u})| + |\vct{b}_l^*\vct{h}|\right) |\bb_l^*\bu| \nonumber \\
& \leq & \frac{2\sqrt{d_0 L}}{d\mu} |\vct{b}_l^*\vct{u}|. \label{eq:step4-1}
\end{eqnarray}
and
\begin{equation}
\label{eq:step4-2}
\left|G'_0\left(\frac{L|\bb_l^*(\bh + \bu)|^2}{8d\mu^2}\right)\right| \leq 2\frac{L|\bb_l^*(\bh + \bu)|^2}{8d\mu^2} \leq \frac{4d_0}{d}
\end{equation}
where both $\max_l |\bb_l^*(\bh + \bu)|$ and $\max_l |\bb_l^*\bh|$ are bounded by $\frac{4 \sqrt{d_0}\mu}{\sqrt{L}}$.
Let $\alpha_l$ be
\begin{align*}
\alpha_l &:=G'_0\left(\frac{L|\bb_l^*(\bh + \bu)|^2}{8d\mu^2}\right) \bb_l^*(\bh + \bu) - G'_0\left(\frac{L|\bb_l^*\bh|^2}{8d\mu^2}\right) \bb_l^*\bh
\\
& = \left(G'_0\left(\frac{L|\bb_l^*(\bh + \bu)|^2}{8d\mu^2}\right) - G'_0\left(\frac{L|\bb_l^*\bh|^2}{8d\mu^2}\right)\right)\vct{b}_l^*\vct{h} + G'_0\left(\frac{L|\bb_l^*(\bh + \bu)|^2}{8d\mu^2}\right)\vct{b}_l^*\vct{u}.
\end{align*}
Applying~\eqref{eq:step4-1} and~\eqref{eq:step4-2} leads to
\[
|\alpha_l| \leq \frac{2\sqrt{d_0 L}}{d \mu}|\vct{b}_l^*\vct{u}|\left(4\mu \sqrt{\frac{d_0}{L}}\right) + \frac{4d_0}{d}|\vct{b}_l^*\vct{u}| = \frac{12 d_0}{d}|\vct{b}_l^*\vct{u}|.
\]
Since $\sum_{l=1}^L \alpha_l \vct{b}_l = \mtx{B}^* \begin{bmatrix} \alpha_1 \\ \vdots \\ \alpha_L \end{bmatrix}$ and $\|\mtx{B}\|=1$, there holds
\[
\left\|\sum_{l=1}^L \alpha_l \vct{b}_l \right\|^2 \leq \sum_{l=1}^L |\alpha_l|^2 \leq \left(\frac{12 d_0}{d}\right)^2 \sum_{l=1}^L |\vct{b}_l^*\vct{u}|^2 = \left(\frac{12 d_0}{d}\right)^2 \|\mtx{B}\vct{u}\|^2 \leq \left(\frac{12 d_0}{d}\right)^2 \|\vct{u}\|^2.
\]
This implies that
\begin{equation}
\label{eq:LipG2}
\|\vct{j}_2\| 
= \frac{\rho L}{8d\mu^2} \left\|\sum_{l=1}^L \alpha_l \vct{b}_l \right\| 
\leq \frac{\rho L}{8d\mu^2}\frac{12 d_0}{d}\|\vct{u}\| 
= \frac{3\rho Ld_0}{2d^2\mu^2}\|\vct{u}\|.
\end{equation}
In summary, by combining~\prettyref{eq:LipFh}, \prettyref{eq:LipFx}, \prettyref{eq:LipGx}, \prettyref{eq:LipG1}, and \prettyref{eq:LipG2}, we conclude that
\[
\|\nabla \tF(\bz + \bw) - \nabla \tF(\bz)\| \leq 10d_0 \|\mathcal{A}\|^2 (\|\vct{u}\| + \|\vct{v}\|) + \frac{5d_0\rho}{d^2}(\|\vct{u}\| + \|\vct{v}\|) + \frac{3\rho Ld_0}{2\mu^2d^2}\|\vct{u}\|.
\]
With $\|\bu\| + \|\bv\| \leq \sqrt{2}\|\bw\|$, there holds
\begin{equation*}
\|\nabla \tF(\bz + \bw) - \nabla \tF(\bz)\|
\leq \sqrt{2}d_0\left[ 10 \|\A\|^2 + \frac{\rho}{d^2} \left(5 + \frac{3L}{2\mu^2}\right)\right] \|\bw\|.
\end{equation*}
\end{proof}


\subsection{Initialization}\label{s:init}

This section is devoted to justifying the validity of the  \textit{Robustness condition} and to proving Theorem~\ref{thm:init}, i.e., establishing the fact that Algorithm~\ref{Initial} constructs an initial guess $(\bu_0, \bv_0)\in \frac{1}{\sqrt{3}}\Kd\bigcap \frac{1}{\sqrt{3}} \Kmu \bigcap \MN_{\frac{2}{5}\eps }.$

\begin{lemma}\label{lem:Ay-hx}
For $\be\sim\mathcal{N}(\bzero, \frac{\sigma^2d_0^2}{2L}\I_L) + \mi\mathcal{N}(\bzero, \frac{\sigma^2d_0^2}{2L}\I_L)$, there holds
\begin{equation}\label{eq:Ay-hx}
\| \A^*(\by) - \bh_0\bx_0^* \| \leq \xi d_0, 
\end{equation}
with probability at least $1 - L^{-\gamma}$ if $L \geq C_{\gamma} (\mu^2_h + \sigma^2) \max\{K,  N\} \log L /\xi^2.$ 
Moreover, 
\begin{equation*}
\|\A^*(\be)\| \leq \xi d_0
\end{equation*}
with  probability at least $1 - L^{-\gamma}$ if $L \geq C_{\gamma}(\frac{\sigma^2}{\xi^2} + \frac{\sigma}{\xi})\max\{K, N\} \log L.$  In particular, we fix $\xi = \frac{\eps}{10\sqrt{2}}$ and then Robustness condition~\ref{cond:AE}  holds, i.e., $\|\A^*(\be)\| \leq \frac{\eps d_0}{10\sqrt{2}}$. 
\end{lemma}

\begin{proof}
In this proof, we can assume $d_0 = 1$ and $\|\bh_0\| = \|\bx_0\| = 1$, without loss of generality. First note that $\E(\A^*\by) = \E(\A^*\A(\bh_0\bx_0^* ) + \A^*(\be))= \bh_0\bx_0^*.$  We will use the matrix Bernstein inequality to show that 
\begin{equation*}
\| \A^*(\by) - \bh_0\bx_0^* \| \leq \xi.
\end{equation*}
By definition of $\A$ and $\A^*$ in~\eqref{def:A} and~\eqref{def:AD},
\begin{equation*}
\A^*(\by) - \bh_0\bx_0^* = \sum_{l=1}^L \left[ \bb_l\bb^*_l\bh_0\bx_0^* (\ba_l\ba_l^* - \I_N) + e_l \bb_l\ba_l^*\right] = \sum_{l=1}^L \CZ_l,
\end{equation*}
where $\CZ_l : = \bb_l\bb_l^*\bh_0\bx_0^*(\ba_l \ba_l^* - \I_N) + e_l \bb_l\ba_l^*$ and $\sum_{l=1}^L \bb_l\bb_l^* = \I_K$. In order to apply Bernstein inequality~\eqref{thm:bern}, we need to estimate both the exponential norm $\|\CZ_l\|_{\psi_1}$  and the variance. 
\begin{eqnarray*}
\|\CZ_l\|_{\psi_1} & \leq & C \|\bb_l\| |\bb_l^*\bh_0| \| \bx_0^*(\ba_l\ba_l^* - \I_N)\|_{\psi_1} + C \|e_l \bb_l\ba_l^*\|_{\psi_1}\\
& \leq & C \frac{\mu_h\sqrt{KN}}{L} + C\frac{\sigma \sqrt{KN}}{L} \leq C\frac{(\mu_h + \sigma) \sqrt{KN}}{L},
\end{eqnarray*}
for some constant $C$.  Here,~\eqref{lem:11JB} of Lemma~\ref{lem:multiple1} gives
\begin{equation*}
\|\bx_0^* (\ba_l\ba_l^* - \I_N)\|_{\psi_1} \leq C\sqrt{N}
\end{equation*}
and 
\begin{equation*}
 \|e_l \bb_l\ba_l^*\|_{\psi_1}\leq  C\sqrt{\frac{K}{L}} (|e_l| \| \ba_l\| )_{\psi_1} \leq C\sqrt{ \frac{K}{L}} \frac{\sigma}{\sqrt{L}}\sqrt{N} \leq \frac{C\sigma\sqrt{KN}}{L}
\end{equation*}
follows from~\eqref{lemma:psi} where both $|e_l|$ and $\|\ba_l\|$ are sub-gaussian random variables.

Now we give an upper bound of $\sigma_0^2 : = \max\{ \|E\sum_{l=1}^L \CZ_l^*\CZ_l\|, \| \E\sum_{l=1}^L \CZ_l\CZ_l^*\|\}$.
\begin{eqnarray*}
\left\| \E \sum_{l=1}^l \CZ_l\CZ_l^* \right\| 
& = & \left\| \sum_{l=1}^L |\bb_l^*\bh_0|^2 \bb_l\bb_l^* \E \|(\ba_l\ba_l^* - \I_N)\bx_0  \|^2 + \sum_{l=1}^L \E(|e_l|^2 \|\ba_l\|^2)\bb_l\bb_l^*  \right\| \\
& \leq  & C N \left\| \sum_{l=1}^L |\bb_l^*\bh_0|^2 \bb_l\bb_l^* \right\|^2 + C\frac{\sigma^2 N}{L}\\
& \leq & C\frac{\mu^2_hN}{L} + C\frac{\sigma^2 N}{L} 
\leq C\frac{(\mu^2_h + \sigma^2)N}{L}.
\end{eqnarray*}
where $\E\|(\ba_l\ba_l^* - \I_N)\bx_0  \|^2 = \bx_0^* \E(\ba_l \ba_l^* - \I_N)^2 \bx_0 = N$ follows from~\eqref{lem:9JB} and $\E( |e_l|^2) = \frac{\sigma^2}{L}.$
\begin{eqnarray*}
\left\| \E \sum_{l=1}^l \CZ^*_l\CZ_l \right\| 
& \leq & \left\| \sum_{l=1}^L |\bb_l^*\bh_0|^2  \|\bb_l\|^2 \E \left[ (\ba\ba_l^* - \I_N)\bx_0\bx_0^* (\ba_l\ba_l^* - \I_N)\right] \right \| \\
& & + \left\| \E\sum_{l=1}^L e_l^2 \|\bb_l\|^2 \ba_l\ba_l^* \right\| \\
& \leq & C\frac{K}{L} \left\| \sum_{l=1}^L |\bb_l^*\bh_0|^2 \I_N\right\| + C\frac{\sigma^2K}{L^2} \left\| \sum_{l=1}^L \E(\ba_l\ba_l^*) \right\| = C\frac{(\sigma^2 + 1) K}{L}
\end{eqnarray*}
where we have used the fact that $\sum_{l=1}^L |\bb_l^*\bh_0|^2 = \|\bh_0\|^2 = 1.$ Therefore, we now have the variance $\sigma^2_0$ bounded by $\frac{(\mu^2_h  + \sigma^2)\max\{K, N\}}{L}.$ We apply Bernstein inequality~\eqref{thm:bern}, and obtain
\begin{eqnarray*}
\left\|\sum_{l=1}^L \CZ_l \right\| & \leq & C_0 \max\Big\{ \sqrt{\frac{( \mu^2_h + \sigma^2)\max\{K,N\} (\gamma + 1) \log L}{L}}, \\
&&  \quad \frac{\sqrt{KN}(\mu_h + \sigma)(\gamma + 1)\log^2 L }{L}  \Big\} \leq \xi
\end{eqnarray*}
with probability at least $1 - L^{-\gamma}$ if $L \geq C_{\gamma} (\mu^2_h + \sigma^2 )\max\{K,  N\}\log^2L /\xi^2.$

Regarding the estimation of $\|\A^*(\be)\|$, the same calculations immediately give
\begin{equation*}
R : = \max_{1\leq l\leq L} \|e_l \bb_l\ba_l^*\|_{\psi_1} \leq \sqrt{\frac{K}{L}}\max_{1\leq l\leq L} ( |e_l| \|\ba_l\|)_{\psi_1}\leq \frac{C\sigma\sqrt{KN}}{L}
\end{equation*}
and
\begin{equation*}
\max \left\{ \left\| \E \left[ \A^*(\be) (\A^*(\be))^*\right] \right\|, \left\| \E \left[  (\A^*(\be))^*\A^*(\be)\right] \right\| \right\} \leq \frac{\sigma^2\max\{K, N\}}{L}.
\end{equation*}

Applying Bernstein inequality~\eqref{thm:bern} again, we get
\begin{equation*}
\|\A^*(\be)\| \leq C_0\sigma \max\Big\{ \sqrt{\frac{(\gamma + 1)\max\{K,N\} \log L}{L}}, \frac{(\gamma + 1)\sqrt{KN}\log^2 L }{L}  \Big\} \leq \xi
\end{equation*}
with probability at least $1 - L^{-\gamma}$ if $L \geq C_{\gamma} (\frac{\sigma^2}{\xi^2} + \frac{\sigma}{\xi})\max\{K, N\}\log^2L.$
\end{proof}

Lemma~\ref{lem:Ay-hx} lays the foundation for the initialization procedure, which says that with enough measurements, the initialization guess via spectral method can be quite close to the ground truth. 
Before moving to the proof of Theorem~\ref{thm:init}, we  introduce a property about the projection onto a closed convex set. 
\begin{lemma}[Theorem 2.8 in~\cite{RM11AP}]\label{lem:KMC}
Let $Q := \{ \bw\in\CC^K | \sqrt{L}\|\BB\bw\|_{\infty} \leq 2\sqrt{d}\mu \}$ 
be a closed nonempty convex set. There holds
\begin{equation*}
\Real( \lag \bz - \PP_Q(\bz) , \bw - \PP_Q(\bz) \rag ) \leq 0, \quad \forall \, \bw \in Q, \bz\in \CC^K 
\end{equation*}
where $\PP_{Q}(\bz)$ is the projection of $\bz$ onto $Q$.
\end{lemma}
This is a direct result from Theorem 2.8 in~\cite{RM11AP}, which is also called \textit{Kolmogorov criterion}. Now we present the proof of Theorem~\ref{thm:init}.

\begin{proof}{[\textbf{of Theorem~\ref{thm:init}}]}
Without loss of generality, we again set $d_0 = 1$ and by definition, all $\bh_0,$ $\bx_0$, $\hbh_0$ and $\hbx_0$ are of unit norm. 
Also we set $\xi = \frac{\eps}{10\sqrt{2}}.$
By applying the triangle inequality to~\eqref{eq:Ay-hx}, it is easy to see that
\begin{equation}\label{eq:hd}
1 - \xi \leq d \leq 1 + \xi, \quad |d - 1| \leq \xi \leq \frac{\eps}{10\sqrt{2}} < \frac{1}{10}, 
\end{equation}
which gives $\frac{9}{10}d_0 \leq d \leq \frac{11}{10}d_0.$
It is easier to get an upper bound for $\|\bv_0\|$ here, i.e., 
\begin{equation*}
\| \bv_0 \| = \sqrt{d} \|\hbx_0\| =  \sqrt{d} \leq \sqrt{1 + \xi} \leq \frac{2}{\sqrt{3}},
\end{equation*}
which implies $\bv_0 \in \frac{1}{\sqrt{3}}\Kd.$ The estimation of $\bu_0$ involves Lemma~\ref{lem:KMC}.
In our case, $\bu_0$ is the minimizer to the function $f(\bz) = \frac{1}{2} \| \bz - \sqrt{d} \hbh_0 \|^2$ over $Q = \{ \bz | \sqrt{L}\|\BB\bz\|_{\infty} \leq 2\sqrt{d}\mu\}.$ Therefore, $\bu_0$ is actually the projection of $\sqrt{d} \hbh_0$ onto $Q$. Note that $\bu_0\in Q$ implies $\sqrt{L}\|\BB\bu_0\|_{\infty} \leq 2\sqrt{d}\mu\leq \frac{4\mu}{\sqrt{3}}$ and hence $\bu_0\in \frac{1}{\sqrt{3}} \Kmu.$ Moreover, $\bu_0$ yields

\begin{eqnarray}
\|\sqrt{d}\hbh_0 - \bw\|^2 
& = & \| \sqrt{d}\hbh_0 - \bu_0\|^2 + 2\Real(\lag \sqrt{d}\hbh - \bu_0, \bu_0 - \bw\rag) + \|\bu_0 - \bw \|^2 \nonumber \\
& \geq & \| \sqrt{d}\hbh_0 - \bu_0\|^2 + \|\bu_0 - \bw \|^2 \label{eq:KMC}
\end{eqnarray}
for all $\bw\in Q$ because the cross term is nonnegative due to Lemma~\ref{lem:KMC}. 
Let $\bw = \bzero\in Q$ and we get
\begin{equation*}
\|\bu_0\|^2 \leq d \leq \frac{4}{3}.
\end{equation*}

So far, we have already shown that $(\bu_0, \bv_0) \in \frac{1}{\sqrt{3}}\Kd$ and $\bu_0 \in \frac{1}{\sqrt{3}}\Kmu$. Now we will show that $\|\bu_0\bv_0^* - \bh_0\bx_0^*\|_F \leq 4\xi.$

First note that $\sigma_i(\A^*(\by)) \leq \xi$ for all $i\geq 2$, which follows from Weyl's inequality~\cite{Stewart90} for singular values where $\sigma_i(\A^*(\by))$ denotes the $i$-th largest singular value of $\A^*(\by)$. Hence there holds
\begin{equation}\label{eq:dhx-hx}
\| d \hbh_0\hbx_0^* - \bh_0\bx_0^* \| \leq \|\A^*(\by) - d \hbh_0\hbx_0^* \| + \|\A^*(\by) - \bh_0\bx_0^* \| \leq 2\xi.
\end{equation}

On the other hand,
\begin{eqnarray*}
\| (\I - \bh_0\bh_0^*)\hbh_0 \| 
& = &  \| (\I - \bh_0\bh_0^*)
\hbh_0\hbx_0^*\hbx_0\hbh_0^* \| \\
& = &  \| (\I - \bh_0\bh_0^*)
( \A^*(\by) - d \hbh_0\hbx_0^* + \hbh_0\hbx_0^* - \bh_0\bx_0^*  ) \hbx_0\hbh_0^* \| \\
& = &  \| (\I - \bh_0\bh_0^*)
( \A^*(\by)  - \bh_0\bx_0^*  ) \hbx_0\hbh_0^* \| \leq \xi
\end{eqnarray*} 
where the second equation follows from $ (\I - \bh_0\bh_0^*) \bh_0\bx_0^* = \bzero$ and $(\A^*(\by) - d \hbh_0\hbx_0^*)\hbx_0\hbh_0^* = \bzero$.
Therefore, we have
\begin{equation}\label{eq:h0-h}
\|  \hbh_0 -  \bh_0^*\hbh_0 \bh_0  \| \leq \xi, 
\quad \|\sqrt{d} \hbh_0 - \alpha_0 \bh_0 \| \leq \sqrt{d}\xi,
\end{equation}
where $\alpha_0 = \sqrt{d}\bh_0^*\hbh_0$. 
If we  substitute $\bw$ by $\alpha_0 \bh_0\in Q$ into~\eqref{eq:KMC}, 
\begin{equation}\label{eq:h0-h-2}
\|\sqrt{d}\hbh_0 - \alpha_0 \bh_0\| \geq \| \bu_0 - \alpha_0 \bh_0\|.
\end{equation}
where $\alpha_0 \bh_0\in Q$ follows from $\sqrt{L} |\alpha_0|\|\BB\bh_0\|_{\infty} \leq |\alpha_0| \mu_h \leq \sqrt{d} \mu_h \leq \sqrt{d} \mu < \sqrt{2d}\mu$.
Combining~\eqref{eq:h0-h} and~\eqref{eq:h0-h-2} leads to $\|\bu_0 - \alpha_0\bh_0\| \leq \sqrt{d}\xi.$ Now we are ready to estimate $\|\bu_0\bv_0^* - \bh_0\bx_0^* \|_F$ as follows, 
\begin{eqnarray*}
\|\bu_0\bv_0^* - \bh_0\bx_0^* \|_F & \leq & 
\|\bu_0\bv_0^* - \alpha_0\bh_0\bv_0^* \|_F + \|\alpha_0\bh_0\bv_0^* - \bh_0\bx_0^* \|_F \\
& \leq & \|\bu_0 - \alpha_0\bh_0\| \|\bv_0\| + \| d \bh_0 \bh^*_0 \hbh_0  \hbx_0^* - \bh_0\bx_0^* \|_F \\
& \leq & \xi \sqrt{d} \|\bv_0\|  + \|d \hbh_0\hbx_0^* - \bh_0\bx_0^*\|_F \\
& \leq & \xi d  + 2 \sqrt{2}\xi \leq \xi (1 + \xi) + 2\sqrt{2}\xi \\
& \leq &  4\xi \leq \frac{2}{5}\eps,
\end{eqnarray*}
where $\|\bv_0\| = \sqrt{d}$, $\bv_0 = \sqrt{d}\hbx_0$ and $\|d \hbh_0\hbx_0^* - \bh_0\bx_0^*\|_F \leq \sqrt{2}\|d \hbh_0\hbx_0^* - \bh_0\bx_0^*\| \leq 2\sqrt{2} \xi$ follows from~\eqref{eq:dhx-hx}.
\end{proof}

\section{Appendix}

\subsection{Descent Lemma}
\begin{lemma}\label{lem:DSL}
If $f(\bz, \bar{\bz})$ is a continuously differentiable real-valued function with two complex variables $\bz$ and $\bar{\bz}$, (for simplicity, we just denote $f(\bz, \bar{\bz})$ by $f(\bz)$ and keep in the mind that $f(\bz)$ only assumes real values) for $\bz := (\bh, \bx) \in \Keps\cap\KF$. 
Suppose that there exists a constant $C_L$ such that
\begin{equation*}
\|\nabla f(\bz + t \Delta \bz) - \nabla f(\bz)\| \leq C_L t\|\Delta\bz\|, \quad \forall 0\leq t\leq 1,
\end{equation*}
for all $\bz\in \Keps\cap\KF$ and $\Delta \bz$ such that  $\bz + t\Delta \bz \in \Keps\cap\KF$ and $0\leq t\leq 1$. Then
\begin{equation*}
f(\bz + \Delta  \bz) \leq f(\bz) + 2\Real( (\Delta \bz)^T \overline{\nabla} f(\bz)) + C_L\|\Delta \bz\|^2
\end{equation*}
where $\overline{\nabla} f(\bz) := \frac{\pa f(\bz, \bar{\bz})}{\pa \bz}$ is the complex conjugate of $\nabla f(\bz) = \frac{\pa f(\bz, \bar{\bz})}{\partial \bar{\bz}}$.
\end{lemma}

\begin{proof}
The proof simply follows from proof of descent lemma (Proposition A.24 in~\cite{B99}). However it is slightly different since we are dealing with complex variables. Denote $g(t) := f(\bz + t\Delta \bz).$ Since $f(\bz, \bar{\bz})$ is a continuously differentiable function, we apply the chain rule 
\begin{equation}
\label{eq:chain}
\frac{d g(t)}{dt} = (\Delta\bz)^T\frac{\pa f}{\pa \bz} (\bz + t\Delta \bz) + (\Delta \bar{\bz})^T \frac{\pa f}{\pa \bar{\bz}} (\bz + t\Delta \bz)  = 2\Real ( (\Delta \bz)^T \overline{\nabla} f (\bz + t\Delta\bz) ).
\end{equation}
Then by the Fundamental Theorem of Calculus, 
\begin{eqnarray*}
f(\bz + t\Delta \bz) - f(\bz) & = & \int_0^1 \frac{dg(t)}{dt} dt = 2\int_0^1 \Real ( (\Delta \bz)^T \overline{\nabla} f (\bz + t\Delta \bz) ) dt \\
& \leq & 2\Real ((\Delta\bz)^T \overline{\nabla} f(\bz_0)) + 2\int_0^1 \Real((\Delta\bz)^T(\overline{\nabla} f(\bz + t\Delta \bz) - \overline{\nabla} f(\bz)))  dt \\
& \leq & 2\Real((\Delta\bz)^T \overline{\nabla} f(\bz)) + 2 \|\Delta\bz\| \int_0^1 \| \nabla f(\bz + t\Delta \bz) - \nabla f(\bz)) \| dt \\
& \leq & 2\Real((\Delta\bz)^T \overline{\nabla} f(\bz)) + C_L\|\Delta \bz\|^2.
\end{eqnarray*}

\end{proof}

\subsection{Some useful facts}
The key concentration inequality we use throughout our paper comes from Proposition 2 in~\cite{KolVal11,VK11}.

\begin{theorem}\label{BernGaussian}
Consider a finite sequence of $\CZ_l$ of independent centered random matrices with dimension $M_1\times M_2$. Assume that  $\|\CZ_l\|_{\psi_1} \leq R$ where the norm $\|\cdot\|_{\psi_1}$ of a matrix is defined as
\begin{equation}\label{def:psi-1}
\|\BZ\|_{\psi_1} := \inf_{u \geq 0} \{ \E[ \exp(\|\BZ\|/u)] \leq 2 \}.
\end{equation}
and introduce the random matrix 
\begin{equation}\label{S}
\BS = \sum_{l=1}^L \CZ_l.
\end{equation}
Compute the variance parameter
\begin{equation}\label{sigmasq}
\sigma_0^2 := \max\{ \|\E(\BS\BS^*)\|, \|\E(\BS^*\BS)\|\} = \max\Big\{ \| \sum_{l=1}^L \E(\CZ_l\CZ_l^*)\|, \| \sum_{l=1}^L \E(\CZ_l^* \CZ_l)\| \Big\},
\end{equation}
then for all $t \geq 0$, we have the tail bound on the operator norm of $\BS$, 

\begin{equation}\label{thm:bern}
\|\BS\| \leq C_0 \max\{ \sigma_0 \sqrt{t + \log(M_1 + M_2)}, R\log\left( \frac{\sqrt{L}R}{\sigma_0}\right)(t + \log(M_1 + M_2)) \}
\end{equation}
with probability at least $1 - e^{-t}$ where $C_0$ is an absolute constant. \end{theorem}

For convenience we also collect some results used throughout the proofs. 
\begin{lemma} \label{lem:6JB}
Let $z$ be a random variable which obeys $\Pr\{ |z| > u \} \leq a e^{-b u }$, then
\begin{equation*}
\|z\|_{\psi_1} \leq (1 + a)/b.
\end{equation*}
which is proven in Lemma 2.2.1 in~\cite{VW96}. Moreover, it is easy to verify that for a scalar $\lambda\in \CC$
\begin{equation*}
\|\lambda z\|_{\psi_1}  = |\lambda| \|z\|_{\psi_1}.
\end{equation*}

\end{lemma}

\begin{lemma}[ Lemma 10-13 in~\cite{RR12}, Lemma 12.4 in~\cite{LS15Blind}] \label{lem:multiple1}
Let $\bu\in\CC^n \sim \mathcal{N}(0, \frac{1}{2}\I_n) + \mi \mathcal{N}(0, \frac{1}{2}\I_n) $, then $\|\bu\|^2 \sim \frac{1}{2}\chi^2_{2n}$ and
\begin{equation}
\label{lem:8JB}
\| \|\bu\|^2 \|_{\psi_1} = \| \lag\bu, \bu\rag \|_{\psi_1}  \leq C n
\end{equation}
and
\begin{equation}
\label{lem:9JB}
\E\ls (\bu\bu^* - \I_n)^2 \rs = n\I_n.
\end{equation}
Let $\bq\in\CC^n$ be any deterministic vector, then the following properties hold

\begin{equation}
\label{lem:11JB}
\| (\bu\bu^* - \I)\bq\|_{\psi_1} \leq C\sqrt{n}\|\bq\|,
\end{equation}
\begin{equation}
\label{lem:12JB}
\E\ls (\bu\bu^* - \I)\bq\bq^* (\bu\bu^* - \I)\rs = \|\bq\|^2 \I_n.
\end{equation}
Let $\bv\sim \mathcal{N}(0, \frac{1}{2}\I_m) + \mi \mathcal{N}(0, \frac{1}{2}\I_m) $  be a complex Gaussian random vector in $\CC^m$, independent of $\bu$, then 
\begin{equation}\label{lemma:psi}
\left\| \|\bu\| \cdot \|\bv\|\right\|_{\psi_1} \leq C\sqrt{mn}.
\end{equation}

\end{lemma}

\section*{Acknowledgement}

S.~Ling, T.~Strohmer, and K.~Wei acknowledge support from the NSF via grant DTRA-DMS 1322393.


\begin{thebibliography}{10}

\bibitem{RR12}
A.~Ahmed, B.~Recht, and J.~Romberg.
\newblock Blind deconvolution using convex programming.
\newblock {\em Information Theory, IEEE Transactions on}, 60(3):1711--1732,
  2014.

\bibitem{B99}
D.~P. Bertsekas.
\newblock {em Nonlinear programming}.
\newblock Athena Scientific.
\newblock 1999.

\bibitem{CLM15}
T.~T. Cai, X.~Li, and Z.~Ma.
\newblock Optimal rates of convergence for noisy sparse phase retrieval via
  thresholded {W}irtinger flow.
\newblock {\em arXiv preprint arXiv:1506.03382}, 2015.

\bibitem{CJ16}
V.~Cambareri and L.~Jacques.
\newblock A non-convex blind calibration method for randomised sensing
  strategies.
\newblock {\em arXiv preprint arXiv:1605.02615}, 2016.

\bibitem{campisi2007blind}
P.~Campisi and K.~Egiazarian.
\newblock {\em Blind image deconvolution: theory and applications}.
\newblock CRC press, 2007.

\bibitem{CR08}
E.~Cand\`es and B.~Recht.
\newblock {Exact matrix completion via convex optimization}.
\newblock {\em Foundations of Computational Mathematics}, 9(6):717--772, 2009.

\bibitem{CLS14}
E.~J. Cand\`{e}s, X.~Li, and M.~Soltanolkotabi.
\newblock Phase retrieval via {W}irtinger flow: Theory and algorithms.
\newblock {\em Information Theory, IEEE Transactions on}, 61(4):1985--2007,
  2015.

\bibitem{CSV11}
E.~J. Cand\`{e}s, T.~Strohmer, and V.~Voroninski.
\newblock Phaselift: Exact and stable signal recovery from magnitude
  measurements via convex programming.
\newblock {\em Communications on Pure and Applied Mathematics},
  66(8):1241--1274, 2013.

\bibitem{chan1998total}
T.~F. Chan and C.-K. Wong.
\newblock Total variation blind deconvolution.
\newblock {\em Image Processing, IEEE Transactions on}, 7(3):370--375, 1998.

\bibitem{CVR14}
S.~Chaudhuri, R.~Velmurugan, and R.~Rameshan.
\newblock {\em Blind Image Deconvolution: Methods and Convergence}.
\newblock Springer, 2014.

\bibitem{CC15}
Y.~Chen and E.~Candes.
\newblock Solving random quadratic systems of equations is nearly as easy as
  solving linear systems.
\newblock In {\em Advances in Neural Information Processing Systems}, pages
  739--747, 2015.

\bibitem{CM15}
Y.~Chen and M.~J. Wainwright.
\newblock Fast low-rank estimation by projected gradient descent: General
  statistical and algorithmic guarantees.
\newblock {\em arXiv preprint arXiv:1509.03025}, 2015.

\bibitem{RM11AP}
R.~Escalante and M.~Raydan.
\newblock {\em Alternating projection methods}, volume~8.
\newblock SIAM, 2011.

\bibitem{jefferies1993restoration}
S.~M. Jefferies and J.~C. Christou.
\newblock Restoration of astronomical images by iterative blind deconvolution.
\newblock {\em The Astrophysical Journal}, 415:862, 1993.

\bibitem{KK16}
M.~Kech and F.~Krahmer.
\newblock Optimal injectivity conditions for bilinear inverse problems with
  applications to identifiability of deconvolution problems.
\newblock {\em arXiv preprint arXiv:1603.07316}, 2016.

\bibitem{KMO09noise}
R.~Keshavan, A.~Montanari, and S.~Oh.
\newblock Matrix completion from noisy entries.
\newblock In {\em Advances in Neural Information Processing Systems}, pages
  952--960, 2009.

\bibitem{KMO09IT}
R.~H. Keshavan, A.~Montanari, and S.~Oh.
\newblock Matrix completion from a few entries.
\newblock {\em Information Theory, IEEE Transactions on}, 56(6):2980--2998,
  2010.

\bibitem{KolVal11}
V.~Koltchinskii et~al.
\newblock Von {N}eumann entropy penalization and low-rank matrix estimation.
\newblock {\em The Annals of Statistics}, 39(6):2936--2973, 2011.

\bibitem{VK11}
V.~Koltchinskii, K.~Lounici, A.~B. Tsybakov, et~al.
\newblock Nuclear-norm penalization and optimal rates for noisy low-rank matrix
  completion.
\newblock {\em The Annals of Statistics}, 39(5):2302--2329, 2011.

\bibitem{LRST10}
J.~D. Lee, B.~Recht, N.~Srebro, J.~Tropp, and R.~R. Salakhutdinov.
\newblock Practical large-scale optimization for max-norm regularization.
\newblock In {\em Advances in Neural Information Processing Systems}, pages
  1297--1305, 2010.

\bibitem{LSJR16}
J.~D. Lee, M.~Simchowitz, M.~I. Jordan, and B.~Recht.
\newblock Gradient descent converges to minimizers.
\newblock {\em University of California, Berkeley}, 1050:16, 2016.

\bibitem{LLJB15}
K.~Lee, Y.~Li, M.~Junge, and Y.~Bresler.
\newblock Blind recovery of sparse signals from subsampled convolution.
\newblock {\em arXiv preprint arXiv:1511.06149}, 2015.

\bibitem{LWD11}
A.~Levin, Y.~Weiss, F.~Durand, and W.~Freeman.
\newblock Understanding blind deconvolution algorithms.
\newblock {\em Pattern Analysis and Machine Intelligence, IEEE Transactions
  on}, 33(12):2354--2367, 2011.

\bibitem{LLB15ID}
Y.~Li, K.~Lee, and Y.~Bresler.
\newblock Identifiability in blind deconvolution with subspace or sparsity
  constraints.
\newblock {\em arXiv preprint arXiv:1505.03399}, 2015.

\bibitem{LLB15BIP}
Y.~Li, K.~Lee, and Y.~Bresler.
\newblock A unified framework for identifiability analysis in bilinear inverse
  problems with applications to subspace and sparsity models.
\newblock {\em arXiv preprint arXiv:1501.06120}, 2015.

\bibitem{LS15Blind}
S.~Ling and T.~Strohmer.
\newblock Blind deconvolution meets blind demixing: Algorithms and performance
  bounds.
\newblock {\em arXiv preprint arXiv:1512.07730}, 2015.

\bibitem{LS15}
S.~Ling and T.~Strohmer.
\newblock Self-calibration and biconvex compressive sensing.
\newblock {\em Inverse Problems}, 31(11):115002, 2015.

\bibitem{Recht11MC}
B.~Recht.
\newblock A simpler approach to matrix completion.
\newblock {\em The Journal of Machine Learning Research}, 12:3413--3430, 2011.

\bibitem{Stewart90}
G.~W. Stewart.
\newblock Perturbation theory for the singular value decomposition.
\newblock {\em technical report CS-TR 2539, university of Maryland}, September
  1990.

\bibitem{Str00}
T.~Strohmer.
\newblock Four short stories about {T}oeplitz matrix calculations.
\newblock {\em Linear Algebra Appl.}, 343/344:321--344, 2002.
\newblock Special issue on structured and infinite systems of linear equations.

\bibitem{Sun16}
J.~Sun.
\newblock {\em When Are Nonconvex Optimization Problems Not Scary?}
\newblock PhD thesis, Columbia University, 2016.

\bibitem{SQW16}
J.~Sun, Q.~Qu, and J.~Wright.
\newblock A geometric analysis of phase retrieval.
\newblock {\em arXiv preprint arXiv:1602.06664}, 2016.

\bibitem{Sun15}
R.~Sun.
\newblock {\em Matrix Completion via Nonconvex Factorization: Algorithms and
  Theory}.
\newblock PhD thesis, University of Minnesota, 2015.

\bibitem{SL15}
R.~Sun and Z.-Q. Luo.
\newblock Guaranteed matrix completion via nonconvex factorization.
\newblock In {\em Foundations of Computer Science (FOCS), 2015 IEEE 56th Annual
  Symposium on}, pages 270--289. IEEE, 2015.

\bibitem{Tong95}
L.~Tong, G.~Xu, B.~Hassibi, and T.~Kailath.
\newblock Blind identification and equalization based on second-order
  statistics: {A} frequency domain approach.
\newblock {\em IEEE Trans. Inf. Theory}, 41(1):329--334, Jan. 1995.

\bibitem{TV05}
D.~Tse and P.~Viswanath.
\newblock {\em Fundamentals of Wireless Communication}.
\newblock Cambridge University Press, 2005.

\bibitem{TBSR15}
S.~Tu, R.~Boczar, M.~Soltanolkotabi, and B.~Recht.
\newblock Low-rank solutions of linear matrix equations via {P}rocrustes flow.
\newblock {\em arXiv preprint arXiv:1507.03566}, 2015.

\bibitem{VW96}
A.~van~der Vaart and J.~Wellner.
\newblock {\em Weak convergence and empirical processes}.
\newblock Springer Series in Statistics. Springer-Verlag, New York, 1996.
\newblock With applications to statistics.

\bibitem{Ver10}
R.~Vershynin.
\newblock Introduction to the non-asymptotic analysis of random matrices.
\newblock In Y.~C. Eldar and G.~Kutyniok, editors, {\em Compressed Sensing:
  Theory and Applications}, chapter~5. Cambridge University Press, 2012.

\bibitem{wang1998blind}
X.~Wang and H.~V. Poor.
\newblock Blind equalization and multiuser detection in dispersive cdma
  channels.
\newblock {\em Communications, IEEE Transactions on}, 46(1):91--103, 1998.

\bibitem{WCCL15}
K.~Wei, J.-F. Cai, T.~F. Chan, and S.~Leung.
\newblock Guarantees of riemannian optimization for low rank matrix recovery.
\newblock {\em arXiv preprint arXiv:1511.01562}, 2015.

\bibitem{WBSJ14}
G.~Wunder, H.~Boche, T.~Strohmer, and P.~Jung.
\newblock Sparse signal processing concepts for efficient {5G}~system design.
\newblock {\em IEEE Access}, 3:195---208, 2015.

\bibitem{ZL16}
Q.~Zheng and J.~Lafferty.
\newblock Convergence analysis for rectangular matrix completion using
  {Burer-Monteiro} factorization and gradient descent.
\newblock {\em Preprint [arxiv:1605.07051]}, 2016.

\end{thebibliography}
\end{document}